\documentclass[showpacs,amsmath,amssymb,twocolumn,pra,superscriptaddress]{revtex4-2}
\usepackage{amsmath}
\usepackage{amsthm} %May cause conflict later
\usepackage{mathtools}
\usepackage{bm}
\usepackage{graphicx}
\usepackage{esvect}
\usepackage{fancyhdr}
\usepackage{graphicx}
\usepackage{graphics}
\usepackage{amssymb}
\usepackage{enumitem}
\usepackage{braket}
\usepackage{xcolor}
\usepackage[colorlinks = true]{hyperref}
\usepackage{subfigure}
\usepackage{tikz}
\usetikzlibrary{decorations.pathreplacing} %For curly braces

\definecolor{egg}{rgb}{.98,.97,.92}
\definecolor{astroorange}{rgb}{1,.93,.79}
\definecolor{darkorange}{rgb}{1,.89,.6}
\definecolor{dullblue}{rgb}{.29,.47,.77}
\definecolor{grayblue}{rgb}{.98,.98,.98}
\definecolor{fadedblue}{rgb}{.78,.86,.92}
\definecolor{tiffanyblue}{rgb}{.96,1,1}
\definecolor{grayish}{rgb}{.93,.93,.97}
\definecolor{charcoal}{rgb}{.247,.259,.27}
\definecolor{evergreen}{rgb}{.7725,.858,.7647}
\definecolor{dullred}{rgb}{.929,.498,.598}
\definecolor{lavender}{rgb}{.8,.741,.85}

\newcommand{\tr}[1]{\mathrm{Tr}\left\{#1\right\}}
\newcommand{\ptr}[2]{\mathrm{Tr}_{#1}\left\{#2\right\}}

\newcommand{\av}[1]{\underset{\tiny{#1}}{\mathbb{E}}}
\newcommand{\norm}[1]{\left|\left| #1 \right|\right|}
\newcommand{\abs}[1]{\left| #1 \right|}

\newcommand{\OTOC}{\mathrm{OTOC}}

\newcommand{\xpos}{0} %Tikz node positions
\newcommand{\ypos}{0} %Tikz node positions
\newcommand{\xrel}{1} %Tikz node relative x distance to center
\newcommand{\yrel}{.5} %Tikz node relative y distance to center
\newcommand{\xglobalshift}{0} %Global Shift Parameter
\newcommand{\yglobalshift}{0} %Global shift parameter
\newcommand{\height}{2em} %node height
\newcommand{\width}{2em} %node width
\newcommand{\name}{} %Tikz node name
\newcommand{\nodenum}{0} %Tikz node name
\newcommand{\heightsingle}{2em} %node height
\newcommand{\heightdouble}{4.6em} %node height
\newcommand{\widthsingle}{2em} %node height
\newcommand{\rowspace}{2*\yrel} %Tikz node row space
\newtheorem{theorem}{Theorem}
\newtheorem{lemma}{Lemma}

\newtheorem{definition}{Definition}
\newtheorem{example}{Example}

\begin{document}
%\title{Operational resource theory of scrambling}
\title{Resource theory of quantum scrambling}
\author{Roy J. Garcia}
\email{roygarcia@g.harvard.edu}
\affiliation{Harvard University, Cambridge, Massachusetts 02138, USA}

\author{Kaifeng Bu}
\email{kfbu@g.harvard.edu}
\affiliation{Harvard University, Cambridge, Massachusetts 02138, USA}

\author{Arthur Jaffe}
\email{jaffe@g.harvard.edu}
\affiliation{Harvard University, Cambridge, Massachusetts 02138, USA}

\date{\today}

\begin{abstract}
%Quantum scrambling refers to the spread of local quantum information into the many degrees of freedom of a quantum system. In this work, we introduce a resource theory of scrambling which incorporates two mechanisms, ``entanglement scrambling'' and ``magic scrambling". We introduce two resource monotones called the Pauli growth and the OTOC (out-of-time-ordered correlator) magic for these two mechanisms, respectively. Moreover, we show that OTOC fluctuations are bounded by the OTOC magic. This proves that small OTOC fluctuations are an indication of magic in Google's recent experiment \mbox{(Science 374, 1479 (2021))}. We also show that both  resource monotones can be used to bound the decoding fidelity in the Hayden-Preskill protocol. These applications provide an operational interpretation of the resource monotones defined in this work.

Quantum scrambling refers to the spread of local quantum information into the many degrees of freedom of a quantum system. 
In this work, we introduce a resource theory of scrambling which incorporates two mechanisms, ``entanglement scrambling'' and ``magic scrambling". We introduce two resource monotones called the Pauli growth and the OTOC (out-of-time-ordered correlator) magic for these two mechanisms, respectively. We use our resource theory to explain recent experimental observations of magic. We also show that both resource monotones can be used to bound the decoding fidelity in Yoshida's black hole decoding protocol. These applications provide an operational interpretation of the resource monotones defined in this work.

\end{abstract}

\maketitle

\section{Introduction}
Quantum scrambling describes the spread of local information in a quantum system. This field has flourished by connecting diverse areas, including  quantum many-body physics~\cite{PhysRevLett.70.3339, Kitaev2015, Nandkishore_2015}, black hole physics~\cite{Hayden_2007, Sekino_2008, Lashkari_2013, Shenker_2014, Maldacena_2016}, and quantum information~\cite{Hayden_2007}. It has become  a prevalent ingredient in many  information processing problems found in: quantum machine learning \cite{Holmes_2021, Garcia2022Barren,  PhysRevLett.124.200504, PhysRevResearch.3.L032057, Garcia_2022}, shadow tomography with classical shadows \cite{Huang_2020, PhysRevResearch.4.013054, Hu2021, Bu2022Classical, PhysRevResearch.3.033155,McGinley_2022}, quantum error correction \cite{Brown2012,Choi_2020}, encryption \cite{Ruck2020}, and emergent quantum state designs \cite{PhysRevLett.128.060601}. For instance, scrambling dynamics is  used in~\cite{PhysRevLett.128.060601} to generate an ensemble of quantum states that  is indistinguishable from the set of all uniformly-random states. 

To quantify scrambling in a quantum system, several measures have been proposed, such as
the average Pauli weight~\cite{Chen2018Strongly, PhysRevE.99.052212, PhysRevX.8.031057}, the out-of-time-ordered correlator (OTOC)~\cite{Aleiner_2016,Larkin1969, Nahum_2017, PhysRevX.8.031057, PhysRevLett.115.131603,Hosur_2016,PhysRevD.96.065005, PRXQuantum.2.020339,PhysRevX.8.021014,von_Keyserlingk_2018}, the operator entanglement entropy~\cite{PhysRevA.63.040304, Band2005, Zhou_2017}, and the tripartite mutual information~\cite{Hosur_2016,Roberts_2017}. For example, Pan et al. measured the tripartite mutual information as a signature of scrambling on a superconducting quantum processor~\cite{PhysRevLett.128.160502}. The OTOC in particular has been used to characterize many-body localization~\cite{Nandkishore_2015,PhysRevB.95.060201, Huang_2016,Chen2016, Chen_2016,Fan_2017,PhysRevB.95.054201,Nandkishore_2015} and fast scramblers, including black holes and the SYK model  \cite{Sekino_2008, Maldacena_2016, PhysRevLett.70.3339, Kitaev2015,  Belyansky_2020}.  Moreover,  two mechanisms of scrambling were investigated by measuring the average OTOC and OTOC fluctuations using Google's Sycamore quantum processor~\cite{Mi_2021}.

Exploiting advantages provided by
quantum phenomena is one of the key problems in information processing tasks. In recent years,
one theoretical framework,  called quantum resource theory \cite{HORODECKI_2012,ChitambarRMP19},  has been developed to quantify these advantages.
In a quantum resource theory, one identifies a set of free quantum states (channels) that are easy to prepare (implement). All other quantum states (channels) are considered resources and are often useful in accomplishing a particular task. A resource monotone quantifies the amount of resource in a state (channel). Using this formalism, it
was shown that many features of quantum resources are very general and can be characterized in a unified manner \cite{ChitambarRMP19,Liu2019}. Examples of resources include entanglement~\cite{Chitambar_2019}, magic~\mbox{\cite{Veitch_2014, Howard_2017,Wang_2019}}, quantum thermodynamics~\cite{Brandao2013,Brandao_secondlaws2015,RennerPRX21,Renner2022},
coherence~\cite{aberg2006quantifying, baumgratz2014quantifying,winter2016operational,bu2017maximum,Streltsov2017colloquium}, and uncomplexity~\cite{Halpern2021}, among others. 

Resource theories have been widely used to quantify advantages in operational tasks~\cite{Takagi_2019}. For instance, quantum entanglement
is an essential resource for quantum teleportation~\cite{PhysRevLett.70.1895}.  Magic,  which characterizes how far away a quantum state (gate) is from the set of stabilizer states (Clifford gates), has been used in quantum computation to establish bounds on classical simulation times~\cite{Bravyi_2016, Bravyi_2019, Howard_2017, Seddon_2021, Seddon_2019, Bu_2019, Bu_2022}. 
Despite its many applications, a resource theory of scrambling has been  lacking. Recently, Yoshida conjectured that one may exist~\cite{Yoshida2021Recovery}.

%In this work, we define a resource theory of scrambling incorporating two mechanisms, ``entanglement scrambling'' and ``magic scrambling". In entanglement scrambling, resourceful unitaries increase the support of a local Pauli operator. In magic scrambling, resourceful unitaries map a Pauli operator to a sum of multiple Pauli operators. We define resource monotones called the Pauli growth and the OTOC magic for each mechanism, respectively.  We show that OTOC fluctuations bound the OTOC magic, which provides a theoretical proof of Google's experimental results~\cite{Mi_2021}. We also use the OTOC magic and the Pauli growth to bound the decoding fidelity in the Hayden-Preskill protocol \cite{Hayden_2007}, in which scrambling is used to recover a quantum state thrown into a  black hole. 

In this work, we define a resource theory of scrambling incorporating two mechanisms, ``entanglement scrambling'' and ``magic scrambling".
In entanglement scrambling, resourceful unitaries increase the support of a local Pauli operator. In magic scrambling, resourceful unitaries map a Pauli operator to a sum of multiple Pauli operators. We define resource monotones called the Pauli growth and the OTOC magic for each mechanism, respectively. We show that OTOC fluctuations bound the OTOC magic, which provides a theoretical proof of Google's experimental results~\cite{Mi_2021}. We also use the OTOC magic and the Pauli growth to bound the decoding fidelity in Yoshida's decoding protocol~\cite{YoshidaEfficient} for the Hayden-Preskill thought experiment~\cite{Hayden_2007}, in which scrambling is used to recover a quantum state thrown into a  black hole.

\section{Main results}
\subsection{Preliminaries}
Given an $n$-qudit system, the generalized $n$-qudit Pauli group is defined as ${\mathcal{P}_d^{\otimes n}={\{P_{\vec{a}}: P_{\vec{a}}=\otimes_{i=1}^n P_{a_i}\}_{\vec{a}\in \mathcal{V}_d^n}}}$, where $d$ is the local  dimension, $P_{a_i}=X^{s_i}Z^{t_i}$, ${a_i=(s_i,t_i)\in \mathcal{V}_d= \mathbb{Z}_d\otimes \mathbb{Z}_d}$ and $\vec{a}=(a_1,\ldots,a_n)$. The generalized Pauli $X$ operator is defined by ${X\ket{j}=\ket{j+1 \ (\mathrm{mod} \ d)}}$ and the generalized Pauli  $Z$ operator is defined by $Z\ket{j}=e^{2ij \pi/d}\ket{j}$. The inner product between the two $n$-qudit operators $O_1$ and $O_2$ is defined as $\langle O_1,O_2\rangle\equiv \frac{1}{d^n}\tr{O_1^\dagger O_2}$. We define the induced norm as $\norm{\cdot}_2=\sqrt{\langle \cdot, \cdot \rangle}$.

Let $O$ be an $n$-qudit operator with a norm of ${\norm{O}_2=1}$. The operator $O$ can be expanded in the Pauli basis, ${O=\sum_{\vec{a}\in\mathcal{V}^n_d}c_{\vec{a}}P_{\vec{a}}}$, where the expansion coefficients satisfy $\sum_{\vec{a}\in\mathcal{V}^n_d}\abs{c_{\vec{a}}}^2=1$. Due to this normalization condition, we define ${P_O[\vec{a}]\equiv \abs{c_{\vec{a}}}^2=\frac{1}{d^{2n}}\abs{\tr{OP_{\vec{a}}}}^2}$, which implies a probability distribution 
over $\mathcal{P}_d^{\otimes n}$.   The average Pauli weight  of $O$, also called the  influence \cite{Montanaro2008},  is 
\begin{equation}\label{eq:pw}
	W(O)=\sum_{\vec{a}\in \mathcal{V}_d^n}\abs{\vec{a}}P_{O}[\vec{a}],
\end{equation}
where $\abs{\vec{a}}$, the number of $a_j$ in $\vec{a}$ such that ${a_j\neq (0,0)}$, is the Pauli weight of the Pauli operator $P_{\vec{a}}$ \mbox{(e.g. $I\otimes X \otimes I \otimes Z$} has a Pauli weight of $2$).

\subsection{Entanglement scrambling}
We first introduce the framework for entanglement scrambling, in which free unitaries (referred to as \textit{non-entangling unitaries}) are defined as the unitaries which map any weight-1 Pauli operator by conjugation to an operator with an average Pauli weight of 1. Non-entangling unitaries are generated by swap gates and  single-qudit unitaries, as shown in~\cite{Bu2022}. The name of this mechanism is motivated by the fact that non-entangling unitaries do not increase the average of an entanglement measure over all bipartitions (see \footnote{For example, the average R\'enyi-2 entanglement entropy, $\overline{S}^{(2)}(\rho)\equiv \frac{1}{2^{n}}\sum_{A\in [n]}-\mathrm{log}\tr{\rho_A^2}$, where $\rho_A$ is the reduced state of an $n$-qudit state $\rho$ on a subsystem $A$, satisfies $\overline{S}^{(2)}(U\rho U^\dagger)=\overline{S}^{(2)}(\rho)$, where $U$ is a non-entangling unitary \cite{Bu2022}.} for an example).

We define a resource monotone called the Pauli growth based on the average Pauli weight in Eq.~\eqref{eq:pw}.

\begin{definition}\label{Lemma:PauliMonotone}
The Pauli growth of a unitary $U$ is
\begin{equation}
	G(U)\equiv \max_{\substack{O:  \norm{O}_2=1, W(O)=1,\\\tr{O}=0}}\left[W(U^\dagger O U)-1\right].
\end{equation}
\end{definition}
It is proved in Appendix~\ref{Sec:PauliGrowth} that the Pauli growth satisfies the following properties, implying that it is a resource monotone, 
\begin{enumerate}
	\item (Faithfulness) $G(V)\geq 0$ for any unitary $V$, and $G(U)=0$ iff $U$ is a non-entangling unitary,
	\item (Invariance) $G(U_1VU_2)= G(V)$ for any unitary $V$ and non-entangling unitaries $U_1$ and $U_2$.\label{Condition2} 
\end{enumerate}
Faithfulness guarantees that only resourceful unitaries generate entanglement scrambling, indicated by a positive Pauli growth. The Pauli growth measures the increase in the average Pauli weight of a weight-1 operator under unitary evolution (see Fig. \ref{Fig:Scrambling}). In other words, it quantifies operator spreading.

Scrambling in an $n$-qubit system is commonly studied by utilizing the out-of-time-ordered correlator, defined as
\begin{equation}\label{Eq:OTOCdef}
	\OTOC(U)=\langle U^\dagger P_D U P_A U^\dagger P_D U P_A \rangle,
\end{equation}
where $P_A$ and $P_D$ are Pauli operators which act non-trivially only on the subsystems $A$ and $D$, respectively. We define the notation $\langle \cdot \rangle\equiv \frac{1}{2^n}\tr{\cdot}$. The OTOC is related to commutator growth via ${\norm{[U^\dagger P_D U, P_A]}_{\mathrm{HS}}^2=2^{n+1}\left(1-\OTOC(U)\right)}$, where $\norm{\cdot}_{\mathrm{HS}}$ denotes the Hilbert-Schmidt norm. For disjoint subsystems $A$ and $D$, this commutator norm measures the spread of the support of $P_D$ to the subsystem $A$ after conjugation by $U$. A small OTOC value is traditionally considered as a signature of scrambling. In the large $n$ limit, the OTOC of $U$ and the Pauli growth satisfy the following relation  (see Appendix~\ref{Sec:ProofOTOCScram} for a proof, which is based on the results in \cite{Bu2022}):  
\begin{equation}
\begin{split}
\av{A}\av{P_A\neq I_A}\OTOC(U)\geq 1-\frac{4}{3n}\left(G(U)+1\right),
\end{split}
\end{equation}
where $D$ is the $n$-th qubit, $A$ is any other single-qubit subsystem, $\av{A}$ is the uniform average over all choices of $A$, and $\av{P_A\neq I_A}$ is the uniform average over all non-identity Pauli operators on $A$. The OTOC is hence an indicator of operator spreading.

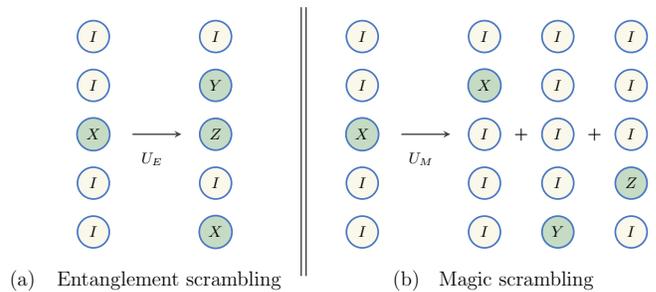
\begin{figure}[h!]
\scalebox{.65}{
\hspace{-1.5em}
\begin{tikzpicture}

%Line Divisions
\draw [thick,color=charcoal]
(2.2*\xrel,.7*\yrel-2.75*\rowspace)--(2.2*\xrel,.7*\yrel+2.75*\rowspace)
(2.3*\xrel,.7*\yrel-2.75*\rowspace)--(2.3*\xrel,.7*\yrel+2.75*\rowspace);
    
%Diagram Labels
  	\renewcommand{\nodenum}{v2}
    \renewcommand{\name}{\large (a)}
	\renewcommand{\xpos}{\xglobalshift-3.5*\xrel}
    \renewcommand{\ypos}{\yglobalshift-5*\yrel}
    \renewcommand{\height}{\heightsingle}
    \renewcommand{\width}{\widthsingle}
    \node[] (\nodenum) at (\xpos,\ypos) {\name}; 

    \renewcommand{\nodenum}{v3}
    \renewcommand{\name}{\large Entanglement scrambling}
	\renewcommand{\xpos}{\xglobalshift-0.5*\xrel}
    \renewcommand{\ypos}{\yglobalshift-5*\yrel}
    \renewcommand{\height}{\heightsingle}
    \renewcommand{\width}{\widthsingle}
    \node[] (\nodenum) at (\xpos,\ypos) {\name}; 

    \renewcommand{\nodenum}{v2}
    \renewcommand{\name}{\large (b)}
	\renewcommand{\xpos}{\xglobalshift+4.35*\xrel}
    \renewcommand{\ypos}{\yglobalshift-5*\yrel}
    \renewcommand{\height}{\heightsingle}
    \renewcommand{\width}{\widthsingle}
    \node[] (\nodenum) at (\xpos,\ypos) {\name}; 
    
    \renewcommand{\nodenum}{v2}
    \renewcommand{\name}{\large Magic scrambling}
	\renewcommand{\xpos}{\xglobalshift+6.6*\xrel}
    \renewcommand{\ypos}{\yglobalshift-5*\yrel}
    \renewcommand{\height}{\heightsingle}
    \renewcommand{\width}{\widthsingle}
    \node[] (\nodenum) at (\xpos,\ypos) {\name}; 
    
\renewcommand{\xglobalshift}{-.5*\xrel} %Global Shift Parameter
\renewcommand{\yglobalshift}{0} %Global shift parameter

  	\renewcommand{\nodenum}{v2}
    \renewcommand{\name}{$I$}
	\renewcommand{\xpos}{\xglobalshift-1.55*\xrel}
    \renewcommand{\ypos}{\yglobalshift+5*\yrel}
    \renewcommand{\height}{\heightsingle}
    \renewcommand{\width}{\widthsingle}
    \node[circle,, line width=.35mm, fill=egg, rounded corners, minimum width=\width, minimum height=\height, draw=dullblue] (\nodenum) at (\xpos,\ypos) {\name}; 
    
  	\renewcommand{\nodenum}{v2}
    \renewcommand{\name}{$I$}
	\renewcommand{\xpos}{\xglobalshift-1.55*\xrel}
    \renewcommand{\ypos}{\yglobalshift+3*\yrel}
    \renewcommand{\height}{\heightsingle}
    \renewcommand{\width}{\widthsingle}
    \node[circle,, line width=.35mm, fill=egg, rounded corners, minimum width=\width, minimum height=\height, draw=dullblue] (\nodenum) at (\xpos,\ypos) {\name}; 
    
	\renewcommand{\nodenum}{v3}
    \renewcommand{\name}{$X$}
	\renewcommand{\xpos}{\xglobalshift-1.55*\xrel}
    \renewcommand{\ypos}{\yglobalshift+\yrel}
    \renewcommand{\height}{\heightsingle}
    \renewcommand{\width}{\widthsingle}
    \node[circle,, line width=.35mm, fill=evergreen, rounded corners, minimum width=\width, minimum height=\height, draw=dullblue] (\nodenum) at (\xpos,\ypos) {\name}; 
        
	\renewcommand{\nodenum}{v3}
    \renewcommand{\name}{$I$}
	\renewcommand{\xpos}{\xglobalshift-1.55*\xrel}
    \renewcommand{\ypos}{\yglobalshift+\yrel-\rowspace}
    \renewcommand{\height}{\heightsingle}
    \renewcommand{\width}{\widthsingle}
    \node[circle, line width=.35mm, fill=egg, rounded corners, minimum width=\width, minimum height=\height, draw=dullblue] (\nodenum) at (\xpos,\ypos) {\name}; 
            
	\renewcommand{\nodenum}{v3}
    \renewcommand{\name}{$I$}
	\renewcommand{\xpos}{\xglobalshift-1.55*\xrel}
    \renewcommand{\ypos}{\yglobalshift+\yrel-2*\rowspace}
    \renewcommand{\height}{\heightsingle}
    \renewcommand{\width}{\widthsingle}
    \node[circle, line width=.35mm, fill=egg, rounded corners, minimum width=\width, minimum height=\height, draw=dullblue] (\nodenum) at (\xpos,\ypos) {\name}; 
            
	\renewcommand{\nodenum}{v3}
    \renewcommand{\name}{$U_E$}
	\renewcommand{\xpos}{\xglobalshift-.35*\xrel}
    \renewcommand{\ypos}{\yglobalshift+\yrel-.5*\rowspace}
    \renewcommand{\height}{\heightsingle}
    \renewcommand{\width}{\widthsingle}
    \node[] (\nodenum) at (\xpos,\ypos) {\name}; 
    
    \draw [-stealth,thick,color=charcoal](\xglobalshift-.75*\xrel,\yglobalshift+1*\yrel) -- (\xglobalshift+.25*\xrel,\yglobalshift+\yrel);
    
\renewcommand{\xglobalshift}{2*\xrel} %Global Shift Parameter
\renewcommand{\yglobalshift}{0} %Global shift parameter

  	\renewcommand{\nodenum}{v2}
    \renewcommand{\name}{$I$}
	\renewcommand{\xpos}{\xglobalshift-1.55*\xrel}
    \renewcommand{\ypos}{\yglobalshift+5*\yrel}
    \renewcommand{\height}{\heightsingle}
    \renewcommand{\width}{\widthsingle}
    \node[circle,, line width=.35mm, fill=egg, rounded corners, minimum width=\width, minimum height=\height, draw=dullblue] (\nodenum) at (\xpos,\ypos) {\name}; 
    
  	\renewcommand{\nodenum}{v2}
    \renewcommand{\name}{$Y$}
	\renewcommand{\xpos}{\xglobalshift-1.55*\xrel}
    \renewcommand{\ypos}{\yglobalshift+3*\yrel}
    \renewcommand{\height}{\heightsingle}
    \renewcommand{\width}{\widthsingle}
    \node[circle,, line width=.35mm, fill=evergreen, rounded corners, minimum width=\width, minimum height=\height, draw=dullblue] (\nodenum) at (\xpos,\ypos) {\name}; 
    
	\renewcommand{\nodenum}{v3}
    \renewcommand{\name}{$Z$}
	\renewcommand{\xpos}{\xglobalshift-1.55*\xrel}
    \renewcommand{\ypos}{\yglobalshift+\yrel}
    \renewcommand{\height}{\heightsingle}
    \renewcommand{\width}{\widthsingle}
    \node[circle,, line width=.35mm, fill=evergreen, rounded corners, minimum width=\width, minimum height=\height, draw=dullblue] (\nodenum) at (\xpos,\ypos) {\name}; 
        
	\renewcommand{\nodenum}{v3}
    \renewcommand{\name}{$I$}
	\renewcommand{\xpos}{\xglobalshift-1.55*\xrel}
    \renewcommand{\ypos}{\yglobalshift+\yrel-\rowspace}
    \renewcommand{\height}{\heightsingle}
    \renewcommand{\width}{\widthsingle}
    \node[circle,, line width=.35mm, fill=egg, rounded corners, minimum width=\width, minimum height=\height, draw=dullblue] (\nodenum) at (\xpos,\ypos) {\name}; 
            
	\renewcommand{\nodenum}{v3}
    \renewcommand{\name}{$X$}
	\renewcommand{\xpos}{\xglobalshift-1.55*\xrel}
    \renewcommand{\ypos}{\yglobalshift+\yrel-2*\rowspace}
    \renewcommand{\height}{\heightsingle}
    \renewcommand{\width}{\widthsingle}
    \node[circle,, line width=.35mm, fill=evergreen, rounded corners, minimum width=\width, minimum height=\height, draw=dullblue] (\nodenum) at (\xpos,\ypos) {\name}; 
    
\renewcommand{\xglobalshift}{5*\xrel} %Global Shift Parameter
\renewcommand{\yglobalshift}{0} %Global shift parameter

  	\renewcommand{\nodenum}{v2}
    \renewcommand{\name}{$I$}
	\renewcommand{\xpos}{\xglobalshift-1.55*\xrel}
    \renewcommand{\ypos}{\yglobalshift+5*\yrel}
    \renewcommand{\height}{\heightsingle}
    \renewcommand{\width}{\widthsingle}
    \node[circle,, line width=.35mm, fill=egg, rounded corners, minimum width=\width, minimum height=\height, draw=dullblue] (\nodenum) at (\xpos,\ypos) {\name}; 
    
  	\renewcommand{\nodenum}{v2}
    \renewcommand{\name}{$I$}
	\renewcommand{\xpos}{\xglobalshift-1.55*\xrel}
    \renewcommand{\ypos}{\yglobalshift+3*\yrel}
    \renewcommand{\height}{\heightsingle}
    \renewcommand{\width}{\widthsingle}
    \node[circle,, line width=.35mm, fill=egg, rounded corners, minimum width=\width, minimum height=\height, draw=dullblue] (\nodenum) at (\xpos,\ypos) {\name}; 
    
	\renewcommand{\nodenum}{v3}
    \renewcommand{\name}{$X$}
	\renewcommand{\xpos}{\xglobalshift-1.55*\xrel}
    \renewcommand{\ypos}{\yglobalshift+\yrel}
    \renewcommand{\height}{\heightsingle}
    \renewcommand{\width}{\widthsingle}
    \node[circle,, line width=.35mm, fill=evergreen, rounded corners, minimum width=\width, minimum height=\height, draw=dullblue] (\nodenum) at (\xpos,\ypos) {\name}; 
        
	\renewcommand{\nodenum}{v3}
    \renewcommand{\name}{$I$}
	\renewcommand{\xpos}{\xglobalshift-1.55*\xrel}
    \renewcommand{\ypos}{\yglobalshift+\yrel-\rowspace}
    \renewcommand{\height}{\heightsingle}
    \renewcommand{\width}{\widthsingle}
    \node[circle,, line width=.35mm, fill=egg, rounded corners, minimum width=\width, minimum height=\height, draw=dullblue] (\nodenum) at (\xpos,\ypos) {\name}; 
            
	\renewcommand{\nodenum}{v3}
    \renewcommand{\name}{$I$}
	\renewcommand{\xpos}{\xglobalshift-1.55*\xrel}
    \renewcommand{\ypos}{\yglobalshift+\yrel-2*\rowspace}
    \renewcommand{\height}{\heightsingle}
    \renewcommand{\width}{\widthsingle}
    \node[circle,, line width=.35mm, fill=egg, rounded corners, minimum width=\width, minimum height=\height, draw=dullblue] (\nodenum) at (\xpos,\ypos) {\name}; 
            
	\renewcommand{\nodenum}{v3}
    \renewcommand{\name}{$U_M$}
	\renewcommand{\xpos}{\xglobalshift-.35*\xrel}
    \renewcommand{\ypos}{\yglobalshift+\yrel-.5*\rowspace}
    \renewcommand{\height}{\heightsingle}
    \renewcommand{\width}{\widthsingle}
    \node[] (\nodenum) at (\xpos,\ypos) {\name}; 
    
    \draw [-stealth,thick,color=charcoal](\xglobalshift-.75*\xrel,\yglobalshift+1*\yrel) -- (\xglobalshift+.25*\xrel,\yglobalshift+\yrel);
    
\renewcommand{\xglobalshift}{7.5*\xrel} %Global Shift Parameter
\renewcommand{\yglobalshift}{0} %Global shift parameter

  	\renewcommand{\nodenum}{v2}
    \renewcommand{\name}{$I$}
	\renewcommand{\xpos}{\xglobalshift-1.55*\xrel}
    \renewcommand{\ypos}{\yglobalshift+5*\yrel}
    \renewcommand{\height}{\heightsingle}
    \renewcommand{\width}{\widthsingle}
    \node[circle,, line width=.35mm, fill=egg, rounded corners, minimum width=\width, minimum height=\height, draw=dullblue] (\nodenum) at (\xpos,\ypos) {\name}; 
    
  	\renewcommand{\nodenum}{v2}
    \renewcommand{\name}{$X$}
	\renewcommand{\xpos}{\xglobalshift-1.55*\xrel}
    \renewcommand{\ypos}{\yglobalshift+3*\yrel}
    \renewcommand{\height}{\heightsingle}
    \renewcommand{\width}{\widthsingle}
    \node[circle,, line width=.35mm, fill=evergreen, rounded corners, minimum width=\width, minimum height=\height, draw=dullblue] (\nodenum) at (\xpos,\ypos) {\name}; 
    
	\renewcommand{\nodenum}{v3}
    \renewcommand{\name}{$I$}
	\renewcommand{\xpos}{\xglobalshift-1.55*\xrel}
    \renewcommand{\ypos}{\yglobalshift+\yrel}
    \renewcommand{\height}{\heightsingle}
    \renewcommand{\width}{\widthsingle}
    \node[circle,, line width=.35mm, fill=egg, rounded corners, minimum width=\width, minimum height=\height, draw=dullblue] (\nodenum) at (\xpos,\ypos) {\name}; 
        
	\renewcommand{\nodenum}{v3}
    \renewcommand{\name}{$I$}
	\renewcommand{\xpos}{\xglobalshift-1.55*\xrel}
    \renewcommand{\ypos}{\yglobalshift+\yrel-\rowspace}
    \renewcommand{\height}{\heightsingle}
    \renewcommand{\width}{\widthsingle}
    \node[circle,, line width=.35mm, fill=egg, rounded corners, minimum width=\width, minimum height=\height, draw=dullblue] (\nodenum) at (\xpos,\ypos) {\name}; 
            
	\renewcommand{\nodenum}{v3}
    \renewcommand{\name}{$I$}
	\renewcommand{\xpos}{\xglobalshift-1.55*\xrel}
    \renewcommand{\ypos}{\yglobalshift+\yrel-2*\rowspace}
    \renewcommand{\height}{\heightsingle}
    \renewcommand{\width}{\widthsingle}
    \node[circle,, line width=.35mm, fill=egg, rounded corners, minimum width=\width, minimum height=\height, draw=dullblue] (\nodenum) at (\xpos,\ypos) {\name}; 
    
\renewcommand{\xglobalshift}{9*\xrel} %Global Shift Parameter
\renewcommand{\yglobalshift}{0} %Global shift parameter

  	\renewcommand{\nodenum}{v2}
    \renewcommand{\name}{$\bm{+}$}
	\renewcommand{\xpos}{\xglobalshift-2.3*\xrel}
    \renewcommand{\ypos}{\yglobalshift+1*\yrel}
    \renewcommand{\height}{\heightsingle}
    \renewcommand{\width}{\widthsingle}
    \node[] (\nodenum) at (\xpos,\ypos) {\name}; 
    
  	\renewcommand{\nodenum}{v2}
    \renewcommand{\name}{$I$}
	\renewcommand{\xpos}{\xglobalshift-1.55*\xrel}
    \renewcommand{\ypos}{\yglobalshift+5*\yrel}
    \renewcommand{\height}{\heightsingle}
    \renewcommand{\width}{\widthsingle}
    \node[circle,, line width=.35mm, fill=egg, rounded corners, minimum width=\width, minimum height=\height, draw=dullblue] (\nodenum) at (\xpos,\ypos) {\name}; 
    
  	\renewcommand{\nodenum}{v2}
    \renewcommand{\name}{$I$}
	\renewcommand{\xpos}{\xglobalshift-1.55*\xrel}
    \renewcommand{\ypos}{\yglobalshift+3*\yrel}
    \renewcommand{\height}{\heightsingle}
    \renewcommand{\width}{\widthsingle}
    \node[circle,, line width=.35mm, fill=egg, rounded corners, minimum width=\width, minimum height=\height, draw=dullblue] (\nodenum) at (\xpos,\ypos) {\name}; 
    
	\renewcommand{\nodenum}{v3}
    \renewcommand{\name}{$I$}
	\renewcommand{\xpos}{\xglobalshift-1.55*\xrel}
    \renewcommand{\ypos}{\yglobalshift+\yrel}
    \renewcommand{\height}{\heightsingle}
    \renewcommand{\width}{\widthsingle}
    \node[circle,, line width=.35mm, fill=egg, rounded corners, minimum width=\width, minimum height=\height, draw=dullblue] (\nodenum) at (\xpos,\ypos) {\name}; 
        
	\renewcommand{\nodenum}{v3}
    \renewcommand{\name}{$I$}
	\renewcommand{\xpos}{\xglobalshift-1.55*\xrel}
    \renewcommand{\ypos}{\yglobalshift+\yrel-\rowspace}
    \renewcommand{\height}{\heightsingle}
    \renewcommand{\width}{\widthsingle}
    \node[circle,, line width=.35mm, fill=egg, rounded corners, minimum width=\width, minimum height=\height, draw=dullblue] (\nodenum) at (\xpos,\ypos) {\name}; 
            
	\renewcommand{\nodenum}{v3}
    \renewcommand{\name}{$Y$}
	\renewcommand{\xpos}{\xglobalshift-1.55*\xrel}
    \renewcommand{\ypos}{\yglobalshift+\yrel-2*\rowspace}
    \renewcommand{\height}{\heightsingle}
    \renewcommand{\width}{\widthsingle}
    \node[circle,, line width=.35mm, fill=evergreen, rounded corners, minimum width=\width, minimum height=\height, draw=dullblue] (\nodenum) at (\xpos,\ypos) {\name}; 
    
\renewcommand{\xglobalshift}{10.5*\xrel} %Global Shift Parameter
\renewcommand{\yglobalshift}{0} %Global shift parameter

  	\renewcommand{\nodenum}{v2}
    \renewcommand{\name}{$\bm{+}$}
	\renewcommand{\xpos}{\xglobalshift-2.3*\xrel}
    \renewcommand{\ypos}{\yglobalshift+1*\yrel}
    \renewcommand{\height}{\heightsingle}
    \renewcommand{\width}{\widthsingle}
    \node[] (\nodenum) at (\xpos,\ypos) {\name}; 
    
  	\renewcommand{\nodenum}{v2}
    \renewcommand{\name}{$I$}
	\renewcommand{\xpos}{\xglobalshift-1.55*\xrel}
    \renewcommand{\ypos}{\yglobalshift+5*\yrel}
    \renewcommand{\height}{\heightsingle}
    \renewcommand{\width}{\widthsingle}
    \node[circle,, line width=.35mm, fill=egg, rounded corners, minimum width=\width, minimum height=\height, draw=dullblue] (\nodenum) at (\xpos,\ypos) {\name}; 
    
  	\renewcommand{\nodenum}{v2}
    \renewcommand{\name}{$I$}
	\renewcommand{\xpos}{\xglobalshift-1.55*\xrel}
    \renewcommand{\ypos}{\yglobalshift+3*\yrel}
    \renewcommand{\height}{\heightsingle}
    \renewcommand{\width}{\widthsingle}
    \node[circle,, line width=.35mm, fill=egg, rounded corners, minimum width=\width, minimum height=\height, draw=dullblue] (\nodenum) at (\xpos,\ypos) {\name}; 
    
	\renewcommand{\nodenum}{v3}
    \renewcommand{\name}{$I$}
	\renewcommand{\xpos}{\xglobalshift-1.55*\xrel}
    \renewcommand{\ypos}{\yglobalshift+\yrel}
    \renewcommand{\height}{\heightsingle}
    \renewcommand{\width}{\widthsingle}
    \node[circle,, line width=.35mm, fill=egg, rounded corners, minimum width=\width, minimum height=\height, draw=dullblue] (\nodenum) at (\xpos,\ypos) {\name}; 
        
	\renewcommand{\nodenum}{v3}
    \renewcommand{\name}{$Z$}
	\renewcommand{\xpos}{\xglobalshift-1.55*\xrel}
    \renewcommand{\ypos}{\yglobalshift+\yrel-\rowspace}
    \renewcommand{\height}{\heightsingle}
    \renewcommand{\width}{\widthsingle}
    \node[circle,, line width=.35mm, fill=evergreen, rounded corners, minimum width=\width, minimum height=\height, draw=dullblue] (\nodenum) at (\xpos,\ypos) {\name}; 
            
	\renewcommand{\nodenum}{v3}
    \renewcommand{\name}{$I$}
	\renewcommand{\xpos}{\xglobalshift-1.55*\xrel}
    \renewcommand{\ypos}{\yglobalshift+\yrel-2*\rowspace}
    \renewcommand{\height}{\heightsingle}
    \renewcommand{\width}{\widthsingle}
    \node[circle,, line width=.35mm, fill=egg, rounded corners, minimum width=\width, minimum height=\height, draw=dullblue] (\nodenum) at (\xpos,\ypos) {\name}; 
\end{tikzpicture}
}
\caption{Schematic examples of the two scrambling mechanisms. (a) Entanglement scrambling: a unitary $U_E$ maps a weight-1 Pauli operator to a weight-3 Pauli operator via conjugation, i.e. $U_E^\dagger (I\otimes I \otimes X\otimes I\otimes I)U_E=I\otimes Y\otimes Z\otimes I \otimes X$. (b) Magic scrambling: a unitary $U_M$ maps a weight-1 Pauli operator to a sum of three weight-1 Pauli operators. Here, $U_E$ ($U_M$) generates no magic (entanglement) scrambling. }
\label{Fig:Scrambling}
\end{figure}

\subsection{Magic scrambling}
We now introduce the framework for magic scrambling. A free unitary is defined to map any Pauli operator to a Pauli operator (up to a phase) under conjugation. By definition, free unitaries are Clifford unitaries. A non-Clifford unitary maps a Pauli operator to a superposition of Pauli operators, i.e. it generates operator entanglement (see Fig. \ref{Fig:Scrambling}). Magic monotones quantify the distance between a unitary and the set of Clifford unitaries. This framework is identical to the resource theory of magic, but we refer to it as magic scrambling to emphasize its operational interpretation.

We introduce a magic monotone, which we call the OTOC magic.

\begin{definition}\label{Lemma:LinearMagic}
The OTOC magic of an $n$-qubit unitary $U$ is
\begin{equation}
	O_M(U)\equiv\max_{P_{\vec{a}},P_{\vec{b}}\in \mathcal{P}_2^{\otimes n}}\left[
	1-\abs{\OTOC(U)}\right],
\end{equation}
where $\mathcal{P}_2$ is the qubit Pauli group and  ${\OTOC(U)=\langle U^\dagger P_{\vec{a}} U P_{\vec{b}} U^\dagger P_{\vec{a}} U P_{\vec{b}} \rangle}$.
\end{definition}
In Appendix~\ref{Proof:LemmaOTOCMagic}, we prove that the OTOC magic satisfies the following monotone properties:
 \begin{enumerate}
	\item (Faithfulness)  $O_M(V)\geq 0$ for any unitary $V$, and $O_M(U)=0$ iff $U$ is a Clifford unitary,
	\item (Invariance) $O_M(U_1VU_2)= O_M(V)$ for any unitary $V$ and Clifford unitaries $U_1$ and $U_2$. 
\end{enumerate}
We compute the OTOC magic for two examples of gates in the Clifford hierarchy \cite{Gottesman_1999}. The \mbox{$k$-th} level of the Clifford hierarchy is defined as ${\mathcal{C}^{(k)}=\{U\in \mathcal{U}(n): U^\dagger \mathcal{P}_2^{\otimes n} U \subset \mathcal{C}^{(k-1)}\}}$, where  $\mathcal{C}^{(1)}$ is the Pauli group and $\mathcal{C}^{(2)}$ is the Clifford group.
\begin{example}
\rm{All non-Clifford unitaries in the 3rd level of the Clifford hierarchy maximize the OTOC magic (see Appendix~\ref{Sec:ProofCliffordHierarchy} for a proof):
\begin{eqnarray}
O_M(U)=1,~~~ \forall U\in \mathcal{C}^{(3)}\backslash \mathcal{C}^{(2)}.
\end{eqnarray}}
\end{example}\label{Ex:non-clifford}

\begin{example}\label{Ex:Phase}
\rm{Define the single-qubit phase shift gate as
\begin{equation}\label{Eq:EpsilonGate}
    U_\varepsilon=\begin{pmatrix}
1 & 0 \\
0 & e^{i\varepsilon}
\end{pmatrix},
\end{equation}
where $\varepsilon\in [0,2\pi)$. Its OTOC magic is ${O_M(U_{\varepsilon})=1-\abs{\cos(2\varepsilon)}}$ (see Appendix~\ref{Sec:ProofPropepsilon} for a proof).
Let $\varepsilon_k=\frac{\pi}{2^{k-1}}$ for any integer $k> 0$. Then $U_{\varepsilon_k}\in\mathcal{C}^{(k)}$ \cite{PhysRevA.95.012329} and the OTOC magic is
\begin{equation}\label{Eq:HierarchyMagic}
	O_M(U_{\varepsilon_k})=1-\abs{\cos\left(\tfrac{\pi}{2^{k-2}}\right)}.
\end{equation}
}
\rm{Since gates in the $k$-th level of the Clifford hierarchy map Pauli operators to gates in the $k-1$ level, one may be tempted to interpret the level of the hierarchy as a measure of a gate's `distance' from the second level (i.e. as a measure of magic). However, according to Eq.~\eqref{Eq:HierarchyMagic}, for $k\geq 3$, the OTOC magic of $U_{\varepsilon_k}$ \textit{decreases} as $k$ increases, indicating that some gates in the higher levels of the Clifford hierarchy can have a small amount of magic  (see Fig.~\ref{Fig:Hierarchy}).}

\end{example}

\begin{figure}[h!]
\vspace{-1em}
\includegraphics[scale=.4]{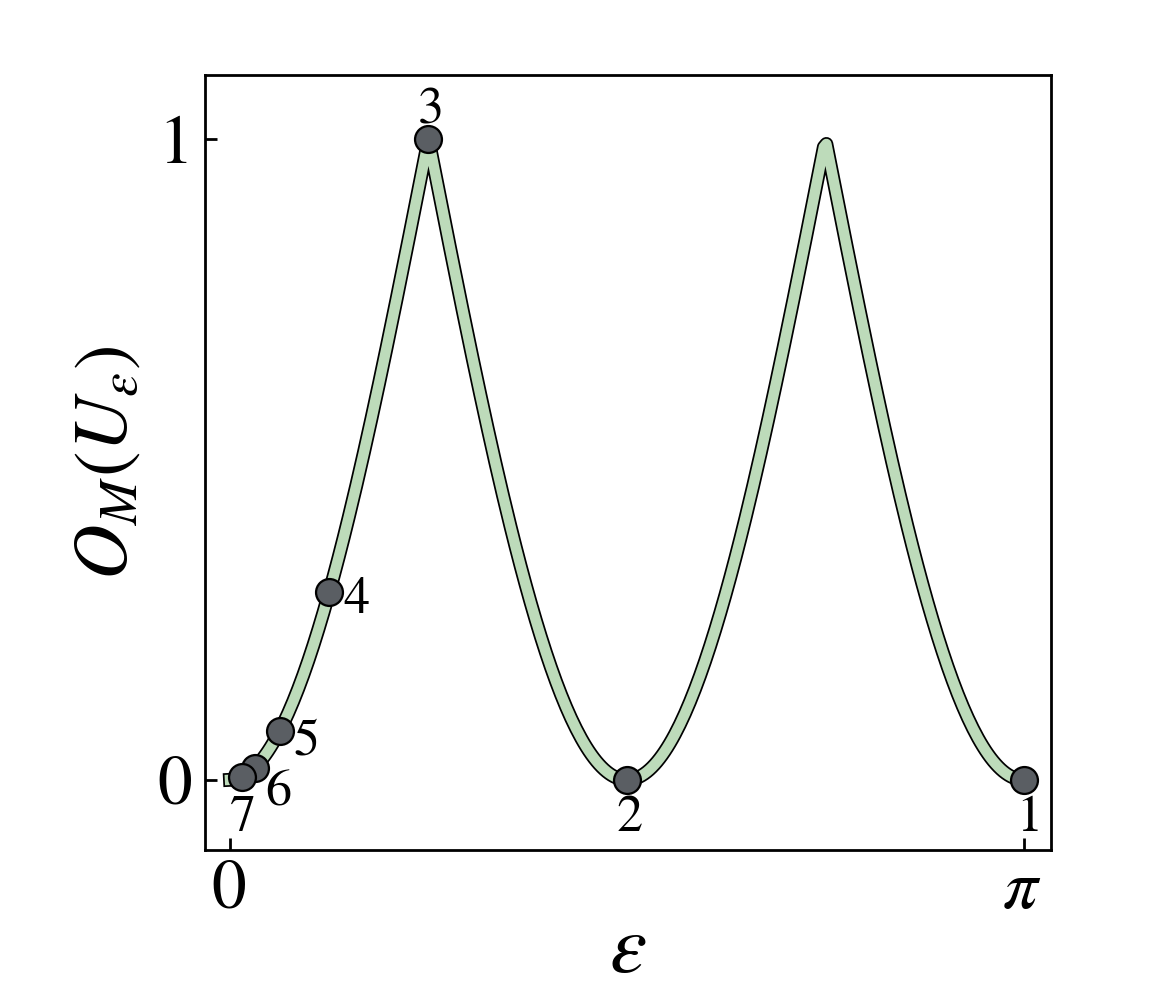}
\vspace{-1em}
\caption{Plot of $O_M(U_{\varepsilon})$ with respect to $\varepsilon$. Numbered points label the values of $O_M(U_{\varepsilon_k})$ when $\varepsilon_k=\frac{\pi}{2^{k-1}}$ and $k=1,2,\ldots, 7$. The $k$-th point corresponds to the OTOC magic of a unitary in the $k$-th level of the Clifford hierarchy. For $k\geq 3$, the OTOC magic decreases with increasing $k$.}
\label{Fig:Hierarchy}
\end{figure}

The OTOC magic can be measured via OTOC measurement protocols based on: a randomized measurement toolbox \cite{PhysRevX.9.021061,Joshi_2020}, classical shadows \cite{PhysRevResearch.3.033155}, and teleportation \cite{Landsman_2019}. Such protocols allow one to measure the OTOC by circumventing measurement techniques such as time reversal \cite{PhysRevA.94.040302,G_rttner_2017,PhysRevLett.120.070501}. Other magic monotones ~\cite{leone2022, PhysRevLett.128.050402,Oliviero2022} have been theoretically related to OTOCs.

We develop a scrambling resource theory based on two separate mechanisms, as some unitaries may be free with respect to one mechanism, but a resource with respect to the other. For example, the T gate is a free entanglement scrambling unitary, but it is a magic scrambling resource. The CNOT gate is a Clifford unitary, but it is an entanglement scrambling resource.

\subsection{Application to Google's  experimental results}
The OTOC fluctuations~\footnote{In this experiment, the OTOC is defined as an expectation value with respect to a pure state, but we utilize the maximally mixed state, as in Eq.~\eqref{Eq:OTOCdef}.}, defined as
\begin{equation}
	\delta=\sqrt{\av{U\sim \mathcal{E}}\abs{\OTOC(U)-\av{V\sim\mathcal{E}}\OTOC(V)}^2}\;,
\end{equation}
where the averages are taken over a unitary ensemble $\mathcal{E}$, have been measured by Google on the Sycamore quantum processor~\cite{Mi_2021}. Small fluctuations were shown to be evidence of operator entanglement. When $\mathcal{E}$ is the Clifford ensemble, $\av{U\sim\mathcal{E}}\OTOC(U)\rightarrow 0$ for large systems and $\OTOC(U)$ fluctuates between $+1$ and $-1$, implying that $\delta=1$. However, when $U$ is sampled from a non-Clifford ensemble, then $\abs{\OTOC(U)}\leq 1$. It was shown in this experiment that OTOC fluctuations decay as the magic of the unitaries in the ensemble is increased. Also, it was numerically shown in~\cite{Ahmadi2022} that the average mana, a magic monotone, increases as the OTOC fluctuations decrease. Here, we establish an inequality between $\delta$ and the OTOC magic.

\begin{theorem}\label{Thm:Fluctuations}
If $\av{U\sim \mathcal{E}}\OTOC(U)\rightarrow 0$, then 
\begin{equation}\label{Eq:FluxBound}
	\av{U\sim \mathcal{E}}O_M(U)\geq 1-\delta.
\end{equation}
\end{theorem}
We prove Theorem~\ref{Thm:Fluctuations} in Appendix~\ref{Proof:Fluctuations}. Small OTOC fluctuations are therefore an indication of magic and operator entanglement \footnote{
In fact, a small value of $\abs{\OTOC(U)}$ is an indication of the magic of $U$, since it bounds $O_M(U)$.}. 

We provide numerical simulations to support the bound in Eq.~\eqref{Eq:FluxBound}. The ensemble $\mathcal{E}$ consists of 4-qubit unitaries of the form:
\begin{equation}\label{Eq:RandomUnitary}
\scalebox{.6}{
\begin{tikzpicture}
    %Lines
    \draw [thick,color=charcoal]
    (0,\yrel+2*\rowspace)--(11*\xrel,\yrel+2*\rowspace)
    (0,\yrel+\rowspace)--(11*\xrel,\yrel+\rowspace)
    (0,\yrel)--(11*\xrel,\yrel)
    (0,-\yrel)--(11*\xrel,-\yrel);
    
\renewcommand{\xglobalshift}{0}
\renewcommand{\yglobalshift}{0}

	\renewcommand{\nodenum}{v1}
    \renewcommand{\name}{\large $U=$}
	\renewcommand{\xpos}{-\xrel+\xglobalshift}
    \renewcommand{\ypos}{1*\rowspace+\yglobalshift}
    \renewcommand{\height}{\heightdouble}
    \renewcommand{\width}{\widthsingle}
    \node[] (\nodenum) at (\xpos,\ypos) {\name};

    \renewcommand{\nodenum}{v1}
    \renewcommand{\name}{}
	\renewcommand{\xpos}{2*\xrel+\xglobalshift}
    \renewcommand{\ypos}{\rowspace+\yglobalshift}
    \renewcommand{\height}{\heightdouble}
    \renewcommand{\width}{\widthsingle}
    \node[rectangle, fill=egg,  line width =.3mm,rounded corners, minimum width=\width, minimum height=  	    \height, draw=charcoal] (\nodenum) at (\xpos,\ypos) {\name};
    
	\renewcommand{\nodenum}{v2}
    \renewcommand{\name}{}
	\renewcommand{\xpos}{3*\xrel+\xglobalshift}
    \renewcommand{\ypos}{2*\rowspace+\yglobalshift}
    \renewcommand{\height}{\heightdouble}
    \renewcommand{\width}{\widthsingle}
    \node[rectangle, fill=egg,  line width =.3mm,rounded corners, minimum width=\width, minimum height=  	    \height, draw=charcoal] (\nodenum) at (\xpos,\ypos) {\name};
    
    \renewcommand{\nodenum}{v3}
    \renewcommand{\name}{}
	\renewcommand{\xpos}{3*\xrel+\xglobalshift}
    \renewcommand{\ypos}{\yglobalshift}
    \renewcommand{\height}{\heightdouble}
    \renewcommand{\width}{\widthsingle}
    \node[rectangle, fill=egg,  line width =.3mm,rounded corners, minimum width=\width, minimum height=  	    \height, draw=charcoal] (\nodenum) at (\xpos,\ypos) {\name};
    
	\renewcommand{\nodenum}{v9}
    \renewcommand{\name}{$\varepsilon$}
	\renewcommand{\xpos}{1*\xrel+\xglobalshift}
    \renewcommand{\ypos}{-1*\rowspace+\yrel}
    \renewcommand{\height}{\heightsingle}
    \renewcommand{\width}{\widthsingle}
    \node[rectangle, fill=evergreen,  line width =.35mm,rounded corners, minimum width=\width, minimum height=  	    \height, draw=charcoal] (\nodenum) at (\xpos,\ypos) {\name};
    
\renewcommand{\xglobalshift}{4*\xrel}
\renewcommand{\yglobalshift}{0}

    \renewcommand{\nodenum}{v1}
    \renewcommand{\name}{}
	\renewcommand{\xpos}{2*\xrel+\xglobalshift}
    \renewcommand{\ypos}{\rowspace+\yglobalshift}
    \renewcommand{\height}{\heightdouble}
    \renewcommand{\width}{\widthsingle}
    \node[rectangle, fill=egg,  line width =.3mm,rounded corners, minimum width=\width, minimum height=  	    \height, draw=charcoal] (\nodenum) at (\xpos,\ypos) {\name};
    
	\renewcommand{\nodenum}{v2}
    \renewcommand{\name}{}
	\renewcommand{\xpos}{3*\xrel+\xglobalshift}
    \renewcommand{\ypos}{2*\rowspace+\yglobalshift}
    \renewcommand{\height}{\heightdouble}
    \renewcommand{\width}{\widthsingle}
    \node[rectangle, fill=egg,  line width =.3mm,rounded corners, minimum width=\width, minimum height=  	    \height, draw=charcoal] (\nodenum) at (\xpos,\ypos) {\name};
    
    \renewcommand{\nodenum}{v3}
    \renewcommand{\name}{}
	\renewcommand{\xpos}{3*\xrel+\xglobalshift}
    \renewcommand{\ypos}{\yglobalshift}
    \renewcommand{\height}{\heightdouble}
    \renewcommand{\width}{\widthsingle}
    \node[rectangle, fill=egg,  line width =.3mm,rounded corners, minimum width=\width, minimum height=  	    \height, draw=charcoal] (\nodenum) at (\xpos,\ypos) {\name};
    
	\renewcommand{\nodenum}{v9}
    \renewcommand{\name}{$\varepsilon$}
	\renewcommand{\xpos}{1*\xrel+\xglobalshift}
    \renewcommand{\ypos}{1*\rowspace+\yrel}
    \renewcommand{\height}{\heightsingle}
    \renewcommand{\width}{\widthsingle}
    \node[rectangle, fill=evergreen,  line width =.35mm,rounded corners, minimum width=\width, minimum height=  	    \height, draw=charcoal] (\nodenum) at (\xpos,\ypos) {\name};

\renewcommand{\xglobalshift}{7*\xrel}
\renewcommand{\yglobalshift}{0}

    \renewcommand{\nodenum}{v1}
    \renewcommand{\name}{}
	\renewcommand{\xpos}{2*\xrel+\xglobalshift}
    \renewcommand{\ypos}{\rowspace+\yglobalshift}
    \renewcommand{\height}{\heightdouble}
    \renewcommand{\width}{\widthsingle}
    \node[rectangle, fill=egg,  line width =.3mm,rounded corners, minimum width=\width, minimum height=  	    \height, draw=charcoal] (\nodenum) at (\xpos,\ypos) {\name};
    
	\renewcommand{\nodenum}{v2}
    \renewcommand{\name}{}
	\renewcommand{\xpos}{3*\xrel+\xglobalshift}
    \renewcommand{\ypos}{2*\rowspace+\yglobalshift}
    \renewcommand{\height}{\heightdouble}
    \renewcommand{\width}{\widthsingle}
    \node[rectangle, fill=egg,  line width =.3mm,rounded corners, minimum width=\width, minimum height=  	    \height, draw=charcoal] (\nodenum) at (\xpos,\ypos) {\name};
    
    \renewcommand{\nodenum}{v3}
    \renewcommand{\name}{}
	\renewcommand{\xpos}{3*\xrel+\xglobalshift}
    \renewcommand{\ypos}{\yglobalshift}
    \renewcommand{\height}{\heightdouble}
    \renewcommand{\width}{\widthsingle}
    \node[rectangle, fill=egg,  line width =.3mm,rounded corners, minimum width=\width, minimum height=  	    \height, draw=charcoal] (\nodenum) at (\xpos,\ypos) {\name};
    
	\renewcommand{\nodenum}{v9}
    \renewcommand{\name}{$\varepsilon$}
	\renewcommand{\xpos}{1*\xrel+\xglobalshift}
    \renewcommand{\ypos}{2*\rowspace+\yrel}
    \renewcommand{\height}{\heightsingle}
    \renewcommand{\width}{\widthsingle}
    \node[rectangle, fill=evergreen,  line width =.35mm,rounded corners, minimum width=\width, minimum height=  	    \height, draw=charcoal] (\nodenum) at (\xpos,\ypos) {\name};
    
	\renewcommand{\nodenum}{v9}
    \renewcommand{\name}{\Large $\cdots$}
	\renewcommand{\xpos}{-3*\xrel+\xglobalshift}
    \renewcommand{\ypos}{0*\rowspace+2*\yrel}
    \renewcommand{\height}{\heightsingle}
    \renewcommand{\width}{\widthsingle}
    \node[] (\nodenum) at (\xpos,\ypos) {\name};
    
    \renewcommand{\xpos}{-.2*\xrel}
	\draw [decorate,line width=.75pt,color=charcoal,decoration={brace,amplitude=5pt},xshift=\xrel,yshift=-\rowspace]	(3*\xrel+\xglobalshift,-.15*\yrel-1*\rowspace+\yglobalshift) -- (2*\xrel+\xglobalshift,-1*\rowspace-.15*\yrel+\yglobalshift) node [black,midway,xshift=9pt] {};
	
    \renewcommand{\nodenum}{v1}
    \renewcommand{\name}{\large $l$}
	\renewcommand{\xpos}{2.5*\xrel+\xglobalshift}
    \renewcommand{\ypos}{.9*\yrel-2*\rowspace+\yglobalshift}
    \renewcommand{\height}{\heightdouble}
    \renewcommand{\width}{\widthsingle}
    \node[] (\nodenum) at (\xpos,\ypos) {\name};
    
\end{tikzpicture}
}.
\end{equation}
The 2-qubit unitaries are randomly sampled from the Clifford group and are arranged in a brick-work architecture with $l=4$ layers. The four single-qubit unitaries are each $U_{\varepsilon}$, defined in Eq.~\eqref{Eq:EpsilonGate}, with a fixed value of $\varepsilon$. We compute fluctuations of the OTOC ${\langle U^\dagger X_1 U Z_4 U^\dagger X_1 U Z_4 \rangle}$, where $X_1$ denotes a Pauli $X$ operator on the first qubit. The OTOC magic of $U$ is tuned by varying $\varepsilon$. Fig.~\ref{Fig:OTOC Flux} plots $\av{U\sim \mathcal{E}}O_M(U)$ and  $1-\delta$ against $\varepsilon$. Both are positively correlated with $\varepsilon$ and  $\av{U\sim \mathcal{E}}O_M(U)$ is bounded by $1-\delta$ from below.

\begin{figure}[t]
\includegraphics[scale=.4]{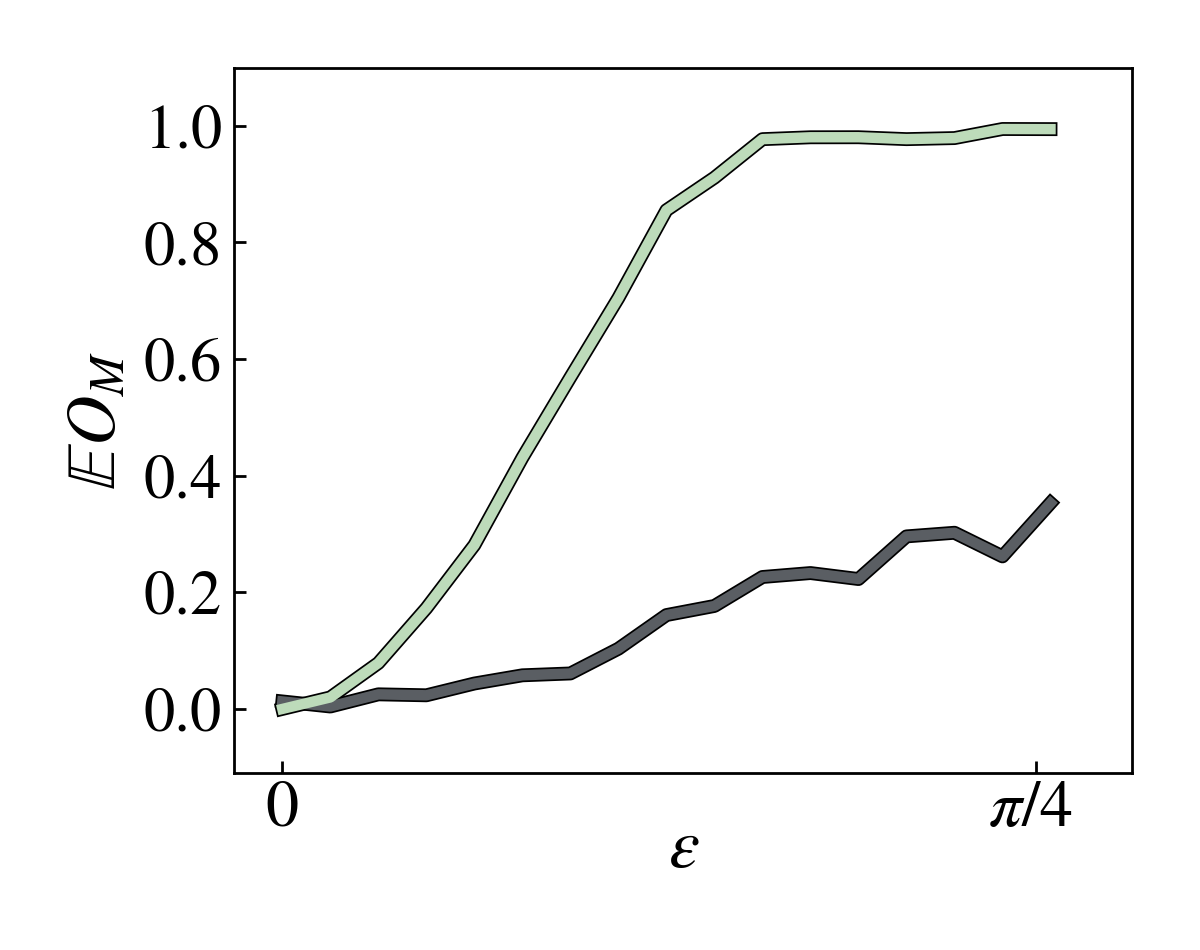}
\vspace{-1em}
\caption{Plot of $\av{U\sim \mathcal{E}}O_M(U)$ (green) and $1-\delta$ (black) against $\varepsilon$. These quantities are empirically computed using 50 random samples from the ensemble $\mathcal{E}$ generated by Eq.~\eqref{Eq:RandomUnitary}.}
\label{Fig:OTOC Flux}
\end{figure}

\subsection{Application to Hayden-Preskill decoding protocol}

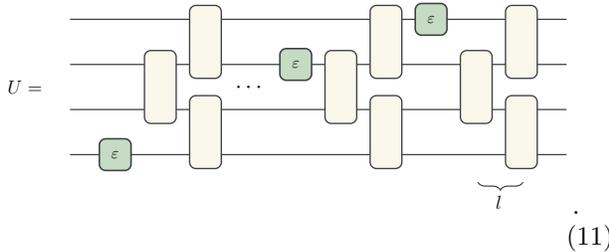
\begin{figure}[h!]
\scalebox{.7}{
\begin{tikzpicture}
    %Lines
    \draw [thick,color=charcoal]
    (0,\yrel+2*\rowspace)--(3.5*\xrel,\yrel+2*\rowspace)
    (0,\yrel+\rowspace)--(3.5*\xrel,\yrel+\rowspace)
    (0,\yrel)--(2.5*\xrel,\yrel)
    (2.5*\xrel,1.1*\yrel)--(3.5*\xrel,1.1*\yrel)
    (2.5*\xrel,0.9*\yrel)--(3.5*\xrel,0.9*\yrel)
    (0,-\yrel)--(2.5*\xrel,-\yrel)
    (2.5*\xrel,-1.1*\yrel)--(3.5*\xrel,-1.1*\yrel)
    (2.5*\xrel,-0.9*\yrel)--(3.5*\xrel,-0.9*\yrel)
    (0,-\yrel-\rowspace)--(3.5*\xrel,-\yrel-\rowspace)
    (0,-\yrel-2*\rowspace)--(3.5*\xrel,-\yrel-2*\rowspace);
    
    	%Brackets
	\renewcommand{\xpos}{-.2*\xrel}
	\draw [decorate,line width=.75pt,color=charcoal,decoration={brace,amplitude=5pt},xshift=\xrel,yshift=-\rowspace]	(\xpos,-\yrel+2*\rowspace) -- (\xpos,+2*\rowspace+\yrel) node [black,midway,xshift=9pt] {};
	    
    	%Brackets
	\renewcommand{\xpos}{-.2*\xrel}
	\draw [decorate,line width=.75pt,color=charcoal,decoration={brace,amplitude=5pt},xshift=\xrel,yshift=-\rowspace]	(\xpos,-\yrel-0*\rowspace) -- (\xpos,-0*\rowspace+\yrel) node [black,midway,xshift=9pt] {};
	
	%Brackets
	\renewcommand{\xpos}{-.2*\xrel}
	\draw [decorate,line width=.75pt,color=charcoal,decoration={brace,amplitude=5pt},xshift=\xrel,yshift=-\rowspace]	(\xpos,-\yrel-2*\rowspace) -- (\xpos,-2*\rowspace+\yrel) node [black,midway,xshift=9pt] {};
	
  	\renewcommand{\nodenum}{v2}
    \renewcommand{\name}{$R$}
	\renewcommand{\xpos}{-1.55*\xrel}
    \renewcommand{\ypos}{5*\yrel}
    \renewcommand{\height}{\heightsingle}
    \renewcommand{\width}{\widthsingle}
    \node[rectangle,, line width=.35mm, fill=egg, rounded corners, minimum width=\width, minimum height=\height, draw=dullblue] (\nodenum) at (\xpos,\ypos) {\name}; 
   	
  	\renewcommand{\nodenum}{v2}
    \renewcommand{\name}{Alice Ref}
	\renewcommand{\xpos}{-2.75*\xrel}
    \renewcommand{\ypos}{5*\yrel}
    \renewcommand{\height}{\heightsingle}
    \renewcommand{\width}{\widthsingle}
    \node[] (\nodenum) at (\xpos,\ypos) {\name}; 
    
  	\renewcommand{\nodenum}{v2}
    \renewcommand{\name}{$A$}
	\renewcommand{\xpos}{-1.55*\xrel}
    \renewcommand{\ypos}{3*\yrel}
    \renewcommand{\height}{\heightsingle}
    \renewcommand{\width}{\widthsingle}
    \node[rectangle,, line width=.35mm, fill=egg, rounded corners, minimum width=\width, minimum height=\height, draw=dullblue] (\nodenum) at (\xpos,\ypos) {\name}; 
        
  	\renewcommand{\nodenum}{v2}
    \renewcommand{\name}{Alice}
	\renewcommand{\xpos}{-2.75*\xrel}
    \renewcommand{\ypos}{3*\yrel}
    \renewcommand{\height}{\heightsingle}
    \renewcommand{\width}{\widthsingle}
    \node[] (\nodenum) at (\xpos,\ypos) {\name}; 
    
	\renewcommand{\nodenum}{v3}
    \renewcommand{\name}{$B$}
	\renewcommand{\xpos}{-1.55*\xrel}
    \renewcommand{\ypos}{\yrel}
    \renewcommand{\height}{\heightsingle}
    \renewcommand{\width}{\widthsingle}
    \node[rectangle,, line width=.35mm, fill=egg, rounded corners, minimum width=\width, minimum height=\height, draw=dullblue] (\nodenum) at (\xpos,\ypos) {\name}; 
    
	\renewcommand{\nodenum}{v3}
    \renewcommand{\name}{Black hole}
	\renewcommand{\xpos}{-2.75*\xrel}
    \renewcommand{\ypos}{\yrel}
    \renewcommand{\height}{\heightsingle}
    \renewcommand{\width}{\widthsingle}
    \node[] (\nodenum) at (\xpos,\ypos) {\name}; 
    
	\renewcommand{\nodenum}{v3}
    \renewcommand{\name}{$B'$}
	\renewcommand{\xpos}{-1.55*\xrel}
    \renewcommand{\ypos}{\yrel-\rowspace}
    \renewcommand{\height}{\heightsingle}
    \renewcommand{\width}{\widthsingle}
    \node[rectangle,, line width=.35mm, fill=egg, rounded corners, minimum width=\width, minimum height=\height, draw=dullblue] (\nodenum) at (\xpos,\ypos) {\name}; 
                
	\renewcommand{\nodenum}{v3}
    \renewcommand{\name}{Radiation}
	\renewcommand{\xpos}{-2.75*\xrel}
    \renewcommand{\ypos}{\yrel-\rowspace}
    \renewcommand{\height}{\heightsingle}
    \renewcommand{\width}{\widthsingle}
    \node[] (\nodenum) at (\xpos,\ypos) {\name}; 
    
	\renewcommand{\nodenum}{v3}
    \renewcommand{\name}{$A'$}
	\renewcommand{\xpos}{-1.55*\xrel}
    \renewcommand{\ypos}{\yrel-2*\rowspace}
    \renewcommand{\height}{\heightsingle}
    \renewcommand{\width}{\widthsingle}
    \node[rectangle,, line width=.35mm, fill=egg, rounded corners, minimum width=\width, minimum height=\height, draw=dullblue] (\nodenum) at (\xpos,\ypos) {\name}; 
        
	\renewcommand{\nodenum}{v3}
    \renewcommand{\name}{Bob}
	\renewcommand{\xpos}{-2.75*\xrel}
    \renewcommand{\ypos}{\yrel-2*\rowspace}
    \renewcommand{\height}{\heightsingle}
    \renewcommand{\width}{\widthsingle}
    \node[] (\nodenum) at (\xpos,\ypos) {\name}; 
    
	\renewcommand{\nodenum}{v3}
    \renewcommand{\name}{$R'$}
	\renewcommand{\xpos}{-1.55*\xrel}
    \renewcommand{\ypos}{\yrel-3*\rowspace}
    \renewcommand{\height}{\heightsingle}
    \renewcommand{\width}{\widthsingle}
    \node[rectangle,, line width=.35mm, fill=egg, rounded corners, minimum width=\width, minimum height=\height, draw=dullblue] (\nodenum) at (\xpos,\ypos) {\name}; 
        
	\renewcommand{\nodenum}{v3}
    \renewcommand{\name}{Bob Ref}
	\renewcommand{\xpos}{-2.75*\xrel}
    \renewcommand{\ypos}{\yrel-3*\rowspace}
    \renewcommand{\height}{\heightsingle}
    \renewcommand{\width}{\widthsingle}
    \node[] (\nodenum) at (\xpos,\ypos) {\name}; 
    
	\renewcommand{\nodenum}{v4}
    \renewcommand{\name}{ $\ket{\mathrm{Bell}}$}
	\renewcommand{\xpos}{-.75*\xrel}
    \renewcommand{\ypos}{4*\yrel}
    \renewcommand{\height}{\heightsingle}
    \renewcommand{\width}{\widthsingle}
    \node[](\nodenum) at (\xpos,\ypos) {\name};
        
	\renewcommand{\nodenum}{v4}
    \renewcommand{\name}{ $\ket{\mathrm{Bell}}$}
	\renewcommand{\xpos}{-.75*\xrel}
    \renewcommand{\ypos}{0*\yrel}
    \renewcommand{\height}{\heightsingle}
    \renewcommand{\width}{\widthsingle}
    \node[](\nodenum) at (\xpos,\ypos) {\name};
    
	\renewcommand{\nodenum}{v5}
    \renewcommand{\name}{$\ket{\mathrm{Bell}}$}
	\renewcommand{\xpos}{-.75*\xrel}
    \renewcommand{\ypos}{-2*\rowspace}
    \renewcommand{\height}{9.7em}
    \renewcommand{\width}{\widthsingle}
    \node[](\nodenum) at (\xpos,\ypos) {\name};  
	
	\renewcommand{\nodenum}{v6}
    \renewcommand{\name}{$U_{\mathrm{bh}}$}
	\renewcommand{\xpos}{\xrel}
    \renewcommand{\ypos}{\rowspace}
    \renewcommand{\height}{\heightdouble}
    \renewcommand{\width}{\widthsingle}
    \node[rectangle, fill=white,  line width =.3mm,rounded corners, minimum width=\width, minimum height=  	    \height, draw=charcoal] (\nodenum) at (\xpos,\ypos) {\name};
    
    \renewcommand{\nodenum}{v7}
    \renewcommand{\name}{$U_{\mathrm{bh}}^*$}
	\renewcommand{\xpos}{\xrel}
    \renewcommand{\ypos}{-\rowspace}
    \renewcommand{\height}{\heightdouble}
    \renewcommand{\width}{\widthsingle}
    \node[rectangle, fill=white,  line width =.3mm,rounded corners, minimum width=\width, minimum height=  	    \height, draw=charcoal] (\nodenum) at (\xpos,\ypos) {\name};
    
    \renewcommand{\nodenum}{v15}
    \renewcommand{\name}{\includegraphics[width=.025\textwidth]{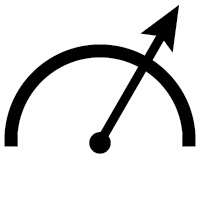}}
	\renewcommand{\xpos}{2.5*\xrel}
    \renewcommand{\ypos}{0}
    \renewcommand{\height}{\heightdouble}
    \renewcommand{\width}{\widthsingle}
    \node[rectangle, fill=white,  line width =.3mm,rounded corners, minimum width=\width, minimum height=  	    \height, draw=charcoal] (\nodenum) at (\xpos,\ypos) {\name};

	\renewcommand{\nodenum}{v17}
    \renewcommand{\name}{$R$}
	\renewcommand{\xpos}{4*\xrel}
    \renewcommand{\ypos}{\yrel+2*\rowspace}
    \renewcommand{\height}{\heightsingle}
    \renewcommand{\width}{\widthsingle}
    \node[rectangle,, line width=.35mm, fill=egg, rounded corners, minimum width=\width, minimum height=\height, draw=dullblue] (\nodenum) at (\xpos,\ypos) {\name};
    
	\renewcommand{\nodenum}{v17}
    \renewcommand{\name}{$C$}
	\renewcommand{\xpos}{4*\xrel}
    \renewcommand{\ypos}{\yrel+\rowspace}
    \renewcommand{\height}{\heightsingle}
    \renewcommand{\width}{\widthsingle}
    \node[rectangle,, line width=.35mm, fill=egg, rounded corners, minimum width=\width, minimum height=\height, draw=dullblue] (\nodenum) at (\xpos,\ypos) {\name};
    
	\renewcommand{\nodenum}{v17}
    \renewcommand{\name}{$D$}
	\renewcommand{\xpos}{4*\xrel}
    \renewcommand{\ypos}{\yrel}
    \renewcommand{\height}{\heightsingle}
    \renewcommand{\width}{\widthsingle}
    \node[rectangle,, line width=.35mm, fill=egg, rounded corners, minimum width=\width, minimum height=\height, draw=dullblue] (\nodenum) at (\xpos,\ypos) {\name};
        
	\renewcommand{\nodenum}{v17}
    \renewcommand{\name}{$D'$}
	\renewcommand{\xpos}{4*\xrel}
    \renewcommand{\ypos}{\yrel-\rowspace}
    \renewcommand{\height}{\heightsingle}
    \renewcommand{\width}{\widthsingle}
    \node[rectangle,, line width=.35mm, fill=egg, rounded corners, minimum width=\width, minimum height=\height, draw=dullblue] (\nodenum) at (\xpos,\ypos) {\name};
        
	\renewcommand{\nodenum}{v17}
    \renewcommand{\name}{$C'$}
	\renewcommand{\xpos}{4*\xrel}
    \renewcommand{\ypos}{\yrel-2*\rowspace}
    \renewcommand{\height}{\heightsingle}
    \renewcommand{\width}{\widthsingle}
    \node[rectangle,, line width=.35mm, fill=egg, rounded corners, minimum width=\width, minimum height=\height, draw=dullblue] (\nodenum) at (\xpos,\ypos) {\name};
            
	\renewcommand{\nodenum}{v17}
    \renewcommand{\name}{$R'$}
	\renewcommand{\xpos}{4*\xrel}
    \renewcommand{\ypos}{\yrel-3*\rowspace}
    \renewcommand{\height}{\heightsingle}
    \renewcommand{\width}{\widthsingle}
    \node[rectangle,, line width=.35mm, fill=egg, rounded corners, minimum width=\width, minimum height=\height, draw=dullblue] (\nodenum) at (\xpos,\ypos) {\name};

\end{tikzpicture}
}
\caption{Quantum circuit for the Hayden-Preskill decoding protocol. Alice maximally entangles her state, defined on the system $A$, with a reference system $R$. Bob maximally entangles his system $A'$ with a reference system $R'$. Alice throws her state into the black hole, system $B$. Afterwards, the new black hole, system $AB$, evolves under the unitary $U_{\mathrm{bh}}$. Bob applies the unitary $U_{\mathrm{bh}}^*$ to his system $B'A'$, where $B'$ represents the Hawking radiation emitted by the black hole before Alice threw her state in. Bob then projectively measures a Bell state, $\ket{\mathrm{Bell}}_{D,D'}=\frac{1}{\sqrt{2^n}}\sum_{i=1}^{2^n}\ket{i}_D\otimes \ket{i}_{D'}$, between systems $D$ and $D'$. System $D$ represents the Hawking radiation emitted after Alice throws her state in. If a successful decoding occurs, then Bob's reference system $R'$ forms a Bell state with $R$.}
\label{Fig:HP}
\end{figure}

We apply the resource theory of scrambling to the Hayden-Preskill thought experiment \cite{Hayden_2007}, an information recovery problem. In this thought experiment, a quantum state is thrown into an $n$-qubit black hole. The black hole's scrambling dynamics, $U_{\mathrm{bh}}$, lead to delocalization of the state's information. By collecting the emitted Hawking radiation, one can decode the state thrown in. We show that the Pauli growth and the OTOC magic bound the decoding fidelity.

Yoshida et al. \cite{YoshidaEfficient, PhysRevX.9.011006,Yoshida2021Recovery} constructed a teleportation-based decoding protocol to recover the input state thrown into the black hole with a decoding fidelity of $F(U_{\mathrm{bh}})$ (see Fig.~\ref{Fig:HP} for a description). The decoding fidelity is determined by the average OTOC:
\begin{equation}\label{Eq:OTOCFidelity}
	\av{P_A, P_D}\OTOC(U_{\mathrm{bh}})=\frac{1}{d_A^2 F(U_{\mathrm{bh}})},
\end{equation}
where $\OTOC(U_{\mathrm{bh}})=\langle U_{\mathrm{bh}}^\dagger P_D U_{\mathrm{bh}} P_A U_{\mathrm{bh}}^\dagger P_D U_{\mathrm{bh}} P_A\rangle $ and ${d_A=2^{n_A}}$ is the Hilbert space dimension of system $A$. The average is taken uniformly over all Pauli operators on the systems $A$ and $D$, defined in Fig.~\ref{Fig:HP}. Using Eq.~\eqref{Eq:OTOCFidelity}, we show that the OTOC magic bounds the decoding fidelity.

\begin{theorem}\label{Theorem:HP}
If 
$\av{\substack{P_A,\\ P_D}}\abs{\OTOC(U_{\mathrm{bh}})}=\av{\substack{P_A,\\ P_D}}\OTOC(U_{\mathrm{bh}}) +\eta$,
where $\eta>0$, then the decoding fidelity $F(U_{\mathrm{bh}})$ is bounded by the OTOC magic of $U_{\mathrm{bh}}$:

\begin{equation}
	F(U_{\mathrm{bh}})\leq \frac{1}{d_A^2(1-O_M(U_{\mathrm{bh}})-\eta)}.
\end{equation}

\end{theorem}
We prove this result in Appendix~\ref{Sec:ProofThmHP}. 
%The approximation in this theorem does not hold if $U_{\mathrm{bh}}$ is Clifford, since ${\OTOC(U_{\mathrm{bh}})=\pm 1}$ in this case. The upper-bound in this theorem is non-trivial in the regime where ${O_M(U_{\mathrm{bh}})<1-\frac{1}{d_A^2}}-\eta$.
%As the OTOC magic of $U_{\mathrm{bh}}$ decreases, the upper bound on the decoding fidelity decreases. 
We conjecture that the OTOC magic can also bound the fluctuations of $F(U_{\mathrm{bh}})$ when $U_{\mathrm{bh}}$ is sampled from a unitary ensemble. 
The effect of magic in a related black hole learning problem has also been investigated \cite{Leone2022Mocking}. We also establish an inequality relating the decoding fidelity and the Pauli growth of $U_{\mathrm{bh}}$ (see Appendix~\ref{Sec:ProofTheoremHPPAuliGrowth} for a proof).

\begin{theorem}\label{Theorem:HPPauliGrowth}
Let $D$ be the $n$-th qubit and let $A$ be any other single-qubit system. Let $\av{A}$ denote the uniform average over all such $A$ systems. In the large $n$ limit, the decoding fidelity $F(U_{\mathrm{bh}})$ and the Pauli growth of $U_{\mathrm{bh}}$ satisfy the following inequality:
\begin{equation}
\begin{split}
\av{A}\frac{1}{F(U_{\mathrm{bh}})}\geq (d_A^2-1)\left[1-\frac{4}{3n}\left(G(U_{\mathrm{bh}})+1\right)\right]+1.
\end{split}
\end{equation}
\end{theorem}
More generally, $D$ can be a randomly selected, single-qubit system.

\section{Conclusion}
We have introduced a resource theory of scrambling comprised of two mechanisms. In the entanglement scrambling mechanism, we introduce the Pauli growth as a monotone to measure operator spreading.
%which refers to the growth of the support of a local operator via unitary evolution.
In the magic scrambling mechanism, we introduce the OTOC magic as a monotone to measure the generation of operator entanglement.
%, which occurs when a unitary maps a Pauli operator to a sum of many Pauli operators.
We use these monotones to bound the OTOC fluctuations measured in Google's experiment~\cite{Mi_2021} and to bound the success of the Hayden-Preskill decoding protocol. These applications provide an operational interpretation of our resource monotones. 

We propose that these monotones may also be used to bound other scrambling tools, such as the tripartite mutual information, the operator entanglement entropy, and Lyapunov exponents in OTOC dynamics. 
Furthermore, we conjecture that this scrambling resource theory can be generalized to quantum channels, which may be useful in understanding noise effects. 

\acknowledgements{The authors thank You Zhou and Aram Harrow for insightful discussions. We are grateful to Renato Renner for valuable comments and suggestions on an earlier draft of this paper. This work was supported in part by ARO Grant W911NF-19-1-0302 and ARO MURI Grant W911NF-20-1-0082, and NSF Eager Grant 2037687.}

\bibliography{Bib}

%apsrev4-2.bst 2019-01-14 (MD) hand-edited version of apsrev4-1.bst
%Control: key (0)
%Control: author (72) initials jnrlst
%Control: editor formatted (1) identically to author
%Control: production of article title (-1) disabled
%Control: page (0) single
%Control: year (1) truncated
%Control: production of eprint (0) enabled
\begin{thebibliography}{94}%
\makeatletter
\providecommand \@ifxundefined [1]{%
 \@ifx{#1\undefined}
}%
\providecommand \@ifnum [1]{%
 \ifnum #1\expandafter \@firstoftwo
 \else \expandafter \@secondoftwo
 \fi
}%
\providecommand \@ifx [1]{%
 \ifx #1\expandafter \@firstoftwo
 \else \expandafter \@secondoftwo
 \fi
}%
\providecommand \natexlab [1]{#1}%
\providecommand \enquote  [1]{``#1''}%
\providecommand \bibnamefont  [1]{#1}%
\providecommand \bibfnamefont [1]{#1}%
\providecommand \citenamefont [1]{#1}%
\providecommand \href@noop [0]{\@secondoftwo}%
\providecommand \href [0]{\begingroup \@sanitize@url \@href}%
\providecommand \@href[1]{\@@startlink{#1}\@@href}%
\providecommand \@@href[1]{\endgroup#1\@@endlink}%
\providecommand \@sanitize@url [0]{\catcode `\\12\catcode `\$12\catcode
  `\&12\catcode `\#12\catcode `\^12\catcode `\_12\catcode `\%12\relax}%
\providecommand \@@startlink[1]{}%
\providecommand \@@endlink[0]{}%
\providecommand \url  [0]{\begingroup\@sanitize@url \@url }%
\providecommand \@url [1]{\endgroup\@href {#1}{\urlprefix }}%
\providecommand \urlprefix  [0]{URL }%
\providecommand \Eprint [0]{\href }%
\providecommand \doibase [0]{https://doi.org/}%
\providecommand \selectlanguage [0]{\@gobble}%
\providecommand \bibinfo  [0]{\@secondoftwo}%
\providecommand \bibfield  [0]{\@secondoftwo}%
\providecommand \translation [1]{[#1]}%
\providecommand \BibitemOpen [0]{}%
\providecommand \bibitemStop [0]{}%
\providecommand \bibitemNoStop [0]{.\EOS\space}%
\providecommand \EOS [0]{\spacefactor3000\relax}%
\providecommand \BibitemShut  [1]{\csname bibitem#1\endcsname}%
\let\auto@bib@innerbib\@empty
%</preamble>
\bibitem [{\citenamefont {Sachdev}\ and\ \citenamefont
  {Ye}(1993)}]{PhysRevLett.70.3339}%
  \BibitemOpen
  \bibfield  {author} {\bibinfo {author} {\bibfnamefont {S.}~\bibnamefont
  {Sachdev}}\ and\ \bibinfo {author} {\bibfnamefont {J.}~\bibnamefont {Ye}},\
  }\href {https://doi.org/10.1103/PhysRevLett.70.3339} {\bibfield  {journal}
  {\bibinfo  {journal} {Phys. Rev. Lett.}\ }\textbf {\bibinfo {volume} {70}},\
  \bibinfo {pages} {3339} (\bibinfo {year} {1993})}\BibitemShut {NoStop}%
\bibitem [{\citenamefont {Kitaev}(2015)}]{Kitaev2015}%
  \BibitemOpen
  \bibfield  {author} {\bibinfo {author} {\bibfnamefont {A.}~\bibnamefont
  {Kitaev}},\ }\href {https://online.kitp.ucsb.edu/online/entangled15/}
  {\bibinfo {title} {A simple model of quantum holography}} (\bibinfo {year}
  {2015})\BibitemShut {NoStop}%
\bibitem [{\citenamefont {Nandkishore}\ and\ \citenamefont
  {Huse}(2015)}]{Nandkishore_2015}%
  \BibitemOpen
  \bibfield  {author} {\bibinfo {author} {\bibfnamefont {R.}~\bibnamefont
  {Nandkishore}}\ and\ \bibinfo {author} {\bibfnamefont {D.~A.}\ \bibnamefont
  {Huse}},\ }\href {https://doi.org/10.1146/annurev-conmatphys-031214-014726}
  {\bibfield  {journal} {\bibinfo  {journal} {Annual Review of Condensed Matter
  Physics}\ }\textbf {\bibinfo {volume} {6}},\ \bibinfo {pages} {15} (\bibinfo
  {year} {2015})}\BibitemShut {NoStop}%
\bibitem [{\citenamefont {Hayden}\ and\ \citenamefont
  {Preskill}(2007)}]{Hayden_2007}%
  \BibitemOpen
  \bibfield  {author} {\bibinfo {author} {\bibfnamefont {P.}~\bibnamefont
  {Hayden}}\ and\ \bibinfo {author} {\bibfnamefont {J.}~\bibnamefont
  {Preskill}},\ }\href {https://doi.org/10.1088/1126-6708/2007/09/120}
  {\bibfield  {journal} {\bibinfo  {journal} {Journal of High Energy Physics}\
  }\textbf {\bibinfo {volume} {2007}},\ \bibinfo {pages} {120} (\bibinfo {year}
  {2007})}\BibitemShut {NoStop}%
\bibitem [{\citenamefont {Sekino}\ and\ \citenamefont
  {Susskind}(2008)}]{Sekino_2008}%
  \BibitemOpen
  \bibfield  {author} {\bibinfo {author} {\bibfnamefont {Y.}~\bibnamefont
  {Sekino}}\ and\ \bibinfo {author} {\bibfnamefont {L.}~\bibnamefont
  {Susskind}},\ }\href {https://doi.org/10.1088/1126-6708/2008/10/065}
  {\bibfield  {journal} {\bibinfo  {journal} {Journal of High Energy Physics}\
  }\textbf {\bibinfo {volume} {2008}},\ \bibinfo {pages} {065} (\bibinfo {year}
  {2008})}\BibitemShut {NoStop}%
\bibitem [{\citenamefont {Lashkari}\ \emph {et~al.}(2013)\citenamefont
  {Lashkari}, \citenamefont {Stanford}, \citenamefont {Hastings}, \citenamefont
  {Osborne},\ and\ \citenamefont {Hayden}}]{Lashkari_2013}%
  \BibitemOpen
  \bibfield  {author} {\bibinfo {author} {\bibfnamefont {N.}~\bibnamefont
  {Lashkari}}, \bibinfo {author} {\bibfnamefont {D.}~\bibnamefont {Stanford}},
  \bibinfo {author} {\bibfnamefont {M.}~\bibnamefont {Hastings}}, \bibinfo
  {author} {\bibfnamefont {T.}~\bibnamefont {Osborne}},\ and\ \bibinfo {author}
  {\bibfnamefont {P.}~\bibnamefont {Hayden}},\ }\bibfield  {journal} {\bibinfo
  {journal} {Journal of High Energy Physics}\ }\textbf {\bibinfo {volume}
  {2013}},\ \href {https://doi.org/10.1007/jhep04(2013)022}
  {10.1007/jhep04(2013)022} (\bibinfo {year} {2013})\BibitemShut {NoStop}%
\bibitem [{\citenamefont {Shenker}\ and\ \citenamefont
  {Stanford}(2014)}]{Shenker_2014}%
  \BibitemOpen
  \bibfield  {author} {\bibinfo {author} {\bibfnamefont {S.~H.}\ \bibnamefont
  {Shenker}}\ and\ \bibinfo {author} {\bibfnamefont {D.}~\bibnamefont
  {Stanford}},\ }\bibfield  {journal} {\bibinfo  {journal} {Journal of High
  Energy Physics}\ }\textbf {\bibinfo {volume} {2014}},\ \href
  {https://doi.org/10.1007/jhep03(2014)067} {10.1007/jhep03(2014)067} (\bibinfo
  {year} {2014})\BibitemShut {NoStop}%
\bibitem [{\citenamefont {Maldacena}\ \emph {et~al.}(2016)\citenamefont
  {Maldacena}, \citenamefont {Shenker},\ and\ \citenamefont
  {Stanford}}]{Maldacena_2016}%
  \BibitemOpen
  \bibfield  {author} {\bibinfo {author} {\bibfnamefont {J.}~\bibnamefont
  {Maldacena}}, \bibinfo {author} {\bibfnamefont {S.~H.}\ \bibnamefont
  {Shenker}},\ and\ \bibinfo {author} {\bibfnamefont {D.}~\bibnamefont
  {Stanford}},\ }\bibfield  {journal} {\bibinfo  {journal} {Journal of High
  Energy Physics}\ }\textbf {\bibinfo {volume} {2016}},\ \href
  {https://doi.org/10.1007/jhep08(2016)106} {10.1007/jhep08(2016)106} (\bibinfo
  {year} {2016})\BibitemShut {NoStop}%
\bibitem [{\citenamefont {Holmes}\ \emph {et~al.}(2021)\citenamefont {Holmes},
  \citenamefont {Arrasmith}, \citenamefont {Yan}, \citenamefont {Coles},
  \citenamefont {Albrecht},\ and\ \citenamefont {Sornborger}}]{Holmes_2021}%
  \BibitemOpen
  \bibfield  {author} {\bibinfo {author} {\bibfnamefont {Z.}~\bibnamefont
  {Holmes}}, \bibinfo {author} {\bibfnamefont {A.}~\bibnamefont {Arrasmith}},
  \bibinfo {author} {\bibfnamefont {B.}~\bibnamefont {Yan}}, \bibinfo {author}
  {\bibfnamefont {P.~J.}\ \bibnamefont {Coles}}, \bibinfo {author}
  {\bibfnamefont {A.}~\bibnamefont {Albrecht}},\ and\ \bibinfo {author}
  {\bibfnamefont {A.~T.}\ \bibnamefont {Sornborger}},\ }\bibfield  {journal}
  {\bibinfo  {journal} {Physical Review Letters}\ }\textbf {\bibinfo {volume}
  {126}},\ \href {https://doi.org/10.1103/physrevlett.126.190501}
  {10.1103/physrevlett.126.190501} (\bibinfo {year} {2021})\BibitemShut
  {NoStop}%
\bibitem [{\citenamefont {Garcia}\ \emph
  {et~al.}(2022{\natexlab{a}})\citenamefont {Garcia}, \citenamefont {Zhao},
  \citenamefont {Bu},\ and\ \citenamefont {Jaffe}}]{Garcia2022Barren}%
  \BibitemOpen
  \bibfield  {author} {\bibinfo {author} {\bibfnamefont {R.~J.}\ \bibnamefont
  {Garcia}}, \bibinfo {author} {\bibfnamefont {C.}~\bibnamefont {Zhao}},
  \bibinfo {author} {\bibfnamefont {K.}~\bibnamefont {Bu}},\ and\ \bibinfo
  {author} {\bibfnamefont {A.}~\bibnamefont {Jaffe}},\ }\href
  {https://doi.org/10.48550/ARXIV.2205.06679} {\bibinfo {title} {Barren
  plateaus from learning scramblers with local cost functions}} (\bibinfo
  {year} {2022}{\natexlab{a}})\BibitemShut {NoStop}%
\bibitem [{\citenamefont {Shen}\ \emph {et~al.}(2020)\citenamefont {Shen},
  \citenamefont {Zhang}, \citenamefont {You},\ and\ \citenamefont
  {Zhai}}]{PhysRevLett.124.200504}%
  \BibitemOpen
  \bibfield  {author} {\bibinfo {author} {\bibfnamefont {H.}~\bibnamefont
  {Shen}}, \bibinfo {author} {\bibfnamefont {P.}~\bibnamefont {Zhang}},
  \bibinfo {author} {\bibfnamefont {Y.-Z.}\ \bibnamefont {You}},\ and\ \bibinfo
  {author} {\bibfnamefont {H.}~\bibnamefont {Zhai}},\ }\href
  {https://doi.org/10.1103/PhysRevLett.124.200504} {\bibfield  {journal}
  {\bibinfo  {journal} {Phys. Rev. Lett.}\ }\textbf {\bibinfo {volume} {124}},\
  \bibinfo {pages} {200504} (\bibinfo {year} {2020})}\BibitemShut {NoStop}%
\bibitem [{\citenamefont {Wu}\ \emph {et~al.}(2021)\citenamefont {Wu},
  \citenamefont {Zhang},\ and\ \citenamefont
  {Zhai}}]{PhysRevResearch.3.L032057}%
  \BibitemOpen
  \bibfield  {author} {\bibinfo {author} {\bibfnamefont {Y.}~\bibnamefont
  {Wu}}, \bibinfo {author} {\bibfnamefont {P.}~\bibnamefont {Zhang}},\ and\
  \bibinfo {author} {\bibfnamefont {H.}~\bibnamefont {Zhai}},\ }\href
  {https://doi.org/10.1103/PhysRevResearch.3.L032057} {\bibfield  {journal}
  {\bibinfo  {journal} {Phys. Rev. Research}\ }\textbf {\bibinfo {volume}
  {3}},\ \bibinfo {pages} {L032057} (\bibinfo {year} {2021})}\BibitemShut
  {NoStop}%
\bibitem [{\citenamefont {Garcia}\ \emph
  {et~al.}(2022{\natexlab{b}})\citenamefont {Garcia}, \citenamefont {Bu},\ and\
  \citenamefont {Jaffe}}]{Garcia_2022}%
  \BibitemOpen
  \bibfield  {author} {\bibinfo {author} {\bibfnamefont {R.~J.}\ \bibnamefont
  {Garcia}}, \bibinfo {author} {\bibfnamefont {K.}~\bibnamefont {Bu}},\ and\
  \bibinfo {author} {\bibfnamefont {A.}~\bibnamefont {Jaffe}},\ }\bibfield
  {journal} {\bibinfo  {journal} {Journal of High Energy Physics}\ }\textbf
  {\bibinfo {volume} {2022}},\ \href {https://doi.org/10.1007/jhep03(2022)027}
  {10.1007/jhep03(2022)027} (\bibinfo {year} {2022}{\natexlab{b}})\BibitemShut
  {NoStop}%
\bibitem [{\citenamefont {Huang}\ \emph {et~al.}(2020)\citenamefont {Huang},
  \citenamefont {Kueng},\ and\ \citenamefont {Preskill}}]{Huang_2020}%
  \BibitemOpen
  \bibfield  {author} {\bibinfo {author} {\bibfnamefont {H.-Y.}\ \bibnamefont
  {Huang}}, \bibinfo {author} {\bibfnamefont {R.}~\bibnamefont {Kueng}},\ and\
  \bibinfo {author} {\bibfnamefont {J.}~\bibnamefont {Preskill}},\ }\href
  {https://doi.org/10.1038/s41567-020-0932-7} {\bibfield  {journal} {\bibinfo
  {journal} {Nature Physics}\ }\textbf {\bibinfo {volume} {16}},\ \bibinfo
  {pages} {1050} (\bibinfo {year} {2020})}\BibitemShut {NoStop}%
\bibitem [{\citenamefont {Hu}\ and\ \citenamefont
  {You}(2022)}]{PhysRevResearch.4.013054}%
  \BibitemOpen
  \bibfield  {author} {\bibinfo {author} {\bibfnamefont {H.-Y.}\ \bibnamefont
  {Hu}}\ and\ \bibinfo {author} {\bibfnamefont {Y.-Z.}\ \bibnamefont {You}},\
  }\href {https://doi.org/10.1103/PhysRevResearch.4.013054} {\bibfield
  {journal} {\bibinfo  {journal} {Phys. Rev. Research}\ }\textbf {\bibinfo
  {volume} {4}},\ \bibinfo {pages} {013054} (\bibinfo {year}
  {2022})}\BibitemShut {NoStop}%
\bibitem [{\citenamefont {Hu}\ \emph {et~al.}(2021)\citenamefont {Hu},
  \citenamefont {Choi},\ and\ \citenamefont {You}}]{Hu2021}%
  \BibitemOpen
  \bibfield  {author} {\bibinfo {author} {\bibfnamefont {H.-Y.}\ \bibnamefont
  {Hu}}, \bibinfo {author} {\bibfnamefont {S.}~\bibnamefont {Choi}},\ and\
  \bibinfo {author} {\bibfnamefont {Y.-Z.}\ \bibnamefont {You}},\ }\href
  {https://doi.org/10.48550/ARXIV.2107.04817} {\bibinfo {title} {Classical
  shadow tomography with locally scrambled quantum dynamics}} (\bibinfo {year}
  {2021})\BibitemShut {NoStop}%
\bibitem [{\citenamefont {Bu}\ \emph {et~al.}(2022{\natexlab{a}})\citenamefont
  {Bu}, \citenamefont {Koh}, \citenamefont {Garcia},\ and\ \citenamefont
  {Jaffe}}]{Bu2022Classical}%
  \BibitemOpen
  \bibfield  {author} {\bibinfo {author} {\bibfnamefont {K.}~\bibnamefont
  {Bu}}, \bibinfo {author} {\bibfnamefont {D.~E.}\ \bibnamefont {Koh}},
  \bibinfo {author} {\bibfnamefont {R.~J.}\ \bibnamefont {Garcia}},\ and\
  \bibinfo {author} {\bibfnamefont {A.}~\bibnamefont {Jaffe}},\ }\href
  {https://doi.org/10.48550/ARXIV.2202.03272} {\bibinfo {title} {Classical
  shadows with pauli-invariant unitary ensembles}} (\bibinfo {year}
  {2022}{\natexlab{a}})\BibitemShut {NoStop}%
\bibitem [{\citenamefont {Garcia}\ \emph {et~al.}(2021)\citenamefont {Garcia},
  \citenamefont {Zhou},\ and\ \citenamefont
  {Jaffe}}]{PhysRevResearch.3.033155}%
  \BibitemOpen
  \bibfield  {author} {\bibinfo {author} {\bibfnamefont {R.~J.}\ \bibnamefont
  {Garcia}}, \bibinfo {author} {\bibfnamefont {Y.}~\bibnamefont {Zhou}},\ and\
  \bibinfo {author} {\bibfnamefont {A.}~\bibnamefont {Jaffe}},\ }\href
  {https://doi.org/10.1103/PhysRevResearch.3.033155} {\bibfield  {journal}
  {\bibinfo  {journal} {Phys. Rev. Research}\ }\textbf {\bibinfo {volume}
  {3}},\ \bibinfo {pages} {033155} (\bibinfo {year} {2021})}\BibitemShut
  {NoStop}%
\bibitem [{\citenamefont {McGinley}\ \emph {et~al.}(2022)\citenamefont
  {McGinley}, \citenamefont {Leontica}, \citenamefont {Garratt}, \citenamefont
  {Jovanovic},\ and\ \citenamefont {Simon}}]{McGinley_2022}%
  \BibitemOpen
  \bibfield  {author} {\bibinfo {author} {\bibfnamefont {M.}~\bibnamefont
  {McGinley}}, \bibinfo {author} {\bibfnamefont {S.}~\bibnamefont {Leontica}},
  \bibinfo {author} {\bibfnamefont {S.~J.}\ \bibnamefont {Garratt}}, \bibinfo
  {author} {\bibfnamefont {J.}~\bibnamefont {Jovanovic}},\ and\ \bibinfo
  {author} {\bibfnamefont {S.~H.}\ \bibnamefont {Simon}},\ }\bibfield
  {journal} {\bibinfo  {journal} {Physical Review A}\ }\textbf {\bibinfo
  {volume} {106}},\ \href {https://doi.org/10.1103/physreva.106.012441}
  {10.1103/physreva.106.012441} (\bibinfo {year} {2022})\BibitemShut {NoStop}%
\bibitem [{\citenamefont {Brown}\ and\ \citenamefont
  {Fawzi}(2012)}]{Brown2012}%
  \BibitemOpen
  \bibfield  {author} {\bibinfo {author} {\bibfnamefont {W.}~\bibnamefont
  {Brown}}\ and\ \bibinfo {author} {\bibfnamefont {O.}~\bibnamefont {Fawzi}},\
  }\href {https://doi.org/10.48550/ARXIV.1210.6644} {\bibinfo {title}
  {Scrambling speed of random quantum circuits}} (\bibinfo {year}
  {2012})\BibitemShut {NoStop}%
\bibitem [{\citenamefont {Choi}\ \emph {et~al.}(2020)\citenamefont {Choi},
  \citenamefont {Bao}, \citenamefont {Qi},\ and\ \citenamefont
  {Altman}}]{Choi_2020}%
  \BibitemOpen
  \bibfield  {author} {\bibinfo {author} {\bibfnamefont {S.}~\bibnamefont
  {Choi}}, \bibinfo {author} {\bibfnamefont {Y.}~\bibnamefont {Bao}}, \bibinfo
  {author} {\bibfnamefont {X.-L.}\ \bibnamefont {Qi}},\ and\ \bibinfo {author}
  {\bibfnamefont {E.}~\bibnamefont {Altman}},\ }\bibfield  {journal} {\bibinfo
  {journal} {Physical Review Letters}\ }\textbf {\bibinfo {volume} {125}},\
  \href {https://doi.org/10.1103/physrevlett.125.030505}
  {10.1103/physrevlett.125.030505} (\bibinfo {year} {2020})\BibitemShut
  {NoStop}%
\bibitem [{\citenamefont {Chamon}\ \emph {et~al.}(2020)\citenamefont {Chamon},
  \citenamefont {Mucciolo},\ and\ \citenamefont {Ruckenstein}}]{Ruck2020}%
  \BibitemOpen
  \bibfield  {author} {\bibinfo {author} {\bibfnamefont {C.}~\bibnamefont
  {Chamon}}, \bibinfo {author} {\bibfnamefont {E.}~\bibnamefont {Mucciolo}},\
  and\ \bibinfo {author} {\bibfnamefont {A.}~\bibnamefont {Ruckenstein}},\
  }\href {https://arxiv.org/abs/2011.06546} {\bibinfo {title} {Quantum
  statistical mechanics of encryption: reaching the speed limit of classical
  block ciphers}} (\bibinfo {year} {2020})\BibitemShut {NoStop}%
\bibitem [{\citenamefont {Ho}\ and\ \citenamefont
  {Choi}(2022)}]{PhysRevLett.128.060601}%
  \BibitemOpen
  \bibfield  {author} {\bibinfo {author} {\bibfnamefont {W.~W.}\ \bibnamefont
  {Ho}}\ and\ \bibinfo {author} {\bibfnamefont {S.}~\bibnamefont {Choi}},\
  }\href {https://doi.org/10.1103/PhysRevLett.128.060601} {\bibfield  {journal}
  {\bibinfo  {journal} {Phys. Rev. Lett.}\ }\textbf {\bibinfo {volume} {128}},\
  \bibinfo {pages} {060601} (\bibinfo {year} {2022})}\BibitemShut {NoStop}%
\bibitem [{\citenamefont {Chen}\ and\ \citenamefont
  {Zhou}(2018)}]{Chen2018Strongly}%
  \BibitemOpen
  \bibfield  {author} {\bibinfo {author} {\bibfnamefont {X.}~\bibnamefont
  {Chen}}\ and\ \bibinfo {author} {\bibfnamefont {T.}~\bibnamefont {Zhou}},\
  }\href {https://doi.org/10.48550/ARXIV.1804.08655} {\bibinfo {title}
  {Operator scrambling and quantum chaos}} (\bibinfo {year} {2018})\BibitemShut
  {NoStop}%
\bibitem [{\citenamefont {Zhou}\ and\ \citenamefont
  {Chen}(2019)}]{PhysRevE.99.052212}%
  \BibitemOpen
  \bibfield  {author} {\bibinfo {author} {\bibfnamefont {T.}~\bibnamefont
  {Zhou}}\ and\ \bibinfo {author} {\bibfnamefont {X.}~\bibnamefont {Chen}},\
  }\href {https://doi.org/10.1103/PhysRevE.99.052212} {\bibfield  {journal}
  {\bibinfo  {journal} {Phys. Rev. E}\ }\textbf {\bibinfo {volume} {99}},\
  \bibinfo {pages} {052212} (\bibinfo {year} {2019})}\BibitemShut {NoStop}%
\bibitem [{\citenamefont {Khemani}\ \emph {et~al.}(2018)\citenamefont
  {Khemani}, \citenamefont {Vishwanath},\ and\ \citenamefont
  {Huse}}]{PhysRevX.8.031057}%
  \BibitemOpen
  \bibfield  {author} {\bibinfo {author} {\bibfnamefont {V.}~\bibnamefont
  {Khemani}}, \bibinfo {author} {\bibfnamefont {A.}~\bibnamefont
  {Vishwanath}},\ and\ \bibinfo {author} {\bibfnamefont {D.~A.}\ \bibnamefont
  {Huse}},\ }\href {https://doi.org/10.1103/PhysRevX.8.031057} {\bibfield
  {journal} {\bibinfo  {journal} {Phys. Rev. X}\ }\textbf {\bibinfo {volume}
  {8}},\ \bibinfo {pages} {031057} (\bibinfo {year} {2018})}\BibitemShut
  {NoStop}%
\bibitem [{\citenamefont {Aleiner}\ \emph {et~al.}(2016)\citenamefont
  {Aleiner}, \citenamefont {Faoro},\ and\ \citenamefont
  {Ioffe}}]{Aleiner_2016}%
  \BibitemOpen
  \bibfield  {author} {\bibinfo {author} {\bibfnamefont {I.~L.}\ \bibnamefont
  {Aleiner}}, \bibinfo {author} {\bibfnamefont {L.}~\bibnamefont {Faoro}},\
  and\ \bibinfo {author} {\bibfnamefont {L.~B.}\ \bibnamefont {Ioffe}},\ }\href
  {https://doi.org/10.1016/j.aop.2016.09.006} {\bibfield  {journal} {\bibinfo
  {journal} {Annals of Physics}\ }\textbf {\bibinfo {volume} {375}},\ \bibinfo
  {pages} {378} (\bibinfo {year} {2016})}\BibitemShut {NoStop}%
\bibitem [{\citenamefont {{Larkin}}\ and\ \citenamefont
  {{Ovchinnikov}}(1969)}]{Larkin1969}%
  \BibitemOpen
  \bibfield  {author} {\bibinfo {author} {\bibfnamefont {A.~I.}\ \bibnamefont
  {{Larkin}}}\ and\ \bibinfo {author} {\bibfnamefont {Y.~N.}\ \bibnamefont
  {{Ovchinnikov}}},\ }\href@noop {} {\bibfield  {journal} {\bibinfo  {journal}
  {Soviet Journal of Experimental and Theoretical Physics}\ }\textbf {\bibinfo
  {volume} {28}},\ \bibinfo {pages} {1200} (\bibinfo {year}
  {1969})}\BibitemShut {NoStop}%
\bibitem [{\citenamefont {Nahum}\ \emph {et~al.}(2017)\citenamefont {Nahum},
  \citenamefont {Ruhman}, \citenamefont {Vijay},\ and\ \citenamefont
  {Haah}}]{Nahum_2017}%
  \BibitemOpen
  \bibfield  {author} {\bibinfo {author} {\bibfnamefont {A.}~\bibnamefont
  {Nahum}}, \bibinfo {author} {\bibfnamefont {J.}~\bibnamefont {Ruhman}},
  \bibinfo {author} {\bibfnamefont {S.}~\bibnamefont {Vijay}},\ and\ \bibinfo
  {author} {\bibfnamefont {J.}~\bibnamefont {Haah}},\ }\bibfield  {journal}
  {\bibinfo  {journal} {Physical Review X}\ }\textbf {\bibinfo {volume} {7}},\
  \href {https://doi.org/10.1103/physrevx.7.031016} {10.1103/physrevx.7.031016}
  (\bibinfo {year} {2017})\BibitemShut {NoStop}%
\bibitem [{\citenamefont {Roberts}\ and\ \citenamefont
  {Stanford}(2015)}]{PhysRevLett.115.131603}%
  \BibitemOpen
  \bibfield  {author} {\bibinfo {author} {\bibfnamefont {D.~A.}\ \bibnamefont
  {Roberts}}\ and\ \bibinfo {author} {\bibfnamefont {D.}~\bibnamefont
  {Stanford}},\ }\href {https://doi.org/10.1103/PhysRevLett.115.131603}
  {\bibfield  {journal} {\bibinfo  {journal} {Phys. Rev. Lett.}\ }\textbf
  {\bibinfo {volume} {115}},\ \bibinfo {pages} {131603} (\bibinfo {year}
  {2015})}\BibitemShut {NoStop}%
\bibitem [{\citenamefont {Hosur}\ \emph {et~al.}(2016)\citenamefont {Hosur},
  \citenamefont {Qi}, \citenamefont {Roberts},\ and\ \citenamefont
  {Yoshida}}]{Hosur_2016}%
  \BibitemOpen
  \bibfield  {author} {\bibinfo {author} {\bibfnamefont {P.}~\bibnamefont
  {Hosur}}, \bibinfo {author} {\bibfnamefont {X.-L.}\ \bibnamefont {Qi}},
  \bibinfo {author} {\bibfnamefont {D.~A.}\ \bibnamefont {Roberts}},\ and\
  \bibinfo {author} {\bibfnamefont {B.}~\bibnamefont {Yoshida}},\ }\bibfield
  {journal} {\bibinfo  {journal} {Journal of High Energy Physics}\ }\textbf
  {\bibinfo {volume} {2016}},\ \href {https://doi.org/10.1007/jhep02(2016)004}
  {10.1007/jhep02(2016)004} (\bibinfo {year} {2016})\BibitemShut {NoStop}%
\bibitem [{\citenamefont {Chowdhury}\ and\ \citenamefont
  {Swingle}(2017)}]{PhysRevD.96.065005}%
  \BibitemOpen
  \bibfield  {author} {\bibinfo {author} {\bibfnamefont {D.}~\bibnamefont
  {Chowdhury}}\ and\ \bibinfo {author} {\bibfnamefont {B.}~\bibnamefont
  {Swingle}},\ }\href {https://doi.org/10.1103/PhysRevD.96.065005} {\bibfield
  {journal} {\bibinfo  {journal} {Phys. Rev. D}\ }\textbf {\bibinfo {volume}
  {96}},\ \bibinfo {pages} {065005} (\bibinfo {year} {2017})}\BibitemShut
  {NoStop}%
\bibitem [{\citenamefont {Harrow}\ \emph {et~al.}(2021)\citenamefont {Harrow},
  \citenamefont {Kong}, \citenamefont {Liu}, \citenamefont {Mehraban},\ and\
  \citenamefont {Shor}}]{PRXQuantum.2.020339}%
  \BibitemOpen
  \bibfield  {author} {\bibinfo {author} {\bibfnamefont {A.~W.}\ \bibnamefont
  {Harrow}}, \bibinfo {author} {\bibfnamefont {L.}~\bibnamefont {Kong}},
  \bibinfo {author} {\bibfnamefont {Z.-W.}\ \bibnamefont {Liu}}, \bibinfo
  {author} {\bibfnamefont {S.}~\bibnamefont {Mehraban}},\ and\ \bibinfo
  {author} {\bibfnamefont {P.~W.}\ \bibnamefont {Shor}},\ }\href
  {https://doi.org/10.1103/PRXQuantum.2.020339} {\bibfield  {journal} {\bibinfo
   {journal} {PRX Quantum}\ }\textbf {\bibinfo {volume} {2}},\ \bibinfo {pages}
  {020339} (\bibinfo {year} {2021})}\BibitemShut {NoStop}%
\bibitem [{\citenamefont {Nahum}\ \emph {et~al.}(2018)\citenamefont {Nahum},
  \citenamefont {Vijay},\ and\ \citenamefont {Haah}}]{PhysRevX.8.021014}%
  \BibitemOpen
  \bibfield  {author} {\bibinfo {author} {\bibfnamefont {A.}~\bibnamefont
  {Nahum}}, \bibinfo {author} {\bibfnamefont {S.}~\bibnamefont {Vijay}},\ and\
  \bibinfo {author} {\bibfnamefont {J.}~\bibnamefont {Haah}},\ }\href
  {https://doi.org/10.1103/PhysRevX.8.021014} {\bibfield  {journal} {\bibinfo
  {journal} {Phys. Rev. X}\ }\textbf {\bibinfo {volume} {8}},\ \bibinfo {pages}
  {021014} (\bibinfo {year} {2018})}\BibitemShut {NoStop}%
\bibitem [{\citenamefont {von Keyserlingk}\ \emph {et~al.}(2018)\citenamefont
  {von Keyserlingk}, \citenamefont {Rakovszky}, \citenamefont {Pollmann},\ and\
  \citenamefont {Sondhi}}]{von_Keyserlingk_2018}%
  \BibitemOpen
  \bibfield  {author} {\bibinfo {author} {\bibfnamefont {C.}~\bibnamefont {von
  Keyserlingk}}, \bibinfo {author} {\bibfnamefont {T.}~\bibnamefont
  {Rakovszky}}, \bibinfo {author} {\bibfnamefont {F.}~\bibnamefont
  {Pollmann}},\ and\ \bibinfo {author} {\bibfnamefont {S.}~\bibnamefont
  {Sondhi}},\ }\bibfield  {journal} {\bibinfo  {journal} {Physical Review X}\
  }\textbf {\bibinfo {volume} {8}},\ \href
  {https://doi.org/10.1103/physrevx.8.021013} {10.1103/physrevx.8.021013}
  (\bibinfo {year} {2018})\BibitemShut {NoStop}%
\bibitem [{\citenamefont {Zanardi}(2001)}]{PhysRevA.63.040304}%
  \BibitemOpen
  \bibfield  {author} {\bibinfo {author} {\bibfnamefont {P.}~\bibnamefont
  {Zanardi}},\ }\href {https://doi.org/10.1103/PhysRevA.63.040304} {\bibfield
  {journal} {\bibinfo  {journal} {Phys. Rev. A}\ }\textbf {\bibinfo {volume}
  {63}},\ \bibinfo {pages} {040304} (\bibinfo {year} {2001})}\BibitemShut
  {NoStop}%
\bibitem [{\citenamefont {Bandyopadhyay}\ and\ \citenamefont
  {Lakshminarayan}(2005)}]{Band2005}%
  \BibitemOpen
  \bibfield  {author} {\bibinfo {author} {\bibfnamefont {J.~N.}\ \bibnamefont
  {Bandyopadhyay}}\ and\ \bibinfo {author} {\bibfnamefont {A.}~\bibnamefont
  {Lakshminarayan}},\ }\href {https://doi.org/10.48550/ARXIV.QUANT-PH/0504052}
  {\bibinfo {title} {Entangling power of quantum chaotic evolutions via
  operator entanglement}} (\bibinfo {year} {2005})\BibitemShut {NoStop}%
\bibitem [{\citenamefont {Zhou}\ and\ \citenamefont {Luitz}(2017)}]{Zhou_2017}%
  \BibitemOpen
  \bibfield  {author} {\bibinfo {author} {\bibfnamefont {T.}~\bibnamefont
  {Zhou}}\ and\ \bibinfo {author} {\bibfnamefont {D.~J.}\ \bibnamefont
  {Luitz}},\ }\bibfield  {journal} {\bibinfo  {journal} {Physical Review B}\
  }\textbf {\bibinfo {volume} {95}},\ \href
  {https://doi.org/10.1103/physrevb.95.094206} {10.1103/physrevb.95.094206}
  (\bibinfo {year} {2017})\BibitemShut {NoStop}%
\bibitem [{\citenamefont {Roberts}\ and\ \citenamefont
  {Yoshida}(2017)}]{Roberts_2017}%
  \BibitemOpen
  \bibfield  {author} {\bibinfo {author} {\bibfnamefont {D.~A.}\ \bibnamefont
  {Roberts}}\ and\ \bibinfo {author} {\bibfnamefont {B.}~\bibnamefont
  {Yoshida}},\ }\bibfield  {journal} {\bibinfo  {journal} {Journal of High
  Energy Physics}\ }\textbf {\bibinfo {volume} {2017}},\ \href
  {https://doi.org/10.1007/jhep04(2017)121} {10.1007/jhep04(2017)121} (\bibinfo
  {year} {2017})\BibitemShut {NoStop}%
\bibitem [{\citenamefont {Zhu}\ \emph {et~al.}(2022)\citenamefont {Zhu},
  \citenamefont {Sun}, \citenamefont {Gong}, \citenamefont {Chen},
  \citenamefont {Zhang}, \citenamefont {Wu}, \citenamefont {Ye}, \citenamefont
  {Zha}, \citenamefont {Li}, \citenamefont {Guo}, \citenamefont {Qian},
  \citenamefont {Huang}, \citenamefont {Yu}, \citenamefont {Deng},
  \citenamefont {Rong}, \citenamefont {Lin}, \citenamefont {Xu}, \citenamefont
  {Sun}, \citenamefont {Guo}, \citenamefont {Li}, \citenamefont {Liang},
  \citenamefont {Peng}, \citenamefont {Fan}, \citenamefont {Zhu},\ and\
  \citenamefont {Pan}}]{PhysRevLett.128.160502}%
  \BibitemOpen
  \bibfield  {author} {\bibinfo {author} {\bibfnamefont {Q.}~\bibnamefont
  {Zhu}}, \bibinfo {author} {\bibfnamefont {Z.-H.}\ \bibnamefont {Sun}},
  \bibinfo {author} {\bibfnamefont {M.}~\bibnamefont {Gong}}, \bibinfo {author}
  {\bibfnamefont {F.}~\bibnamefont {Chen}}, \bibinfo {author} {\bibfnamefont
  {Y.-R.}\ \bibnamefont {Zhang}}, \bibinfo {author} {\bibfnamefont
  {Y.}~\bibnamefont {Wu}}, \bibinfo {author} {\bibfnamefont {Y.}~\bibnamefont
  {Ye}}, \bibinfo {author} {\bibfnamefont {C.}~\bibnamefont {Zha}}, \bibinfo
  {author} {\bibfnamefont {S.}~\bibnamefont {Li}}, \bibinfo {author}
  {\bibfnamefont {S.}~\bibnamefont {Guo}}, \bibinfo {author} {\bibfnamefont
  {H.}~\bibnamefont {Qian}}, \bibinfo {author} {\bibfnamefont {H.-L.}\
  \bibnamefont {Huang}}, \bibinfo {author} {\bibfnamefont {J.}~\bibnamefont
  {Yu}}, \bibinfo {author} {\bibfnamefont {H.}~\bibnamefont {Deng}}, \bibinfo
  {author} {\bibfnamefont {H.}~\bibnamefont {Rong}}, \bibinfo {author}
  {\bibfnamefont {J.}~\bibnamefont {Lin}}, \bibinfo {author} {\bibfnamefont
  {Y.}~\bibnamefont {Xu}}, \bibinfo {author} {\bibfnamefont {L.}~\bibnamefont
  {Sun}}, \bibinfo {author} {\bibfnamefont {C.}~\bibnamefont {Guo}}, \bibinfo
  {author} {\bibfnamefont {N.}~\bibnamefont {Li}}, \bibinfo {author}
  {\bibfnamefont {F.}~\bibnamefont {Liang}}, \bibinfo {author} {\bibfnamefont
  {C.-Z.}\ \bibnamefont {Peng}}, \bibinfo {author} {\bibfnamefont
  {H.}~\bibnamefont {Fan}}, \bibinfo {author} {\bibfnamefont {X.}~\bibnamefont
  {Zhu}},\ and\ \bibinfo {author} {\bibfnamefont {J.-W.}\ \bibnamefont {Pan}},\
  }\href {https://doi.org/10.1103/PhysRevLett.128.160502} {\bibfield  {journal}
  {\bibinfo  {journal} {Phys. Rev. Lett.}\ }\textbf {\bibinfo {volume} {128}},\
  \bibinfo {pages} {160502} (\bibinfo {year} {2022})}\BibitemShut {NoStop}%
\bibitem [{\citenamefont {Swingle}\ and\ \citenamefont
  {Chowdhury}(2017)}]{PhysRevB.95.060201}%
  \BibitemOpen
  \bibfield  {author} {\bibinfo {author} {\bibfnamefont {B.}~\bibnamefont
  {Swingle}}\ and\ \bibinfo {author} {\bibfnamefont {D.}~\bibnamefont
  {Chowdhury}},\ }\href {https://doi.org/10.1103/PhysRevB.95.060201} {\bibfield
   {journal} {\bibinfo  {journal} {Phys. Rev. B}\ }\textbf {\bibinfo {volume}
  {95}},\ \bibinfo {pages} {060201} (\bibinfo {year} {2017})}\BibitemShut
  {NoStop}%
\bibitem [{\citenamefont {Huang}\ \emph {et~al.}(2016)\citenamefont {Huang},
  \citenamefont {Zhang},\ and\ \citenamefont {Chen}}]{Huang_2016}%
  \BibitemOpen
  \bibfield  {author} {\bibinfo {author} {\bibfnamefont {Y.}~\bibnamefont
  {Huang}}, \bibinfo {author} {\bibfnamefont {Y.-L.}\ \bibnamefont {Zhang}},\
  and\ \bibinfo {author} {\bibfnamefont {X.}~\bibnamefont {Chen}},\ }\href
  {https://doi.org/10.1002/andp.201600318} {\bibfield  {journal} {\bibinfo
  {journal} {Annalen der Physik}\ }\textbf {\bibinfo {volume} {529}},\ \bibinfo
  {pages} {1600318} (\bibinfo {year} {2016})}\BibitemShut {NoStop}%
\bibitem [{\citenamefont {Chen}(2016)}]{Chen2016}%
  \BibitemOpen
  \bibfield  {author} {\bibinfo {author} {\bibfnamefont {Y.}~\bibnamefont
  {Chen}},\ }\href {https://doi.org/10.48550/ARXIV.1608.02765} {\bibinfo
  {title} {Universal logarithmic scrambling in many body localization}}
  (\bibinfo {year} {2016})\BibitemShut {NoStop}%
\bibitem [{\citenamefont {Chen}\ \emph {et~al.}(2016)\citenamefont {Chen},
  \citenamefont {Zhou}, \citenamefont {Huse},\ and\ \citenamefont
  {Fradkin}}]{Chen_2016}%
  \BibitemOpen
  \bibfield  {author} {\bibinfo {author} {\bibfnamefont {X.}~\bibnamefont
  {Chen}}, \bibinfo {author} {\bibfnamefont {T.}~\bibnamefont {Zhou}}, \bibinfo
  {author} {\bibfnamefont {D.~A.}\ \bibnamefont {Huse}},\ and\ \bibinfo
  {author} {\bibfnamefont {E.}~\bibnamefont {Fradkin}},\ }\href
  {https://doi.org/10.1002/andp.201600332} {\bibfield  {journal} {\bibinfo
  {journal} {Annalen der Physik}\ }\textbf {\bibinfo {volume} {529}},\ \bibinfo
  {pages} {1600332} (\bibinfo {year} {2016})}\BibitemShut {NoStop}%
\bibitem [{\citenamefont {Fan}\ \emph {et~al.}(2017)\citenamefont {Fan},
  \citenamefont {Zhang}, \citenamefont {Shen},\ and\ \citenamefont
  {Zhai}}]{Fan_2017}%
  \BibitemOpen
  \bibfield  {author} {\bibinfo {author} {\bibfnamefont {R.}~\bibnamefont
  {Fan}}, \bibinfo {author} {\bibfnamefont {P.}~\bibnamefont {Zhang}}, \bibinfo
  {author} {\bibfnamefont {H.}~\bibnamefont {Shen}},\ and\ \bibinfo {author}
  {\bibfnamefont {H.}~\bibnamefont {Zhai}},\ }\href
  {https://doi.org/10.1016/j.scib.2017.04.011} {\bibfield  {journal} {\bibinfo
  {journal} {Science Bulletin}\ }\textbf {\bibinfo {volume} {62}},\ \bibinfo
  {pages} {707} (\bibinfo {year} {2017})}\BibitemShut {NoStop}%
\bibitem [{\citenamefont {He}\ and\ \citenamefont
  {Lu}(2017)}]{PhysRevB.95.054201}%
  \BibitemOpen
  \bibfield  {author} {\bibinfo {author} {\bibfnamefont {R.-Q.}\ \bibnamefont
  {He}}\ and\ \bibinfo {author} {\bibfnamefont {Z.-Y.}\ \bibnamefont {Lu}},\
  }\href {https://doi.org/10.1103/PhysRevB.95.054201} {\bibfield  {journal}
  {\bibinfo  {journal} {Phys. Rev. B}\ }\textbf {\bibinfo {volume} {95}},\
  \bibinfo {pages} {054201} (\bibinfo {year} {2017})}\BibitemShut {NoStop}%
\bibitem [{\citenamefont {Belyansky}\ \emph {et~al.}(2020)\citenamefont
  {Belyansky}, \citenamefont {Bienias}, \citenamefont {Kharkov}, \citenamefont
  {Gorshkov},\ and\ \citenamefont {Swingle}}]{Belyansky_2020}%
  \BibitemOpen
  \bibfield  {author} {\bibinfo {author} {\bibfnamefont {R.}~\bibnamefont
  {Belyansky}}, \bibinfo {author} {\bibfnamefont {P.}~\bibnamefont {Bienias}},
  \bibinfo {author} {\bibfnamefont {Y.~A.}\ \bibnamefont {Kharkov}}, \bibinfo
  {author} {\bibfnamefont {A.~V.}\ \bibnamefont {Gorshkov}},\ and\ \bibinfo
  {author} {\bibfnamefont {B.}~\bibnamefont {Swingle}},\ }\bibfield  {journal}
  {\bibinfo  {journal} {Physical Review Letters}\ }\textbf {\bibinfo {volume}
  {125}},\ \href {https://doi.org/10.1103/physrevlett.125.130601}
  {10.1103/physrevlett.125.130601} (\bibinfo {year} {2020})\BibitemShut
  {NoStop}%
\bibitem [{\citenamefont {Mi}\ \emph {et~al.}(2021)\citenamefont {Mi},
  \citenamefont {Roushan}, \citenamefont {Quintana}, \citenamefont
  {Mandr{\`{a}}}, \citenamefont {Marshall}, \citenamefont {Neill},
  \citenamefont {Arute}, \citenamefont {Arya}, \citenamefont {Atalaya},
  \citenamefont {Babbush}, \citenamefont {Bardin}, \citenamefont {Barends},
  \citenamefont {Basso}, \citenamefont {Bengtsson}, \citenamefont {Boixo},
  \citenamefont {Bourassa}, \citenamefont {Broughton}, \citenamefont {Buckley},
  \citenamefont {Buell}, \citenamefont {Burkett}, \citenamefont {Bushnell},
  \citenamefont {Chen}, \citenamefont {Chiaro}, \citenamefont {Collins},
  \citenamefont {Courtney}, \citenamefont {Demura}, \citenamefont {Derk},
  \citenamefont {Dunsworth}, \citenamefont {Eppens}, \citenamefont {Erickson},
  \citenamefont {Farhi}, \citenamefont {Fowler}, \citenamefont {Foxen},
  \citenamefont {Gidney}, \citenamefont {Giustina}, \citenamefont {Gross},
  \citenamefont {Harrigan}, \citenamefont {Harrington}, \citenamefont {Hilton},
  \citenamefont {Ho}, \citenamefont {Hong}, \citenamefont {Huang},
  \citenamefont {Huggins}, \citenamefont {Ioffe}, \citenamefont {Isakov},
  \citenamefont {Jeffrey}, \citenamefont {Jiang}, \citenamefont {Jones},
  \citenamefont {Kafri}, \citenamefont {Kelly}, \citenamefont {Kim},
  \citenamefont {Kitaev}, \citenamefont {Klimov}, \citenamefont {Korotkov},
  \citenamefont {Kostritsa}, \citenamefont {Landhuis}, \citenamefont {Laptev},
  \citenamefont {Lucero}, \citenamefont {Martin}, \citenamefont {McClean},
  \citenamefont {McCourt}, \citenamefont {McEwen}, \citenamefont {Megrant},
  \citenamefont {Miao}, \citenamefont {Mohseni}, \citenamefont {Montazeri},
  \citenamefont {Mruczkiewicz}, \citenamefont {Mutus}, \citenamefont {Naaman},
  \citenamefont {Neeley}, \citenamefont {Newman}, \citenamefont {Niu},
  \citenamefont {O'Brien}, \citenamefont {Opremcak}, \citenamefont {Ostby},
  \citenamefont {Pato}, \citenamefont {Petukhov}, \citenamefont {Redd},
  \citenamefont {Rubin}, \citenamefont {Sank}, \citenamefont {Satzinger},
  \citenamefont {Shvarts}, \citenamefont {Strain}, \citenamefont {Szalay},
  \citenamefont {Trevithick}, \citenamefont {Villalonga}, \citenamefont
  {White}, \citenamefont {Yao}, \citenamefont {Yeh}, \citenamefont {Zalcman},
  \citenamefont {Neven}, \citenamefont {Aleiner}, \citenamefont {Kechedzhi},
  \citenamefont {Smelyanskiy},\ and\ \citenamefont {Chen}}]{Mi_2021}%
  \BibitemOpen
  \bibfield  {author} {\bibinfo {author} {\bibfnamefont {X.}~\bibnamefont
  {Mi}}, \bibinfo {author} {\bibfnamefont {P.}~\bibnamefont {Roushan}},
  \bibinfo {author} {\bibfnamefont {C.}~\bibnamefont {Quintana}}, \bibinfo
  {author} {\bibfnamefont {S.}~\bibnamefont {Mandr{\`{a}}}}, \bibinfo {author}
  {\bibfnamefont {J.}~\bibnamefont {Marshall}}, \bibinfo {author}
  {\bibfnamefont {C.}~\bibnamefont {Neill}}, \bibinfo {author} {\bibfnamefont
  {F.}~\bibnamefont {Arute}}, \bibinfo {author} {\bibfnamefont
  {K.}~\bibnamefont {Arya}}, \bibinfo {author} {\bibfnamefont {J.}~\bibnamefont
  {Atalaya}}, \bibinfo {author} {\bibfnamefont {R.}~\bibnamefont {Babbush}},
  \bibinfo {author} {\bibfnamefont {J.~C.}\ \bibnamefont {Bardin}}, \bibinfo
  {author} {\bibfnamefont {R.}~\bibnamefont {Barends}}, \bibinfo {author}
  {\bibfnamefont {J.}~\bibnamefont {Basso}}, \bibinfo {author} {\bibfnamefont
  {A.}~\bibnamefont {Bengtsson}}, \bibinfo {author} {\bibfnamefont
  {S.}~\bibnamefont {Boixo}}, \bibinfo {author} {\bibfnamefont
  {A.}~\bibnamefont {Bourassa}}, \bibinfo {author} {\bibfnamefont
  {M.}~\bibnamefont {Broughton}}, \bibinfo {author} {\bibfnamefont {B.~B.}\
  \bibnamefont {Buckley}}, \bibinfo {author} {\bibfnamefont {D.~A.}\
  \bibnamefont {Buell}}, \bibinfo {author} {\bibfnamefont {B.}~\bibnamefont
  {Burkett}}, \bibinfo {author} {\bibfnamefont {N.}~\bibnamefont {Bushnell}},
  \bibinfo {author} {\bibfnamefont {Z.}~\bibnamefont {Chen}}, \bibinfo {author}
  {\bibfnamefont {B.}~\bibnamefont {Chiaro}}, \bibinfo {author} {\bibfnamefont
  {R.}~\bibnamefont {Collins}}, \bibinfo {author} {\bibfnamefont
  {W.}~\bibnamefont {Courtney}}, \bibinfo {author} {\bibfnamefont
  {S.}~\bibnamefont {Demura}}, \bibinfo {author} {\bibfnamefont {A.~R.}\
  \bibnamefont {Derk}}, \bibinfo {author} {\bibfnamefont {A.}~\bibnamefont
  {Dunsworth}}, \bibinfo {author} {\bibfnamefont {D.}~\bibnamefont {Eppens}},
  \bibinfo {author} {\bibfnamefont {C.}~\bibnamefont {Erickson}}, \bibinfo
  {author} {\bibfnamefont {E.}~\bibnamefont {Farhi}}, \bibinfo {author}
  {\bibfnamefont {A.~G.}\ \bibnamefont {Fowler}}, \bibinfo {author}
  {\bibfnamefont {B.}~\bibnamefont {Foxen}}, \bibinfo {author} {\bibfnamefont
  {C.}~\bibnamefont {Gidney}}, \bibinfo {author} {\bibfnamefont
  {M.}~\bibnamefont {Giustina}}, \bibinfo {author} {\bibfnamefont {J.~A.}\
  \bibnamefont {Gross}}, \bibinfo {author} {\bibfnamefont {M.~P.}\ \bibnamefont
  {Harrigan}}, \bibinfo {author} {\bibfnamefont {S.~D.}\ \bibnamefont
  {Harrington}}, \bibinfo {author} {\bibfnamefont {J.}~\bibnamefont {Hilton}},
  \bibinfo {author} {\bibfnamefont {A.}~\bibnamefont {Ho}}, \bibinfo {author}
  {\bibfnamefont {S.}~\bibnamefont {Hong}}, \bibinfo {author} {\bibfnamefont
  {T.}~\bibnamefont {Huang}}, \bibinfo {author} {\bibfnamefont {W.~J.}\
  \bibnamefont {Huggins}}, \bibinfo {author} {\bibfnamefont {L.~B.}\
  \bibnamefont {Ioffe}}, \bibinfo {author} {\bibfnamefont {S.~V.}\ \bibnamefont
  {Isakov}}, \bibinfo {author} {\bibfnamefont {E.}~\bibnamefont {Jeffrey}},
  \bibinfo {author} {\bibfnamefont {Z.}~\bibnamefont {Jiang}}, \bibinfo
  {author} {\bibfnamefont {C.}~\bibnamefont {Jones}}, \bibinfo {author}
  {\bibfnamefont {D.}~\bibnamefont {Kafri}}, \bibinfo {author} {\bibfnamefont
  {J.}~\bibnamefont {Kelly}}, \bibinfo {author} {\bibfnamefont
  {S.}~\bibnamefont {Kim}}, \bibinfo {author} {\bibfnamefont {A.}~\bibnamefont
  {Kitaev}}, \bibinfo {author} {\bibfnamefont {P.~V.}\ \bibnamefont {Klimov}},
  \bibinfo {author} {\bibfnamefont {A.~N.}\ \bibnamefont {Korotkov}}, \bibinfo
  {author} {\bibfnamefont {F.}~\bibnamefont {Kostritsa}}, \bibinfo {author}
  {\bibfnamefont {D.}~\bibnamefont {Landhuis}}, \bibinfo {author}
  {\bibfnamefont {P.}~\bibnamefont {Laptev}}, \bibinfo {author} {\bibfnamefont
  {E.}~\bibnamefont {Lucero}}, \bibinfo {author} {\bibfnamefont
  {O.}~\bibnamefont {Martin}}, \bibinfo {author} {\bibfnamefont {J.~R.}\
  \bibnamefont {McClean}}, \bibinfo {author} {\bibfnamefont {T.}~\bibnamefont
  {McCourt}}, \bibinfo {author} {\bibfnamefont {M.}~\bibnamefont {McEwen}},
  \bibinfo {author} {\bibfnamefont {A.}~\bibnamefont {Megrant}}, \bibinfo
  {author} {\bibfnamefont {K.~C.}\ \bibnamefont {Miao}}, \bibinfo {author}
  {\bibfnamefont {M.}~\bibnamefont {Mohseni}}, \bibinfo {author} {\bibfnamefont
  {S.}~\bibnamefont {Montazeri}}, \bibinfo {author} {\bibfnamefont
  {W.}~\bibnamefont {Mruczkiewicz}}, \bibinfo {author} {\bibfnamefont
  {J.}~\bibnamefont {Mutus}}, \bibinfo {author} {\bibfnamefont
  {O.}~\bibnamefont {Naaman}}, \bibinfo {author} {\bibfnamefont
  {M.}~\bibnamefont {Neeley}}, \bibinfo {author} {\bibfnamefont
  {M.}~\bibnamefont {Newman}}, \bibinfo {author} {\bibfnamefont {M.~Y.}\
  \bibnamefont {Niu}}, \bibinfo {author} {\bibfnamefont {T.~E.}\ \bibnamefont
  {O'Brien}}, \bibinfo {author} {\bibfnamefont {A.}~\bibnamefont {Opremcak}},
  \bibinfo {author} {\bibfnamefont {E.}~\bibnamefont {Ostby}}, \bibinfo
  {author} {\bibfnamefont {B.}~\bibnamefont {Pato}}, \bibinfo {author}
  {\bibfnamefont {A.}~\bibnamefont {Petukhov}}, \bibinfo {author}
  {\bibfnamefont {N.}~\bibnamefont {Redd}}, \bibinfo {author} {\bibfnamefont
  {N.~C.}\ \bibnamefont {Rubin}}, \bibinfo {author} {\bibfnamefont
  {D.}~\bibnamefont {Sank}}, \bibinfo {author} {\bibfnamefont {K.~J.}\
  \bibnamefont {Satzinger}}, \bibinfo {author} {\bibfnamefont {V.}~\bibnamefont
  {Shvarts}}, \bibinfo {author} {\bibfnamefont {D.}~\bibnamefont {Strain}},
  \bibinfo {author} {\bibfnamefont {M.}~\bibnamefont {Szalay}}, \bibinfo
  {author} {\bibfnamefont {M.~D.}\ \bibnamefont {Trevithick}}, \bibinfo
  {author} {\bibfnamefont {B.}~\bibnamefont {Villalonga}}, \bibinfo {author}
  {\bibfnamefont {T.}~\bibnamefont {White}}, \bibinfo {author} {\bibfnamefont
  {Z.~J.}\ \bibnamefont {Yao}}, \bibinfo {author} {\bibfnamefont
  {P.}~\bibnamefont {Yeh}}, \bibinfo {author} {\bibfnamefont {A.}~\bibnamefont
  {Zalcman}}, \bibinfo {author} {\bibfnamefont {H.}~\bibnamefont {Neven}},
  \bibinfo {author} {\bibfnamefont {I.}~\bibnamefont {Aleiner}}, \bibinfo
  {author} {\bibfnamefont {K.}~\bibnamefont {Kechedzhi}}, \bibinfo {author}
  {\bibfnamefont {V.}~\bibnamefont {Smelyanskiy}},\ and\ \bibinfo {author}
  {\bibfnamefont {Y.}~\bibnamefont {Chen}},\ }\href
  {https://doi.org/10.1126/science.abg5029} {\bibfield  {journal} {\bibinfo
  {journal} {Science}\ }\textbf {\bibinfo {volume} {374}},\ \bibinfo {pages}
  {1479} (\bibinfo {year} {2021})}\BibitemShut {NoStop}%
\bibitem [{\citenamefont {Horodecki}\ and\ \citenamefont
  {Oppenheim}(2012)}]{HORODECKI_2012}%
  \BibitemOpen
  \bibfield  {author} {\bibinfo {author} {\bibfnamefont {M.}~\bibnamefont
  {Horodecki}}\ and\ \bibinfo {author} {\bibfnamefont {J.}~\bibnamefont
  {Oppenheim}},\ }\href {https://doi.org/10.1142/s0217979213450197} {\bibfield
  {journal} {\bibinfo  {journal} {International Journal of Modern Physics B}\
  }\textbf {\bibinfo {volume} {27}},\ \bibinfo {pages} {1345019} (\bibinfo
  {year} {2012})}\BibitemShut {NoStop}%
\bibitem [{\citenamefont {Chitambar}\ and\ \citenamefont
  {Gour}(2019{\natexlab{a}})}]{ChitambarRMP19}%
  \BibitemOpen
  \bibfield  {author} {\bibinfo {author} {\bibfnamefont {E.}~\bibnamefont
  {Chitambar}}\ and\ \bibinfo {author} {\bibfnamefont {G.}~\bibnamefont
  {Gour}},\ }\href {https://doi.org/10.1103/RevModPhys.91.025001} {\bibfield
  {journal} {\bibinfo  {journal} {Rev. Mod. Phys.}\ }\textbf {\bibinfo {volume}
  {91}},\ \bibinfo {pages} {025001} (\bibinfo {year}
  {2019}{\natexlab{a}})}\BibitemShut {NoStop}%
\bibitem [{\citenamefont {Liu}\ \emph {et~al.}(2019)\citenamefont {Liu},
  \citenamefont {Bu},\ and\ \citenamefont {Takagi}}]{Liu2019}%
  \BibitemOpen
  \bibfield  {author} {\bibinfo {author} {\bibfnamefont {Z.-W.}\ \bibnamefont
  {Liu}}, \bibinfo {author} {\bibfnamefont {K.}~\bibnamefont {Bu}},\ and\
  \bibinfo {author} {\bibfnamefont {R.}~\bibnamefont {Takagi}},\ }\href
  {https://doi.org/10.1103/PhysRevLett.123.020401} {\bibfield  {journal}
  {\bibinfo  {journal} {Phys. Rev. Lett.}\ }\textbf {\bibinfo {volume} {123}},\
  \bibinfo {pages} {020401} (\bibinfo {year} {2019})}\BibitemShut {NoStop}%
\bibitem [{\citenamefont {Chitambar}\ and\ \citenamefont
  {Gour}(2019{\natexlab{b}})}]{Chitambar_2019}%
  \BibitemOpen
  \bibfield  {author} {\bibinfo {author} {\bibfnamefont {E.}~\bibnamefont
  {Chitambar}}\ and\ \bibinfo {author} {\bibfnamefont {G.}~\bibnamefont
  {Gour}},\ }\bibfield  {journal} {\bibinfo  {journal} {Reviews of Modern
  Physics}\ }\textbf {\bibinfo {volume} {91}},\ \href
  {https://doi.org/10.1103/revmodphys.91.025001} {10.1103/revmodphys.91.025001}
  (\bibinfo {year} {2019}{\natexlab{b}})\BibitemShut {NoStop}%
\bibitem [{\citenamefont {Veitch}\ \emph {et~al.}(2014)\citenamefont {Veitch},
  \citenamefont {Mousavian}, \citenamefont {Gottesman},\ and\ \citenamefont
  {Emerson}}]{Veitch_2014}%
  \BibitemOpen
  \bibfield  {author} {\bibinfo {author} {\bibfnamefont {V.}~\bibnamefont
  {Veitch}}, \bibinfo {author} {\bibfnamefont {S.~A.~H.}\ \bibnamefont
  {Mousavian}}, \bibinfo {author} {\bibfnamefont {D.}~\bibnamefont
  {Gottesman}},\ and\ \bibinfo {author} {\bibfnamefont {J.}~\bibnamefont
  {Emerson}},\ }\href {https://doi.org/10.1088/1367-2630/16/1/013009}
  {\bibfield  {journal} {\bibinfo  {journal} {New Journal of Physics}\ }\textbf
  {\bibinfo {volume} {16}},\ \bibinfo {pages} {013009} (\bibinfo {year}
  {2014})}\BibitemShut {NoStop}%
\bibitem [{\citenamefont {Howard}\ and\ \citenamefont
  {Campbell}(2017)}]{Howard_2017}%
  \BibitemOpen
  \bibfield  {author} {\bibinfo {author} {\bibfnamefont {M.}~\bibnamefont
  {Howard}}\ and\ \bibinfo {author} {\bibfnamefont {E.}~\bibnamefont
  {Campbell}},\ }\bibfield  {journal} {\bibinfo  {journal} {Physical Review
  Letters}\ }\textbf {\bibinfo {volume} {118}},\ \href
  {https://doi.org/10.1103/physrevlett.118.090501}
  {10.1103/physrevlett.118.090501} (\bibinfo {year} {2017})\BibitemShut
  {NoStop}%
\bibitem [{\citenamefont {Wang}\ \emph {et~al.}(2019)\citenamefont {Wang},
  \citenamefont {Wilde},\ and\ \citenamefont {Su}}]{Wang_2019}%
  \BibitemOpen
  \bibfield  {author} {\bibinfo {author} {\bibfnamefont {X.}~\bibnamefont
  {Wang}}, \bibinfo {author} {\bibfnamefont {M.~M.}\ \bibnamefont {Wilde}},\
  and\ \bibinfo {author} {\bibfnamefont {Y.}~\bibnamefont {Su}},\ }\href
  {https://doi.org/10.1088/1367-2630/ab451d} {\bibfield  {journal} {\bibinfo
  {journal} {New Journal of Physics}\ }\textbf {\bibinfo {volume} {21}},\
  \bibinfo {pages} {103002} (\bibinfo {year} {2019})}\BibitemShut {NoStop}%
\bibitem [{\citenamefont {Brand\~ao}\ \emph {et~al.}(2013)\citenamefont
  {Brand\~ao}, \citenamefont {Horodecki}, \citenamefont {Oppenheim},
  \citenamefont {Renes},\ and\ \citenamefont {Spekkens}}]{Brandao2013}%
  \BibitemOpen
  \bibfield  {author} {\bibinfo {author} {\bibfnamefont {F.~G. S.~L.}\
  \bibnamefont {Brand\~ao}}, \bibinfo {author} {\bibfnamefont {M.}~\bibnamefont
  {Horodecki}}, \bibinfo {author} {\bibfnamefont {J.}~\bibnamefont
  {Oppenheim}}, \bibinfo {author} {\bibfnamefont {J.~M.}\ \bibnamefont
  {Renes}},\ and\ \bibinfo {author} {\bibfnamefont {R.~W.}\ \bibnamefont
  {Spekkens}},\ }\href {https://doi.org/10.1103/PhysRevLett.111.250404}
  {\bibfield  {journal} {\bibinfo  {journal} {Phys. Rev. Lett.}\ }\textbf
  {\bibinfo {volume} {111}},\ \bibinfo {pages} {250404} (\bibinfo {year}
  {2013})}\BibitemShut {NoStop}%
\bibitem [{\citenamefont {Brand{\~a}o}\ \emph {et~al.}(2015)\citenamefont
  {Brand{\~a}o}, \citenamefont {Horodecki}, \citenamefont {Ng}, \citenamefont
  {Oppenheim},\ and\ \citenamefont {Wehner}}]{Brandao_secondlaws2015}%
  \BibitemOpen
  \bibfield  {author} {\bibinfo {author} {\bibfnamefont {F.}~\bibnamefont
  {Brand{\~a}o}}, \bibinfo {author} {\bibfnamefont {M.}~\bibnamefont
  {Horodecki}}, \bibinfo {author} {\bibfnamefont {N.}~\bibnamefont {Ng}},
  \bibinfo {author} {\bibfnamefont {J.}~\bibnamefont {Oppenheim}},\ and\
  \bibinfo {author} {\bibfnamefont {S.}~\bibnamefont {Wehner}},\ }\href
  {https://doi.org/10.1073/pnas.1411728112} {\bibfield  {journal} {\bibinfo
  {journal} {Proceedings of the National Academy of Sciences}\ }\textbf
  {\bibinfo {volume} {112}},\ \bibinfo {pages} {3275} (\bibinfo {year}
  {2015})}\BibitemShut {NoStop}%
\bibitem [{\citenamefont {Chiribella}\ \emph {et~al.}(2021)\citenamefont
  {Chiribella}, \citenamefont {Yang},\ and\ \citenamefont
  {Renner}}]{RennerPRX21}%
  \BibitemOpen
  \bibfield  {author} {\bibinfo {author} {\bibfnamefont {G.}~\bibnamefont
  {Chiribella}}, \bibinfo {author} {\bibfnamefont {Y.}~\bibnamefont {Yang}},\
  and\ \bibinfo {author} {\bibfnamefont {R.}~\bibnamefont {Renner}},\ }\href
  {https://doi.org/10.1103/PhysRevX.11.021014} {\bibfield  {journal} {\bibinfo
  {journal} {Phys. Rev. X}\ }\textbf {\bibinfo {volume} {11}},\ \bibinfo
  {pages} {021014} (\bibinfo {year} {2021})}\BibitemShut {NoStop}%
\bibitem [{\citenamefont {Meng}\ \emph {et~al.}(2022)\citenamefont {Meng},
  \citenamefont {Chiribella}, \citenamefont {Renner},\ and\ \citenamefont
  {Yung}}]{Renner2022}%
  \BibitemOpen
  \bibfield  {author} {\bibinfo {author} {\bibfnamefont {F.}~\bibnamefont
  {Meng}}, \bibinfo {author} {\bibfnamefont {G.}~\bibnamefont {Chiribella}},
  \bibinfo {author} {\bibfnamefont {R.}~\bibnamefont {Renner}},\ and\ \bibinfo
  {author} {\bibfnamefont {M.-H.}\ \bibnamefont {Yung}},\ }\href
  {https://doi.org/10.48550/ARXIV.2203.09369} {\bibinfo {title} {The
  nonequilibrium cost of accurate information processing}} (\bibinfo {year}
  {2022})\BibitemShut {NoStop}%
\bibitem [{\citenamefont {Aberg}(2006)}]{aberg2006quantifying}%
  \BibitemOpen
  \bibfield  {author} {\bibinfo {author} {\bibfnamefont {J.}~\bibnamefont
  {Aberg}},\ }\href@noop {} {\bibfield  {journal} {\bibinfo  {journal} {arXiv
  preprint quant-ph/0612146}\ } (\bibinfo {year} {2006})}\BibitemShut {NoStop}%
\bibitem [{\citenamefont {Baumgratz}\ \emph {et~al.}(2014)\citenamefont
  {Baumgratz}, \citenamefont {Cramer},\ and\ \citenamefont
  {Plenio}}]{baumgratz2014quantifying}%
  \BibitemOpen
  \bibfield  {author} {\bibinfo {author} {\bibfnamefont {T.}~\bibnamefont
  {Baumgratz}}, \bibinfo {author} {\bibfnamefont {M.}~\bibnamefont {Cramer}},\
  and\ \bibinfo {author} {\bibfnamefont {M.~B.}\ \bibnamefont {Plenio}},\
  }\href {https://doi.org/10.1103/PhysRevLett.113.140401} {\bibfield  {journal}
  {\bibinfo  {journal} {Phys. Rev. Lett.}\ }\textbf {\bibinfo {volume} {113}},\
  \bibinfo {pages} {140401} (\bibinfo {year} {2014})}\BibitemShut {NoStop}%
\bibitem [{\citenamefont {Winter}\ and\ \citenamefont
  {Yang}(2016)}]{winter2016operational}%
  \BibitemOpen
  \bibfield  {author} {\bibinfo {author} {\bibfnamefont {A.}~\bibnamefont
  {Winter}}\ and\ \bibinfo {author} {\bibfnamefont {D.}~\bibnamefont {Yang}},\
  }\href {https://doi.org/10.1103/PhysRevLett.116.120404} {\bibfield  {journal}
  {\bibinfo  {journal} {Phys. Rev. Lett.}\ }\textbf {\bibinfo {volume} {116}},\
  \bibinfo {pages} {120404} (\bibinfo {year} {2016})}\BibitemShut {NoStop}%
\bibitem [{\citenamefont {Bu}\ \emph {et~al.}(2017)\citenamefont {Bu},
  \citenamefont {Singh}, \citenamefont {Fei}, \citenamefont {Pati},\ and\
  \citenamefont {Wu}}]{bu2017maximum}%
  \BibitemOpen
  \bibfield  {author} {\bibinfo {author} {\bibfnamefont {K.}~\bibnamefont
  {Bu}}, \bibinfo {author} {\bibfnamefont {U.}~\bibnamefont {Singh}}, \bibinfo
  {author} {\bibfnamefont {S.-M.}\ \bibnamefont {Fei}}, \bibinfo {author}
  {\bibfnamefont {A.~K.}\ \bibnamefont {Pati}},\ and\ \bibinfo {author}
  {\bibfnamefont {J.}~\bibnamefont {Wu}},\ }\href
  {https://doi.org/10.1103/PhysRevLett.119.150405} {\bibfield  {journal}
  {\bibinfo  {journal} {Phys. Rev. Lett.}\ }\textbf {\bibinfo {volume} {119}},\
  \bibinfo {pages} {150405} (\bibinfo {year} {2017})}\BibitemShut {NoStop}%
\bibitem [{\citenamefont {Streltsov}\ \emph {et~al.}(2017)\citenamefont
  {Streltsov}, \citenamefont {Adesso},\ and\ \citenamefont
  {Plenio}}]{Streltsov2017colloquium}%
  \BibitemOpen
  \bibfield  {author} {\bibinfo {author} {\bibfnamefont {A.}~\bibnamefont
  {Streltsov}}, \bibinfo {author} {\bibfnamefont {G.}~\bibnamefont {Adesso}},\
  and\ \bibinfo {author} {\bibfnamefont {M.~B.}\ \bibnamefont {Plenio}},\
  }\href {https://doi.org/10.1103/RevModPhys.89.041003} {\bibfield  {journal}
  {\bibinfo  {journal} {Rev. Mod. Phys.}\ }\textbf {\bibinfo {volume} {89}},\
  \bibinfo {pages} {041003} (\bibinfo {year} {2017})}\BibitemShut {NoStop}%
\bibitem [{\citenamefont {Halpern}\ \emph {et~al.}(2021)\citenamefont
  {Halpern}, \citenamefont {Kothakonda}, \citenamefont {Haferkamp},
  \citenamefont {Munson}, \citenamefont {Eisert},\ and\ \citenamefont
  {Faist}}]{Halpern2021}%
  \BibitemOpen
  \bibfield  {author} {\bibinfo {author} {\bibfnamefont {N.~Y.}\ \bibnamefont
  {Halpern}}, \bibinfo {author} {\bibfnamefont {N.~B.~T.}\ \bibnamefont
  {Kothakonda}}, \bibinfo {author} {\bibfnamefont {J.}~\bibnamefont
  {Haferkamp}}, \bibinfo {author} {\bibfnamefont {A.}~\bibnamefont {Munson}},
  \bibinfo {author} {\bibfnamefont {J.}~\bibnamefont {Eisert}},\ and\ \bibinfo
  {author} {\bibfnamefont {P.}~\bibnamefont {Faist}},\ }\href
  {https://doi.org/10.48550/ARXIV.2110.11371} {\bibinfo {title} {Resource
  theory of quantum uncomplexity}} (\bibinfo {year} {2021})\BibitemShut
  {NoStop}%
\bibitem [{\citenamefont {Takagi}\ \emph {et~al.}(2019)\citenamefont {Takagi},
  \citenamefont {Regula}, \citenamefont {Bu}, \citenamefont {Liu},\ and\
  \citenamefont {Adesso}}]{Takagi_2019}%
  \BibitemOpen
  \bibfield  {author} {\bibinfo {author} {\bibfnamefont {R.}~\bibnamefont
  {Takagi}}, \bibinfo {author} {\bibfnamefont {B.}~\bibnamefont {Regula}},
  \bibinfo {author} {\bibfnamefont {K.}~\bibnamefont {Bu}}, \bibinfo {author}
  {\bibfnamefont {Z.-W.}\ \bibnamefont {Liu}},\ and\ \bibinfo {author}
  {\bibfnamefont {G.}~\bibnamefont {Adesso}},\ }\bibfield  {journal} {\bibinfo
  {journal} {Physical Review Letters}\ }\textbf {\bibinfo {volume} {122}},\
  \href {https://doi.org/10.1103/physrevlett.122.140402}
  {10.1103/physrevlett.122.140402} (\bibinfo {year} {2019})\BibitemShut
  {NoStop}%
\bibitem [{\citenamefont {Bennett}\ \emph {et~al.}(1993)\citenamefont
  {Bennett}, \citenamefont {Brassard}, \citenamefont {Cr\'epeau}, \citenamefont
  {Jozsa}, \citenamefont {Peres},\ and\ \citenamefont
  {Wootters}}]{PhysRevLett.70.1895}%
  \BibitemOpen
  \bibfield  {author} {\bibinfo {author} {\bibfnamefont {C.~H.}\ \bibnamefont
  {Bennett}}, \bibinfo {author} {\bibfnamefont {G.}~\bibnamefont {Brassard}},
  \bibinfo {author} {\bibfnamefont {C.}~\bibnamefont {Cr\'epeau}}, \bibinfo
  {author} {\bibfnamefont {R.}~\bibnamefont {Jozsa}}, \bibinfo {author}
  {\bibfnamefont {A.}~\bibnamefont {Peres}},\ and\ \bibinfo {author}
  {\bibfnamefont {W.~K.}\ \bibnamefont {Wootters}},\ }\href
  {https://doi.org/10.1103/PhysRevLett.70.1895} {\bibfield  {journal} {\bibinfo
   {journal} {Phys. Rev. Lett.}\ }\textbf {\bibinfo {volume} {70}},\ \bibinfo
  {pages} {1895} (\bibinfo {year} {1993})}\BibitemShut {NoStop}%
\bibitem [{\citenamefont {Bravyi}\ \emph {et~al.}(2016)\citenamefont {Bravyi},
  \citenamefont {Smith},\ and\ \citenamefont {Smolin}}]{Bravyi_2016}%
  \BibitemOpen
  \bibfield  {author} {\bibinfo {author} {\bibfnamefont {S.}~\bibnamefont
  {Bravyi}}, \bibinfo {author} {\bibfnamefont {G.}~\bibnamefont {Smith}},\ and\
  \bibinfo {author} {\bibfnamefont {J.~A.}\ \bibnamefont {Smolin}},\ }\bibfield
   {journal} {\bibinfo  {journal} {Physical Review X}\ }\textbf {\bibinfo
  {volume} {6}},\ \href {https://doi.org/10.1103/physrevx.6.021043}
  {10.1103/physrevx.6.021043} (\bibinfo {year} {2016})\BibitemShut {NoStop}%
\bibitem [{\citenamefont {Bravyi}\ \emph {et~al.}(2019)\citenamefont {Bravyi},
  \citenamefont {Browne}, \citenamefont {Calpin}, \citenamefont {Campbell},
  \citenamefont {Gosset},\ and\ \citenamefont {Howard}}]{Bravyi_2019}%
  \BibitemOpen
  \bibfield  {author} {\bibinfo {author} {\bibfnamefont {S.}~\bibnamefont
  {Bravyi}}, \bibinfo {author} {\bibfnamefont {D.}~\bibnamefont {Browne}},
  \bibinfo {author} {\bibfnamefont {P.}~\bibnamefont {Calpin}}, \bibinfo
  {author} {\bibfnamefont {E.}~\bibnamefont {Campbell}}, \bibinfo {author}
  {\bibfnamefont {D.}~\bibnamefont {Gosset}},\ and\ \bibinfo {author}
  {\bibfnamefont {M.}~\bibnamefont {Howard}},\ }\href
  {https://doi.org/10.22331/q-2019-09-02-181} {\bibfield  {journal} {\bibinfo
  {journal} {Quantum}\ }\textbf {\bibinfo {volume} {3}},\ \bibinfo {pages}
  {181} (\bibinfo {year} {2019})}\BibitemShut {NoStop}%
\bibitem [{\citenamefont {Seddon}\ \emph {et~al.}(2021)\citenamefont {Seddon},
  \citenamefont {Regula}, \citenamefont {Pashayan}, \citenamefont {Ouyang},\
  and\ \citenamefont {Campbell}}]{Seddon_2021}%
  \BibitemOpen
  \bibfield  {author} {\bibinfo {author} {\bibfnamefont {J.~R.}\ \bibnamefont
  {Seddon}}, \bibinfo {author} {\bibfnamefont {B.}~\bibnamefont {Regula}},
  \bibinfo {author} {\bibfnamefont {H.}~\bibnamefont {Pashayan}}, \bibinfo
  {author} {\bibfnamefont {Y.}~\bibnamefont {Ouyang}},\ and\ \bibinfo {author}
  {\bibfnamefont {E.~T.}\ \bibnamefont {Campbell}},\ }\bibfield  {journal}
  {\bibinfo  {journal} {{PRX} Quantum}\ }\textbf {\bibinfo {volume} {2}},\
  \href {https://doi.org/10.1103/prxquantum.2.010345}
  {10.1103/prxquantum.2.010345} (\bibinfo {year} {2021})\BibitemShut {NoStop}%
\bibitem [{\citenamefont {Seddon}\ and\ \citenamefont
  {Campbell}(2019)}]{Seddon_2019}%
  \BibitemOpen
  \bibfield  {author} {\bibinfo {author} {\bibfnamefont {J.~R.}\ \bibnamefont
  {Seddon}}\ and\ \bibinfo {author} {\bibfnamefont {E.~T.}\ \bibnamefont
  {Campbell}},\ }\href {https://doi.org/10.1098/rspa.2019.0251} {\bibfield
  {journal} {\bibinfo  {journal} {Proceedings of the Royal Society A:
  Mathematical, Physical and Engineering Sciences}\ }\textbf {\bibinfo {volume}
  {475}},\ \bibinfo {pages} {20190251} (\bibinfo {year} {2019})}\BibitemShut
  {NoStop}%
\bibitem [{\citenamefont {Bu}\ and\ \citenamefont {Koh}(2019)}]{Bu_2019}%
  \BibitemOpen
  \bibfield  {author} {\bibinfo {author} {\bibfnamefont {K.}~\bibnamefont
  {Bu}}\ and\ \bibinfo {author} {\bibfnamefont {D.~E.}\ \bibnamefont {Koh}},\
  }\bibfield  {journal} {\bibinfo  {journal} {Physical Review Letters}\
  }\textbf {\bibinfo {volume} {123}},\ \href
  {https://doi.org/10.1103/physrevlett.123.170502}
  {10.1103/physrevlett.123.170502} (\bibinfo {year} {2019})\BibitemShut
  {NoStop}%
\bibitem [{\citenamefont {Bu}\ and\ \citenamefont {Koh}(2022)}]{Bu_2022}%
  \BibitemOpen
  \bibfield  {author} {\bibinfo {author} {\bibfnamefont {K.}~\bibnamefont
  {Bu}}\ and\ \bibinfo {author} {\bibfnamefont {D.~E.}\ \bibnamefont {Koh}},\
  }\href {https://doi.org/10.1007/s00220-022-04320-1} {\bibfield  {journal}
  {\bibinfo  {journal} {Communications in Mathematical Physics}\ }\textbf
  {\bibinfo {volume} {390}},\ \bibinfo {pages} {471} (\bibinfo {year}
  {2022})}\BibitemShut {NoStop}%
\bibitem [{\citenamefont {Yoshida}(2021)}]{Yoshida2021Recovery}%
  \BibitemOpen
  \bibfield  {author} {\bibinfo {author} {\bibfnamefont {B.}~\bibnamefont
  {Yoshida}},\ }\href {https://doi.org/10.48550/ARXIV.2106.15628} {\bibinfo
  {title} {Recovery algorithms for clifford hayden-preskill problem}} (\bibinfo
  {year} {2021})\BibitemShut {NoStop}%
\bibitem [{\citenamefont {Yoshida}\ and\ \citenamefont
  {Kitaev}(2017)}]{YoshidaEfficient}%
  \BibitemOpen
  \bibfield  {author} {\bibinfo {author} {\bibfnamefont {B.}~\bibnamefont
  {Yoshida}}\ and\ \bibinfo {author} {\bibfnamefont {A.}~\bibnamefont
  {Kitaev}},\ }\href {https://doi.org/10.48550/ARXIV.1710.03363} {\bibinfo
  {title} {Efficient decoding for the hayden-preskill protocol}} (\bibinfo
  {year} {2017})\BibitemShut {NoStop}%
\bibitem [{\citenamefont {Montanaro}\ and\ \citenamefont
  {Osborne}(2008)}]{Montanaro2008}%
  \BibitemOpen
  \bibfield  {author} {\bibinfo {author} {\bibfnamefont {A.}~\bibnamefont
  {Montanaro}}\ and\ \bibinfo {author} {\bibfnamefont {T.~J.}\ \bibnamefont
  {Osborne}},\ }\href {https://doi.org/10.48550/ARXIV.0810.2435} {\bibinfo
  {title} {Quantum boolean functions}} (\bibinfo {year} {2008})\BibitemShut
  {NoStop}%
\bibitem [{\citenamefont {Bu}\ \emph {et~al.}(2022{\natexlab{b}})\citenamefont
  {Bu}, \citenamefont {Garcia}, \citenamefont {Jaffe}, \citenamefont {Koh},\
  and\ \citenamefont {Li}}]{Bu2022}%
  \BibitemOpen
  \bibfield  {author} {\bibinfo {author} {\bibfnamefont {K.}~\bibnamefont
  {Bu}}, \bibinfo {author} {\bibfnamefont {R.~J.}\ \bibnamefont {Garcia}},
  \bibinfo {author} {\bibfnamefont {A.}~\bibnamefont {Jaffe}}, \bibinfo
  {author} {\bibfnamefont {D.~E.}\ \bibnamefont {Koh}},\ and\ \bibinfo {author}
  {\bibfnamefont {L.}~\bibnamefont {Li}},\ }\href
  {https://doi.org/10.48550/ARXIV.2204.12051} {\bibinfo {title} {Complexity of
  quantum circuits via sensitivity, magic, and coherence}} (\bibinfo {year}
  {2022}{\natexlab{b}})\BibitemShut {NoStop}%
\bibitem [{Note1()}]{Note1}%
  \BibitemOpen
  \bibinfo {note} {For example, the average R\'enyi-2 entanglement entropy,
  $\protect \overline {S}^{(2)}(\rho )\equiv \protect \frac {1}{2^{n}}\DOTSB
  \sum@ \slimits@ _{A\in [n]}-\protect \mathrm {log}\protect \mathrm {Tr}\left
  \{\rho _A^2\right \}$, where $\rho _A$ is the reduced state of an $n$-qudit
  state $\rho $ on a subsystem $A$, satisfies $\protect \overline
  {S}^{(2)}(U\rho U^\dagger )=\protect \overline {S}^{(2)}(\rho )$, where $U$
  is a non-entangling unitary \cite {Bu2022}.}\BibitemShut {Stop}%
\bibitem [{\citenamefont {Gottesman}\ and\ \citenamefont
  {Chuang}(1999)}]{Gottesman_1999}%
  \BibitemOpen
  \bibfield  {author} {\bibinfo {author} {\bibfnamefont {D.}~\bibnamefont
  {Gottesman}}\ and\ \bibinfo {author} {\bibfnamefont {I.~L.}\ \bibnamefont
  {Chuang}},\ }\href {https://doi.org/10.1038/46503} {\bibfield  {journal}
  {\bibinfo  {journal} {Nature}\ }\textbf {\bibinfo {volume} {402}},\ \bibinfo
  {pages} {390} (\bibinfo {year} {1999})}\BibitemShut {NoStop}%
\bibitem [{\citenamefont {Cui}\ \emph {et~al.}(2017)\citenamefont {Cui},
  \citenamefont {Gottesman},\ and\ \citenamefont
  {Krishna}}]{PhysRevA.95.012329}%
  \BibitemOpen
  \bibfield  {author} {\bibinfo {author} {\bibfnamefont {S.~X.}\ \bibnamefont
  {Cui}}, \bibinfo {author} {\bibfnamefont {D.}~\bibnamefont {Gottesman}},\
  and\ \bibinfo {author} {\bibfnamefont {A.}~\bibnamefont {Krishna}},\ }\href
  {https://doi.org/10.1103/PhysRevA.95.012329} {\bibfield  {journal} {\bibinfo
  {journal} {Phys. Rev. A}\ }\textbf {\bibinfo {volume} {95}},\ \bibinfo
  {pages} {012329} (\bibinfo {year} {2017})}\BibitemShut {NoStop}%
\bibitem [{\citenamefont {Vermersch}\ \emph {et~al.}(2019)\citenamefont
  {Vermersch}, \citenamefont {Elben}, \citenamefont {Sieberer}, \citenamefont
  {Yao},\ and\ \citenamefont {Zoller}}]{PhysRevX.9.021061}%
  \BibitemOpen
  \bibfield  {author} {\bibinfo {author} {\bibfnamefont {B.}~\bibnamefont
  {Vermersch}}, \bibinfo {author} {\bibfnamefont {A.}~\bibnamefont {Elben}},
  \bibinfo {author} {\bibfnamefont {L.~M.}\ \bibnamefont {Sieberer}}, \bibinfo
  {author} {\bibfnamefont {N.~Y.}\ \bibnamefont {Yao}},\ and\ \bibinfo {author}
  {\bibfnamefont {P.}~\bibnamefont {Zoller}},\ }\href
  {https://doi.org/10.1103/PhysRevX.9.021061} {\bibfield  {journal} {\bibinfo
  {journal} {Phys. Rev. X}\ }\textbf {\bibinfo {volume} {9}},\ \bibinfo {pages}
  {021061} (\bibinfo {year} {2019})}\BibitemShut {NoStop}%
\bibitem [{\citenamefont {Joshi}\ \emph {et~al.}(2020)\citenamefont {Joshi},
  \citenamefont {Elben}, \citenamefont {Vermersch}, \citenamefont {Brydges},
  \citenamefont {Maier}, \citenamefont {Zoller}, \citenamefont {Blatt},\ and\
  \citenamefont {Roos}}]{Joshi_2020}%
  \BibitemOpen
  \bibfield  {author} {\bibinfo {author} {\bibfnamefont {M.~K.}\ \bibnamefont
  {Joshi}}, \bibinfo {author} {\bibfnamefont {A.}~\bibnamefont {Elben}},
  \bibinfo {author} {\bibfnamefont {B.}~\bibnamefont {Vermersch}}, \bibinfo
  {author} {\bibfnamefont {T.}~\bibnamefont {Brydges}}, \bibinfo {author}
  {\bibfnamefont {C.}~\bibnamefont {Maier}}, \bibinfo {author} {\bibfnamefont
  {P.}~\bibnamefont {Zoller}}, \bibinfo {author} {\bibfnamefont
  {R.}~\bibnamefont {Blatt}},\ and\ \bibinfo {author} {\bibfnamefont {C.~F.}\
  \bibnamefont {Roos}},\ }\bibfield  {journal} {\bibinfo  {journal} {Physical
  Review Letters}\ }\textbf {\bibinfo {volume} {124}},\ \href
  {https://doi.org/10.1103/physrevlett.124.240505}
  {10.1103/physrevlett.124.240505} (\bibinfo {year} {2020})\BibitemShut
  {NoStop}%
\bibitem [{\citenamefont {Landsman}\ \emph {et~al.}(2019)\citenamefont
  {Landsman}, \citenamefont {Figgatt}, \citenamefont {Schuster}, \citenamefont
  {Linke}, \citenamefont {Yoshida}, \citenamefont {Yao},\ and\ \citenamefont
  {Monroe}}]{Landsman_2019}%
  \BibitemOpen
  \bibfield  {author} {\bibinfo {author} {\bibfnamefont {K.~A.}\ \bibnamefont
  {Landsman}}, \bibinfo {author} {\bibfnamefont {C.}~\bibnamefont {Figgatt}},
  \bibinfo {author} {\bibfnamefont {T.}~\bibnamefont {Schuster}}, \bibinfo
  {author} {\bibfnamefont {N.~M.}\ \bibnamefont {Linke}}, \bibinfo {author}
  {\bibfnamefont {B.}~\bibnamefont {Yoshida}}, \bibinfo {author} {\bibfnamefont
  {N.~Y.}\ \bibnamefont {Yao}},\ and\ \bibinfo {author} {\bibfnamefont
  {C.}~\bibnamefont {Monroe}},\ }\href
  {https://doi.org/10.1038/s41586-019-0952-6} {\bibfield  {journal} {\bibinfo
  {journal} {Nature}\ }\textbf {\bibinfo {volume} {567}},\ \bibinfo {pages}
  {61} (\bibinfo {year} {2019})}\BibitemShut {NoStop}%
\bibitem [{\citenamefont {Swingle}\ \emph {et~al.}(2016)\citenamefont
  {Swingle}, \citenamefont {Bentsen}, \citenamefont {Schleier-Smith},\ and\
  \citenamefont {Hayden}}]{PhysRevA.94.040302}%
  \BibitemOpen
  \bibfield  {author} {\bibinfo {author} {\bibfnamefont {B.}~\bibnamefont
  {Swingle}}, \bibinfo {author} {\bibfnamefont {G.}~\bibnamefont {Bentsen}},
  \bibinfo {author} {\bibfnamefont {M.}~\bibnamefont {Schleier-Smith}},\ and\
  \bibinfo {author} {\bibfnamefont {P.}~\bibnamefont {Hayden}},\ }\href
  {https://doi.org/10.1103/PhysRevA.94.040302} {\bibfield  {journal} {\bibinfo
  {journal} {Phys. Rev. A}\ }\textbf {\bibinfo {volume} {94}},\ \bibinfo
  {pages} {040302} (\bibinfo {year} {2016})}\BibitemShut {NoStop}%
\bibitem [{\citenamefont {Gärttner}\ \emph {et~al.}(2017)\citenamefont
  {Gärttner}, \citenamefont {Bohnet}, \citenamefont {Safavi-Naini},
  \citenamefont {Wall}, \citenamefont {Bollinger},\ and\ \citenamefont
  {Rey}}]{G_rttner_2017}%
  \BibitemOpen
  \bibfield  {author} {\bibinfo {author} {\bibfnamefont {M.}~\bibnamefont
  {Gärttner}}, \bibinfo {author} {\bibfnamefont {J.~G.}\ \bibnamefont
  {Bohnet}}, \bibinfo {author} {\bibfnamefont {A.}~\bibnamefont
  {Safavi-Naini}}, \bibinfo {author} {\bibfnamefont {M.~L.}\ \bibnamefont
  {Wall}}, \bibinfo {author} {\bibfnamefont {J.~J.}\ \bibnamefont
  {Bollinger}},\ and\ \bibinfo {author} {\bibfnamefont {A.~M.}\ \bibnamefont
  {Rey}},\ }\href {https://doi.org/10.1038/nphys4119} {\bibfield  {journal}
  {\bibinfo  {journal} {Nature Physics}\ }\textbf {\bibinfo {volume} {13}},\
  \bibinfo {pages} {781} (\bibinfo {year} {2017})}\BibitemShut {NoStop}%
\bibitem [{\citenamefont {Wei}\ \emph {et~al.}(2018)\citenamefont {Wei},
  \citenamefont {Ramanathan},\ and\ \citenamefont
  {Cappellaro}}]{PhysRevLett.120.070501}%
  \BibitemOpen
  \bibfield  {author} {\bibinfo {author} {\bibfnamefont {K.~X.}\ \bibnamefont
  {Wei}}, \bibinfo {author} {\bibfnamefont {C.}~\bibnamefont {Ramanathan}},\
  and\ \bibinfo {author} {\bibfnamefont {P.}~\bibnamefont {Cappellaro}},\
  }\href {https://doi.org/10.1103/PhysRevLett.120.070501} {\bibfield  {journal}
  {\bibinfo  {journal} {Phys. Rev. Lett.}\ }\textbf {\bibinfo {volume} {120}},\
  \bibinfo {pages} {070501} (\bibinfo {year} {2018})}\BibitemShut {NoStop}%
\bibitem [{\citenamefont {Leone}\ \emph
  {et~al.}(2022{\natexlab{a}})\citenamefont {Leone}, \citenamefont {Oliviero},\
  and\ \citenamefont {Hamma}}]{leone2022}%
  \BibitemOpen
  \bibfield  {author} {\bibinfo {author} {\bibfnamefont {L.}~\bibnamefont
  {Leone}}, \bibinfo {author} {\bibfnamefont {S.~F.~E.}\ \bibnamefont
  {Oliviero}},\ and\ \bibinfo {author} {\bibfnamefont {A.}~\bibnamefont
  {Hamma}},\ }\href {https://doi.org/10.48550/ARXIV.2204.02995} {\bibinfo
  {title} {Magic hinders quantum certification}} (\bibinfo {year}
  {2022}{\natexlab{a}})\BibitemShut {NoStop}%
\bibitem [{\citenamefont {Leone}\ \emph
  {et~al.}(2022{\natexlab{b}})\citenamefont {Leone}, \citenamefont {Oliviero},\
  and\ \citenamefont {Hamma}}]{PhysRevLett.128.050402}%
  \BibitemOpen
  \bibfield  {author} {\bibinfo {author} {\bibfnamefont {L.}~\bibnamefont
  {Leone}}, \bibinfo {author} {\bibfnamefont {S.~F.~E.}\ \bibnamefont
  {Oliviero}},\ and\ \bibinfo {author} {\bibfnamefont {A.}~\bibnamefont
  {Hamma}},\ }\href {https://doi.org/10.1103/PhysRevLett.128.050402} {\bibfield
   {journal} {\bibinfo  {journal} {Phys. Rev. Lett.}\ }\textbf {\bibinfo
  {volume} {128}},\ \bibinfo {pages} {050402} (\bibinfo {year}
  {2022}{\natexlab{b}})}\BibitemShut {NoStop}%
\bibitem [{\citenamefont {Oliviero}\ \emph {et~al.}(2022)\citenamefont
  {Oliviero}, \citenamefont {Leone}, \citenamefont {Hamma},\ and\ \citenamefont
  {Lloyd}}]{Oliviero2022}%
  \BibitemOpen
  \bibfield  {author} {\bibinfo {author} {\bibfnamefont {S.~F.~E.}\
  \bibnamefont {Oliviero}}, \bibinfo {author} {\bibfnamefont {L.}~\bibnamefont
  {Leone}}, \bibinfo {author} {\bibfnamefont {A.}~\bibnamefont {Hamma}},\ and\
  \bibinfo {author} {\bibfnamefont {S.}~\bibnamefont {Lloyd}},\ }\href
  {https://doi.org/10.48550/ARXIV.2204.00015} {\bibinfo {title} {Measuring
  magic on a quantum processor}} (\bibinfo {year} {2022})\BibitemShut {NoStop}%
\bibitem [{Note2()}]{Note2}%
  \BibitemOpen
  \bibinfo {note} {In this experiment, the OTOC is defined as an expectation
  value with respect to a pure state, but we utilize the maximally mixed state,
  as in Eq.~\protect \textup {\hbox {\mathsurround \z@ \protect \normalfont
  (\ignorespaces \ref {Eq:OTOCdef}\unskip \@@italiccorr )}}.}\BibitemShut
  {Stop}%
\bibitem [{\citenamefont {Ahmadi}\ and\ \citenamefont
  {Greplova}(2022)}]{Ahmadi2022}%
  \BibitemOpen
  \bibfield  {author} {\bibinfo {author} {\bibfnamefont {A.}~\bibnamefont
  {Ahmadi}}\ and\ \bibinfo {author} {\bibfnamefont {E.}~\bibnamefont
  {Greplova}},\ }\href {https://doi.org/10.48550/ARXIV.2204.11236} {\bibinfo
  {title} {Quantifying quantum computational complexity via information
  scrambling}} (\bibinfo {year} {2022})\BibitemShut {NoStop}%
\bibitem [{Note3()}]{Note3}%
  \BibitemOpen
  \bibinfo {note} {In fact, a small value of $\left | \protect \mathrm
  {OTOC}(U) \right |$ is an indication of the magic of $U$, since it bounds
  $O_M(U)$.}\BibitemShut {Stop}%
\bibitem [{\citenamefont {Yoshida}\ and\ \citenamefont
  {Yao}(2019)}]{PhysRevX.9.011006}%
  \BibitemOpen
  \bibfield  {author} {\bibinfo {author} {\bibfnamefont {B.}~\bibnamefont
  {Yoshida}}\ and\ \bibinfo {author} {\bibfnamefont {N.~Y.}\ \bibnamefont
  {Yao}},\ }\href {https://doi.org/10.1103/PhysRevX.9.011006} {\bibfield
  {journal} {\bibinfo  {journal} {Phys. Rev. X}\ }\textbf {\bibinfo {volume}
  {9}},\ \bibinfo {pages} {011006} (\bibinfo {year} {2019})}\BibitemShut
  {NoStop}%
\bibitem [{\citenamefont {Leone}\ \emph
  {et~al.}(2022{\natexlab{c}})\citenamefont {Leone}, \citenamefont {Oliviero},
  \citenamefont {Piemontese}, \citenamefont {True},\ and\ \citenamefont
  {Hamma}}]{Leone2022Mocking}%
  \BibitemOpen
  \bibfield  {author} {\bibinfo {author} {\bibfnamefont {L.}~\bibnamefont
  {Leone}}, \bibinfo {author} {\bibfnamefont {S.~F.~E.}\ \bibnamefont
  {Oliviero}}, \bibinfo {author} {\bibfnamefont {S.}~\bibnamefont
  {Piemontese}}, \bibinfo {author} {\bibfnamefont {S.}~\bibnamefont {True}},\
  and\ \bibinfo {author} {\bibfnamefont {A.}~\bibnamefont {Hamma}},\ }\href
  {https://doi.org/10.48550/ARXIV.2206.06385} {\bibinfo {title} {To learn a
  mocking-black hole}} (\bibinfo {year} {2022}{\natexlab{c}})\BibitemShut
  {NoStop}%
\end{thebibliography}%

\begin{appendix}

\newpage

\onecolumngrid

\section{Relation between scrambling mechanisms}
To establish a relation between the two scrambling mechanisms, we use the magic entropy, introduced in \cite{Bu2022}:
\begin{equation}
    M(U)=\max_{P_{\vec{a}}\in\mathcal{P}_d^{\otimes n}, W(P_{\vec{a}})=1}H(U^\dagger P_{\vec{a}}U),
\end{equation}
where ${H(O)=-\sum_{\vec{a}\in \mathcal{V}_d^n}P_O[\vec{a}]\mathrm{log}P_O[\vec{a}]}$ is the quantum Fourier entropy (QFE)
of an $n$-qudit operator $O$ satisfying ${\norm{O}_2=1}$.
The magic entropy is faithful in that ${M(V)\geq 0}$ for any unitary $V$ and $M(U)=0$ iff $U$ is a Clifford unitary. It also satisfies right-invariance, meaning $M(VU)=M(V)$ for any unitary $V$ and a Clifford unitary $U$.
%The magic entropy is an operationally meaningful tool for studying magic scrambling. 
%Suppose that $U$ maps a weight-1 Pauli operator $P_{\vec{a}}$ to a uniform superposition of $r$ Pauli operators. Then $H(U^\dagger P_{\vec{a}}U)=\mathrm{log}(r)$ and increasing $r$ will increase the QFE. In the case of a non-uniform superposition, the QFE can be small. 
The QFE, and subsequently the magic entropy, quantifies operator entanglement, as it can be large when a unitary maps a Pauli operator to a uniform sum of many Pauli operators. 
%The magic entropy bounds the Pauli growth. 
Assuming that the quantum Fourier entropy-influence conjecture \cite{Bu2022} holds,  i.e. $H(O)\leq cW(O)$ for a Hermitian, $n$-qubit operator $O$ satisfying $O^2=I$ and some constant $c$, then
\begin{eqnarray}
M(U)\leq c(G(U)+1).
\end{eqnarray}

\begin{proof}
From the quantum Fourier entropy-infuence conjecture,
\begin{equation}
\begin{split}
	H(U^\dagger P_{\vec{a}} U)
	&\leq cW(U^\dagger P_{\vec{a}} U)\\
	&= c(W(U^\dagger P_{\vec{a}} U)-1+1).
\end{split}
\end{equation}
Taking the maximization over all weight-1 Pauli operators yields
\begin{equation}
\begin{split}
	\max_{P_{\vec{a}}\in\mathcal{P}_d^{\otimes n}, W(P_{\vec{a}})=1} H(U^\dagger P_{\vec{a}} U)
	&\leq c \left[\max_{P_{\vec{a}}\in\mathcal{P}_d^{\otimes n}, W(P_{\vec{a}})=1} (W(U^\dagger P_{\vec{a}} U)-1)+1\right]\\
	&\leq c \left[\max_{\substack{O:  \norm{O}_2=1, W(O)=1,\\\tr{O}=0}}(W(U^\dagger O U)-1)+1\right].
\end{split}
\end{equation}
This implies that
\begin{equation}
	M(U)\leq c(G(U)+1).
\end{equation}

\end{proof}

\section{Pauli growth is an entanglement scrambling monotone}
\label{Sec:PauliGrowth}
This proof is based on results in~\cite{Bu2022}. We first prove that $G$ is faithful. First, we show that $G(V)\geq 0$ for any unitary $V$. Consider a traceless, normalized operator $O$. The operator $V^\dagger O V$ must also be traceless and normalized. That is, the Pauli expansion of $V^\dagger O V$ contains no identity term. Hence each Pauli operator in the expansion has a Pauli weight of at least 1. The average Pauli weight of $V^\dagger O V$ is therefore lower bounded by 1. Explicitly,
\begin{equation}\label{Eq:Weightbound}
\begin{split}
	W(V^\dagger O V)
	&=\sum_{\vec{a}\in \mathcal{V}_d^n}\abs{\vec{a}}P_{V^\dagger OV}[\vec{a}]\\
		&=\sum_{\substack{\vec{a}\in \mathcal{V}_d^n,\\ \vec{a}\neq (0,0)^n}}\abs{\vec{a}}P_{V^\dagger OV}[\vec{a}]\\
		&\geq \sum_{\substack{\vec{a}\in \mathcal{V}_d^n,\\ \vec{a}\neq (0,0)^n}} \min_{\substack{\vec{a}\in \mathcal{V}_d^n,\\ \vec{a}\neq (0,0)^n}}\left[\abs{\vec{a}}\right]P_{V^\dagger OV}[\vec{a}]\\
		&= \sum_{\substack{\vec{a}\in \mathcal{V}_d^n,\\ \vec{a}\neq (0,0)^n}} P_{V^\dagger OV}[\vec{a}]\\
	&=1.
\end{split}
\end{equation}
In the second line, we use that $P_{V^\dagger OV}[(0,0)^n]=0$, since $V^\dagger O V$ has no identity component. Since $W(V^\dagger O V)\geq 1$ for any traceless, normalized operator $O$, it follows from the definition of the Pauli growth that $G(V)\geq 0$.

Now define $U$ such that $G(U)=0$. We will prove that $U$ must be generated by a product of swap gates and single-qudit unitaries.  Since $G=0$, then 

\begin{equation}
\max_{\substack{O:  \norm{O}_2=1, W(O)=1,\\\tr{O}=0}}\left[W(U^\dagger O U)-1\right]= 0.
\end{equation}
Hence, if $O$ is a traceless, normalized operator satisfying $W(O)=1$, then ${W(U^\dagger O U)\leq1}$.  From Ineq.~\ref{Eq:Weightbound}, it follows that ${W(U^\dagger O U)\geq1}$. These two conditions require that ${W(U^\dagger O U)=1}$. 

We now consider the case where $O$ is a weight-1 Pauli operator. Let $X_{1}$ denote an $n$-qudit Pauli operator which acts $X$ only on the first qudit. For simplicity, we refer to this as a local operator. Since $\tr{U^\dagger X_{1} U}$=0, then $U^\dagger X_{1} U$ has no identity component in its Pauli expansion. This, along with the condition $W(U^\dagger X_1 U)=1$, implies that the expansion of $U^\dagger X_{1} U$ in the Pauli basis must only contain weight-1 Pauli operators. Therefore,
\begin{equation}\label{Eq:XEvolved}
	U^\dagger X_1U = \sum_{i=1}^n\left[\alpha_i Q_i^X+\beta_iQ_i^Y +\gamma_iQ_i^Z\right],
\end{equation}
where $Q_i^X=\sum_{j=1}^{d-1}c_{ij}X_i^j$, $Q_i^Y=\sum_{j,k=1}^{d-1}c_{ijk}X_i^jZ_i^k$, and $Q_i^Z=\sum_{j=1}^{d-1}c_{ij}Z_i^j$. The operator $U^\dagger X_1U$ satisfies ${\abs{U^\dagger X_1U}^2=I}$. Consider the qubit case where $d=2$, so that $(U^\dagger X_1U)^2=I$. When expanding $(U^\dagger X_1U)^2$ in the Pauli basis, for fixed integers $l$ and $l'$ with $l\neq l'$, we get terms like $\alpha_l \alpha_{l'} Q_l^X\otimes Q_{l'}^X$ and $2\alpha_l \beta_{l'} Q_l^X\otimes Q_{l'}^Y$.  These terms must vanish to ensure that $(U^\dagger X_1U)^2=I$. Hence, for the qubit case, there exists an integer $l$ such that
\begin{equation}\label{Eq:XSimple}
	U^\dagger X_1U = \alpha_l Q_l^X+\beta_lQ_l^Y +\gamma_lQ_l^Z.
\end{equation}
Therefore, $U^\dagger X_1U $ is a local operator, as it acts non-trivially only on the $l$-th qubit. Now take the qudit case where ${d>2}$. Again taking Eq.~\eqref{Eq:XEvolved} and expanding $\abs{U^\dagger X_1U}^2$ in the Pauli basis, we get terms like ${\alpha_l ^* \alpha_{l'}Q_l^{X\dagger}  \otimes Q_{l'}^X}$ and ${\alpha_l ^* \beta_{l'}Q_l^{X\dagger}  \otimes Q_{l'}^Y}$. These terms must vanish since $\abs{U^\dagger X_1U}^2=I$. Therefore, in the qudit case, Eq.~\eqref{Eq:XSimple} still holds. Similarly, for a fixed integer $m$,
\begin{equation}
	U^\dagger Z_1U = \alpha_m' Q_m^X+\beta_m'Q_m^Y +\gamma_m' Q_m^Z,
\end{equation}
where the coefficients $\alpha_m'$, $\beta_m'$, and $\gamma_m'$ are defined similarly to those in Eq.~\eqref{Eq:XEvolved}.
If $l\neq m$, then $[U^\dagger X_1U,U^\dagger Z_1U]=0$, implying that $[X_1,Z_1]=0$. This commutator cannot vanish, hence $l=m$. The operator $U^\dagger X_1Z_1 U$ is also a local operator since $U^\dagger X_1Z_1 U=(U^\dagger X_1U)(U^\dagger Z_1U)$. In general, for a weight-one Pauli operator $P_1\in \mathcal{P}_2^{\otimes n}$, $U^\dagger P_1 U$ is a local operator acting non-trivially on site $l$. Therefore, we can write
\begin{equation}
\begin{split}
	U^\dagger P_1 U = \mathbb{SW}_{1,l}(V_1^\dagger \otimes V_{(1)}'^\dagger)P_1(V_1\otimes V_{(1)}')\mathbb{SW}_{1,l},
\end{split}
\end{equation}
where $\mathbb{SW}_{1,l}$ is the swap operator between site $1$ and site $l$, $V_1$ is a unitary defined on the first qudit, and $V_{(1)}'$ is a unitary defined on the remaining qudits. Unitaries $V_1$ and $V_{(1)}'$ must have a vanishing $G$. Therefore $U=\left(V_1\otimes V_{(1)}'\right)\mathbb{SW}_{1,l}$.

 Let $P_2$ be a Pauli operator acting non-trivially only on the second site. Unitary $V_{(1)}'$ maps $P_2$ to an operator on the site $m\neq 1$. Then
 \begin{equation}
 \begin{split}
 	U^\dagger P_2 U 
 	&= \mathbb{SW}_{1,l}(V_1^\dagger \otimes V_{(1)}'^\dagger)P_2(V_1\otimes V_{(1)}')\mathbb{SW}_{1,l}\\
  	&= \mathbb{SW}_{1,l}\mathbb{SW}_{2,m}(V_1^\dagger \otimes V_2^\dagger \otimes V_{(2)}'^\dagger )P_2(V_1\otimes V_{2}\otimes V_{(2)}')\mathbb{SW}_{2,m}\mathbb{SW}_{1,l},
\end{split}
 \end{equation}
 where $V_2$ is a unitary defined on the second qudit and $V_{(2)}' $ is a unitary defined on the remaining $n-2$ qudits. Both $V_2$ and $V_{(2)}'$ have a vanishing $G$. Therefore $U=(V_1\otimes V_2 \otimes V_{(2)}')\mathbb{SW}_{2,m}\mathbb{SW}_{1,l}$. 
 We can continue this process by letting $U$ act on $P_k$, a Pauli operator which acts non-trivially only on the $k$-th site, until we find that, if $G(U)=0$, then $U$ is a product of swap operators and single-qudit unitaries. 
 
 Next, we prove that if $U$ is generated by swap operators and single-qudit unitaries, then $G(U)=0$. The average Pauli weight of a normalized operator $O$ can be written as
\begin{equation}
\begin{split}
	W( O )
	&=\sum_{j\in[n]}\sum_{\vec{a}:a_j \neq (0,0)}P_{O}[\vec{a}]\\
	&=\sum_{j\in[n]}\frac{1}{d^{2n}}\sum_{\vec{a}:a_j \neq (0,0)}\abs{\tr{O  P_{\vec{a}}}}^2\\
	&=\sum_{j\in[n]}\left[\frac{1}{d^{2n}}\sum_{\vec{a}}\abs{\tr{O P_{\vec{a}}}}^2-\frac{1}{d^{2n}}\sum_{\vec{a}:a_j=(0,0)}\abs{\tr{O P_{\vec{a}}}}^2\right]\\
	&=\sum_{j\in[n]}\left[\frac{1}{d^n}\tr{\abs{O}^2}-\frac{1}{d^2}\frac{1}{d^{n-1}}\ptr{\overline{j}}{\abs{\ptr{j}{O}}^2}\right]\\
	&=\sum_{j\in[n]}\frac{1}{d^n}\left[\tr{\abs{O}^2}-\frac{1}{d}\ptr{\overline{j}}{\abs{\ptr{j}{O}}^2}\right].
\end{split}
\end{equation}
The average Pauli weight of $U^\dagger OU$ when $U=\otimes_{i=1}^nU_i$ is
\begin{equation}\label{Eq:SingleWeight}
\begin{split}
W( U^\dagger O U)
	&=\sum_{j\in[n]}\frac{1}{d^n}\left[\tr{\abs{ U^\dagger O U}^2}-\frac{1}{d}\ptr{\overline{j}}{\abs{\ptr{j}{ U^\dagger O U}}^2}\right]\\
	&=\sum_{j\in[n]}\frac{1}{d^n}\left[\tr{\abs{ O }^2}-\frac{1}{d}\ptr{\overline{j}}{\abs{\ptr{j}{  O }}^2}\right]\\
	&=W(O).
\end{split}
\end{equation}
The average Pauli weight of $U^\dagger OU$ when $U$ is as a product of swap operators is
\begin{equation}
\begin{split}\label{Eq:PermWeight}
W(U^\dagger O U)
	&=\sum_{j\in[n]}\sum_{\vec{a}:a_j \neq (0,0)}P_{U^\dagger O U}[\vec{a}]\\
	&=\sum_{j\in[n]}\frac{1}{d^{2n}}\sum_{\vec{a}:a_j \neq (0,0)}\abs{\tr{U^\dagger O U  P_{\vec{a}}}}^2\\
	&=\sum_{j\in[n]}\frac{1}{d^{2n}}\sum_{\vec{a}:a_j \neq (0,0)}\abs{\tr{O U  P_{\vec{a}}U^\dagger }}^2\\
	&=\sum_{k\in[n]}\frac{1}{d^{2n}}\sum_{\vec{a}:a_k \neq (0,0)}\abs{\tr{O P_{\vec{a}}}}^2\\
	&=W(O).
\end{split}
\end{equation}
In line four, we use that $U$ permutes the position of the strictly non-identity, single-qudit Pauli operator from the site $j$ to the site $k$. The sum over $j$ is invariant under the relabeling $j\rightarrow k$. From Eqs.~\eqref{Eq:SingleWeight} and \eqref{Eq:PermWeight}, it follows that if $U$ is generated by swap operators and single-qudit unitaries, then $W(U^\dagger O U)=W(O)$ for any normalized operator $O$. Hence, $G(U)=0$. Therefore, $G(U)=0$ iff $U$ is generated by swap operators and single-qudit unitaries, making $G$ faithful.

Furthermore, by defining the function $G_{\mathrm{Pauli}}(U)\equiv \max_{P_{\vec{a}}\in \mathcal{P}_d^{\otimes n},  W(P_{\vec{a}})=1}\left[W(U^\dagger P_{\vec{a}} U)-1\right]$, one can also show through similar arguments that $G_{\mathrm{Pauli}}(U)=0$ iff $U$ is non-entangling. Since $W(U^\dagger P_{\vec{a}}U)\geq 1$ when $P_{\vec{a}}$ is a non-identity Pauli operator, then $G_{\mathrm{Pauli}}(U)=0$ iff $W(U^\dagger P_{\vec{a}}U)=1$ for all weight-1 Pauli operators $P_{\vec{a}}$. This implies that free unitaries, i.e. unitaries which map a weight-1 Pauli operator to an operator with an average Pauli weight of 1, are generated only by swap gates and single-qudit unitaries.

We now prove the invariance condition for the Pauli growth. Let $V$ be any unitary and take $U_1$ and $U_2$ to be non-entangling unitaries. Then
\begin{equation}
\begin{split}
	G(U_1VU_2)
	&=\max_{\substack{O:  \norm{O}_2=1, W(O)=1,\\\tr{O}=0}}\left[W(U_2^\dagger V^\dagger U_1^\dagger O U_1V U_2)-1\right]\\
	&=\max_{\substack{O:  \norm{O}_2=1, W(O)=1,\\\tr{O}=0}}\left[W(V^\dagger U_1^\dagger OU_1V)-1\right]\\
	&=\max_{\substack{O':  \norm{O'}_2=1, W(O')=1,\\\tr{O'}=0}}\left[W(V^\dagger O' V)
	-1\right]\\
	&=G(V).
\end{split}
\end{equation}
In the third line, we define $O'=U_1^\dagger OU_1$. This operator satisfies $\norm{O'}_2=1$, $W(O')=1$, and $\tr{O'}=0$. Hence, the maximization is invariant under the transformation $O\rightarrow O'$.

\section{Proof of the OTOC-Pauli growth inequality}\label{Sec:ProofOTOCScram}
In the large $n$ limit, the OTOC of $U$ and the average Pauli weight satisfy the following relation  \cite{Bu2022}: 
\begin{equation}
	{\av{A}\av{P_A\neq I_A}\OTOC(U)=1-\frac{4}{3n}W(U^\dagger P_DU)},
\end{equation}
where $D$ is the $n$-th qubit, $A$ is any other single-qubit subsystem, $\av{A}$ is the uniform average over all choices of $A$, and $\av{P_A\neq I_A}$ is the uniform average over all non-identity Pauli operators on $A$.
The connection between the average Pauli weight and the OTOC has also been considered for special quantum circuits, such as Brownian quantum circuits \cite{PhysRevE.99.052212}.

The Pauli growth can be bounded by the average OTOC:
\begin{equation}
\begin{split}
	G(U)
	&= \max_{\substack{O:  \norm{O}_2=1, W(O)=1,\\\tr{O}=0}}\left[W(U^\dagger O U)-1\right]\\
	&\geq  W(U^\dagger P_D U)-1\\
	&= 
	\frac{3n}{4}\left(1-\av{A}\av{P_A\neq I_A}\OTOC(U)\right)-1.
\end{split}
\end{equation}
Rewriting this in terms of the average OTOC,
\begin{equation}
\begin{split}
\av{A}\av{P_A\neq I_A}\OTOC(U)\geq 1-\frac{4}{3n}\left(G(U)+1\right).
\end{split}
\end{equation}

\section{OTOC magic is a magic monotone}\label{Proof:LemmaOTOCMagic}
We prove that $O_M$ is faithful. First, we show that $O_M(U)\geq 0$ for any unitary $U$. This follows from 
\begin{equation}
\begin{split}
    O_M(U)
    &=\max_{P_{\vec{a}},P_{\vec{b}}\in \mathcal{P}_2^{\otimes n}}\left[1-\abs{\OTOC(U)}\right]\\
    &=1-\min_{P_{\vec{a}},P_{\vec{b}}\in \mathcal{P}_2^{\otimes n}}\left[\abs{\OTOC(U)}\right]\\
    &=1-\min_{P_{\vec{a}},P_{\vec{b}}\in \mathcal{P}_2^{\otimes n}}\left[\abs{\langle U^\dagger P_{\vec{a}} U P_{\vec{b}} U^\dagger P_{\vec{a}} U P_{\vec{b}}\rangle}\right]\\
    &\geq 1-\abs{\langle U^\dagger P_{\vec{a}'} U P_{\vec{b}'} U^\dagger P_{\vec{a}'} U P_{\vec{b}'}\rangle}\\
     &\geq 1-\norm{U^\dagger P_{\vec{a}'} U P_{\vec{b}'} U^\dagger P_{\vec{a}'} U P_{\vec{b}'}}_{\infty}\\
     &= 1-1\\
     &= 0.
\end{split}
\end{equation}
In line four, the Pauli operators $P_{\vec{a}'},P_{\vec{b}'}\in \mathcal{P}_2^{\otimes n}$ are fixed.

We now show that $O_M(U)=0$ iff $U$ is Clifford. First, let $U$ satisfy $O_M(U)=0$, such that $\abs{\OTOC(U)}=1$. In the case where any two operators $A$ and $B$ satisfy $\langle A,A\rangle=\langle B, B \rangle=1$, then $\abs{\langle A, B \rangle}=1$ only if $A=e^{i\theta}B$ for a phase $\theta$. The OTOC is
\begin{equation}
\begin{split}
	\abs{\OTOC(U)}
	&=\abs{\langle U^\dagger P_{\vec{a}} U P_{\vec{b}} U^\dagger P_{\vec{a}} U P_{\vec{b}}\rangle}\\
	&=\abs{\langle U^\dagger P_{\vec{a}} U P_{\vec{b}} U^\dagger P_{\vec{a}} U , P_{\vec{b}}\rangle},
\end{split}
\end{equation}
where $\abs{\langle U^\dagger P_{\vec{a}} U P_{\vec{b}} U^\dagger P_{\vec{a}} U , P_{\vec{b}}\rangle}=1$ only if $ U^\dagger P_{\vec{a}} U P_{\vec{b}} U^\dagger P_{\vec{a}} U = e^{i\theta}P_{\vec{b}}$. This implies that
\begin{equation}
\begin{split}
	U^\dagger P_{\vec{a}} U P_{\vec{b}}
	& = e^{i\theta}P_{\vec{b}}U^\dagger P_{\vec{a}} U.
\end{split}
\end{equation}
We square this and take the trace
\begin{equation}
\begin{split}
\tr{(U^\dagger P_{\vec{a}} U P_{\vec{b}})^2 }
	& = \tr{(e^{i\theta}P_{\vec{b}}U^\dagger P_{\vec{a}} U)^2}\\
	& = e^{2i\theta}\tr{(U^\dagger P_{\vec{a}} UP_{\vec{b}})^2}.
\end{split}
\end{equation}
This equation is satisfied if $e^{2i\theta}=1$, implying that $e^{i\theta}=\pm 1$. Therefore,
\begin{equation}
\begin{split}
	U^\dagger P_{\vec{a}} U P_{\vec{b}}
	& = \pm P_{\vec{b}}U^\dagger P_{\vec{a}} U.
\end{split}
\end{equation}
Operators $P_{\vec{b}}$ and $U^\dagger P_{\vec{a}} U $ must either commute or anti-commute for all $P_{\vec{b}}$ and $P_{\vec{a}}$. This requires that $U^\dagger P_{\vec{a}} U $ be a Pauli operator, up to a phase. To prove this, we introduce the following lemma.

\begin{lemma}\label{Lemma:PauliCommute}
For any two, unique $n$-qubit Pauli operators $P_{\vec{a}}$ and $P_{\vec{b}}$, there exists a Pauli operator $P_{\vec{c}}$ which commutes with one and anti-commutes with the other.
\end{lemma}
\begin{proof}
Since $P_{\vec{a}}$ and $P_{\vec{b}}$ are unique Pauli operators, they must differ on at least one site. Let them differ on site $i$:
\begin{equation}
\begin{split}
	P_{\vec{a}}&=P_{a_i}\otimes P_{{\vec{a}},n-1},\\
	P_{\vec{b}}&=P_{b_i}\otimes P_{{\vec{b}},n-1},
\end{split}
\end{equation}
where, for example, $P_{a_i}$ is a Pauli operator on site $i$ and $P_{{\vec{a}},n-1}$ is a Pauli operator on the remaining sites. Also, $P_{a_i}\neq P_{b_i}$. Define $P_{\vec{c}}$ as $P_{\vec{c}}=P_{c_i}\otimes I^{\otimes n-1}$. If $P_{a_i}=I$, then defining $P_{c_i}$ to be a non-trivial Pauli operator such that $P_{c_i}\neq P_{b_i}$ ensures that $P_{\vec{c}}$ commutes with $P_{\vec{a}}$ and anti-commutes with $P_{\vec{b}}$. If $P_{b_i}=I$, then defining $P_{c_i}$ to be a non-trivial Pauli operator such that $P_{c_i}\neq P_{a_i}$ ensures that $P_{\vec{c}}$ commutes with $P_{\vec{b}}$ and anti-commutes with $P_{\vec{a}}$. If neither $P_{a_i}$ nor $P_{b_i}$ is $I$, then setting $P_{c_i}=P_{a_i}$ ensures that $P_{{\vec{c}}}$ commutes with $P_{\vec{a}}$ and anti-commutes with $P_{\vec{b}}$.
\end{proof}
The operator $U^\dagger P_{\vec{a}} U$ can be expanded in the Pauli basis:
\begin{equation}
	U^\dagger P_{\vec{a}} U=\sum_{{P_{\vec{c}}}\in \mathcal{P}^{\otimes n}_2}c_{\vec{c}} P_{\vec{c}}.
\end{equation}
However, for any two Pauli operators in the expansion, by Lemma~\ref{Lemma:PauliCommute} there exists a Pauli operator $P_{\vec{b}}$ that commutes with one Pauli operator and anti-commutes with the other. This implies that $P_{\vec{b}}$ neither commutes nor anti-commutes with $U^\dagger P_{\vec{a}} U$. In order for $U^\dagger P_{\vec{a}} U$ to commute or anti-commute with all Pauli operators $P_{\vec{b}}$, it must be proportional to a single Pauli operator, $P_{\vec{c}}$:
\begin{equation}
	U^\dagger P_{\vec{a}} U=e^{i\theta_2} P_{\vec{c}}.
\end{equation}
This must be true for all $P_{\vec{a}}\in \mathcal{P}_2^{\otimes n}$, implying that $U$ is Clifford. Therefore, if $O_M(U)=0$, then $U$ is a Clifford unitary.

Now we prove the converse. If $U$ is a Clifford unitary, then
\begin{equation}
\begin{split}
O_M(U)
	&=\max_{P_{\vec{a}},P_{\vec{b}}  \in \mathcal{P}_2^{\otimes n}}\left[
	1-\abs{\OTOC(U)}\right]\\
	&=\max_{P_{\vec{a}},P_{\vec{b}}\in \mathcal{P}_2^{\otimes n}}\left[
	1-\abs{\langle U^\dagger P_{\vec{a}} U P_{\vec{b}}  U^\dagger P_{\vec{a}} U  P_{\vec{b}}\rangle}\right]\\
	&=\max_{P_{\vec{c}},P_{\vec{b}}\in \mathcal{P}_2^{\otimes n}}\left[
	1-\abs{\langle P_{\vec{c}}P_{\vec{b}}  P_{\vec{c}}  P_{\vec{b}}\rangle}\right]\\
	&=\max_{P_{\vec{c}},P_{\vec{b}}\in \mathcal{P}_2^{\otimes n}}\left[
	1-\abs{\pm\langle P_{\vec{c}}^2 P_{\vec{b}}^2\rangle}\right]\\
	&=0,
\end{split}
\end{equation}
where in line three, we use that $U^\dagger P_{\vec{a}}U=e^{i\theta_2} P_{\vec{c}}$ is a Pauli operator and we use that $U^\dagger \mathcal{P}_2^{\otimes n} U=\mathcal{P}_2^{\otimes n}$ . In line four, we use that Pauli operators either commute or anti-commute with each other. Hence, $O_M(U)=0$ iff $U$ is Clifford, making $O_M$ faithful.

We now prove invariance. Let $V$ be any unitary and let $U_1$ and $U_2$ be Clifford unitaries. Then
\begin{equation}
\begin{split}
O_M(U_1VU_2)
	&=\max_{P_{\vec{a}},P_{\vec{b}}  \in \mathcal{P}_2^{\otimes n}}\left[
	1-\abs{\OTOC(U_1VU_2)}\right]\\
	&=\max_{P_{\vec{a}},P_{\vec{b}}\in \mathcal{P}_2^{\otimes n}}\left[
	1-\abs{\langle U_2^\dagger V^\dagger  U_1^\dagger P_{\vec{a}} U_1VU_2 P_{\vec{b}} U_2^\dagger V^\dagger  U_1^\dagger P_{\vec{a}} U_1 VU_2 P_{\vec{b}}\rangle}\right]\\
	&=\max_{P_{\vec{c}},P_{\vec{d}}\in \mathcal{P}_2^{\otimes n}}\left[
	1-\abs{\langle V^\dagger  P_{\vec{c}}V P_{\vec{d}} V^\dagger  P_{\vec{c}}V P_{\vec{d}}\rangle}\right]\\
	&=O_M(V).
\end{split}
\end{equation}
In line three, we use $U_1^\dagger P_{\vec{a}} U_1\propto P_{\vec{c}}$ and $U_2 P_{\vec{b}} U_2^\dagger \propto P_{\vec{d}}$ (up to a phase). Therefore the OTOC magic is a magic monotone.

\section{OTOC magic of non-Clifford gates in the third level of the Clifford hierarchy}\label{Sec:ProofCliffordHierarchy}
If $U\in \mathcal{C}^{(3)}\backslash \mathcal{C}^{(2)}$, then
\begin{equation}
\begin{split}
1-\abs{\OTOC(U)}
	&=1-\abs{\langle U^\dagger P_{\vec{a}} U P_{\vec{b}} U^\dagger P_{\vec{a}} U P_{\vec{b}}\rangle}\\
	&=1-\abs{\langle U^\dagger P_{\vec{a}}^\dagger U P_{\vec{b}} U^\dagger P_{\vec{a}} U P_{\vec{b}}\rangle}\\
	&=1-\abs{\langle U_{\mathrm{Cl}}^\dagger P_{\vec{b}} U_{\mathrm{Cl}} P_{\vec{b}}\rangle}\\
	&=1-\abs{\langle P_{\vec{c}}P_{\vec{b}}\rangle}\\
	&=1-\delta_{P_{\vec{c}},P_{\vec{b}}}.
\end{split}
\end{equation}
In the second line, we use that $P_{\vec{a}}$ is Hermitian. In the third line, $U^\dagger P_{\vec{a}} U=e^{i\theta} U_{\mathrm{Cl}}$ is a Clifford unitary. In the fourth line, $U_{\mathrm{Cl}}^\dagger P_{\vec{b}} U_{\mathrm{Cl}}=e^{i\theta_2}P_{\vec{c}}$ is a Pauli operator. The last line indicates that $1-\abs{\OTOC(U)}\in\{0,1\}$. Therefore $\mathrm{max}_{P_{\vec{a}},P_{\vec{b}}\in \mathcal{P}_2^{\otimes n}}[1-\abs{\OTOC(U)}]\in\{0,1\}$. However, $O_M(U)=0$ iff $U$ is Clifford. Since $U$ is non-Clifford, then $O_M(U)=1$.

\section{OTOC magic of the phase gate}\label{Sec:ProofPropepsilon}
The Pauli expansion of $U_{\varepsilon}$ is
\begin{equation}
\begin{split}
    U_{\varepsilon}&=\ket{0}\bra{0}+e^{i\varepsilon}\ket{1}\bra{1}\\
    &=\tfrac{1}{2}(I+Z)+\tfrac{e^{i\varepsilon}}{2}(I-Z)\\
    &=\tfrac{1}{2}(1+e^{i\varepsilon})I+\tfrac{1}{2}(1-e^{i\varepsilon})Z\\
    &=\tfrac{e^{i\varepsilon/2}}{2}\left[\left(e^{-i\varepsilon/2}+e^{i\varepsilon/2}\right)I+\left(e^{-i\varepsilon/2}-e^{i\varepsilon/2}\right)Z\right]\\
     &=e^{i\varepsilon/2}\left[\frac{1}{2}\left(e^{i\varepsilon/2}+e^{-i\varepsilon/2}\right)I-\frac{i}{2i}\left(e^{i\varepsilon/2}-e^{-i\varepsilon/2}\right)Z\right]\\
    &=e^{i\varepsilon/2}\left[\cos(\tfrac{\varepsilon}{2})I-i\sin(\tfrac{\varepsilon}{2})Z\right].
\end{split}
\end{equation}
Defining the two single-qubit Pauli operators $P_a,P_b\in \mathcal{P}_2$, the magnitude of the OTOC is
\begin{equation}
    \begin{split}
    \abs{\OTOC(U_{\varepsilon})}
    &=\abs{\left\langle U_{\varepsilon}^\dagger P_aU_{\varepsilon} P_b U_{\varepsilon}^\dagger P_a U_{\varepsilon} P_b\right\rangle}\\
    &=\abs{\left\langle (U_{\varepsilon}^\dagger P_a U_{\varepsilon} P_b)^2\right\rangle}\\
    &=\abs{\left\langle \left(e^{-i\varepsilon/2}\left[\cos(\tfrac{\varepsilon}{2})I+i\sin(\tfrac{\varepsilon}{2})Z\right]P_a e^{i\varepsilon/2}\left[\cos(\tfrac{\varepsilon}{2})I-i\sin(\tfrac{\varepsilon}{2})Z\right] P_b\right)^2\right\rangle}\\
    &=\abs{\left\langle \left(\left[\cos(\tfrac{\varepsilon}{2})I+i\sin(\tfrac{\varepsilon}{2})Z\right]P_a \left[\cos(\tfrac{\varepsilon}{2})I-i\sin(\tfrac{\varepsilon}{2})Z\right] P_b\right)^2\right\rangle}\\
    &=\abs{\left\langle \left(\left[\cos(\tfrac{\varepsilon}{2})I+i\sin(\tfrac{\varepsilon}{2})Z\right] \left[\cos(\tfrac{\varepsilon}{2})I-if(P_a,Z)\sin(\tfrac{\varepsilon}{2})Z\right] P_aP_b\right)^2\right\rangle}\\
    &=\abs{\left\langle \left(\left[(\cos^2(\tfrac{\varepsilon}{2})+f(P_a,Z)\sin^2(\tfrac{\varepsilon}{2}))I+i\cos(\tfrac{\varepsilon}{2})\sin(\tfrac{\varepsilon}{2})(1-f(P_a,Z))Z\right]P_aP_b\right)^2\right\rangle}\\
    &=\abs{\left\langle \left(\left[AI+BZ\right]P_aP_b\right)^2\right\rangle}\\
    &=\abs{\left\langle \left[AI+BZ\right]P_aP_b\left[AI+BZ\right]P_aP_b\right\rangle}\\
    &=\abs{\left\langle \left[AI+BZ\right]\left[AI+f(P_a,Z)f(P_b,Z)BZ\right]P_aP_bP_aP_b\right\rangle}\\
    &=\abs{f(P_a,P_b)\left\langle \left[AI+BZ\right]\left[AI+f(P_a,Z)f(P_b,Z)BZ\right]P_a^2P_b^2\right\rangle}\\
    &=\abs{\left\langle (A^2+f(P_a,Z)f(P_b,Z)B^2)I+AB(1+f(P_a,Z)f(P_b,Z))Z\right\rangle}\\
    &=\abs{ A^2+f(P_a,Z)f(P_b,Z)B^2}.
    \end{split}
\end{equation}
In line five, we define $f(P_a,Z)$ such that $f(P_a,Z)=1$ or $ -1$ if $P_a$ and $Z$ commute or anti-commute, respectively. In line six, we define $A=\cos^2(\tfrac{\varepsilon}{2})+f(P_a,Z)\sin^2(\tfrac{\varepsilon}{2})$ and $B=i\cos(\tfrac{\varepsilon}{2})\sin(\tfrac{\varepsilon}{2})(1-f(P_a,Z))$.

We evaluate $\abs{\OTOC(U_{\varepsilon})}$ for all values of $f(P_a,Z)$ and $f(P_b,Z)$. For $f(P_a,Z)=1$, $A=\cos^2(\tfrac{\varepsilon}{2})+\sin^2(\tfrac{\varepsilon}{2})=1$ and $B=0$, yielding
\begin{equation}
\begin{split}
    \abs{\OTOC(U_{\varepsilon})}&=1.
\end{split}
\end{equation}
For $f(P_a,Z)=-1$, $A=\cos^2(\tfrac{\varepsilon}{2})-\sin^2(\tfrac{\varepsilon}{2})=\cos(\varepsilon)$ and $B=2i\cos(\tfrac{\varepsilon}{2})\sin(\tfrac{\varepsilon}{2})$, yielding
\begin{equation}
\begin{split}
    \abs{\OTOC(U_{\varepsilon})}
    &=\abs{ A^2-f(P_b,Z)B^2}\\
    &=\abs{ \cos^2(\varepsilon)-f(P_b,Z)(2i\cos(\tfrac{\varepsilon}{2})\sin(\tfrac{\varepsilon}{2}))^2}\\
    &=\abs{ \cos^2(\varepsilon)+f(P_b,Z)\sin^2(\varepsilon)}.
\end{split}
\end{equation}
Taking $f(P_a,Z)=-1$ and $f(P_b,Z)=1$,
\begin{equation}
\begin{split}
    \abs{\OTOC(U_{\varepsilon})}
    &= \abs{\cos^2(\varepsilon)+\sin^2(\varepsilon)}\\
    &=1.
\end{split}
\end{equation}
Taking $f(P_a,Z)=-1$ and $f(P_b,Z)=-1$,
\begin{equation}
\begin{split}
    \abs{\OTOC(U_{\varepsilon})}
    &= \abs{\cos^2(\varepsilon)-\sin^2(\varepsilon)}\\
    &=\abs{\cos(2\varepsilon)}.
\end{split}
\end{equation}
Therefore, $\abs{\OTOC(U_{\varepsilon})}\in \{1,\abs{\cos(2\varepsilon)}\}$.
The OTOC magic of $U_{\varepsilon}$ is 
\begin{equation}
\begin{split}
    O_M(U_{\varepsilon})
    &=\max_{P_a,P_b\in \mathcal{P}_2}[1-\abs{\OTOC(U_{\varepsilon})}]\\
    &=\max\{0,1-\abs{\cos(2\varepsilon)}\}\\
    &=1-\abs{\cos(2\varepsilon)}.
\end{split}
\end{equation}
Taking $\varepsilon_k=\frac{\pi}{2^{k-1}}$, it follows that
\begin{equation}
\begin{split}
    O_M(U_{\varepsilon_k})
    &=1-\abs{\cos\left(\tfrac{\pi}{2^{k-2}}\right)}.
\end{split}
\end{equation}

We also prove that $U_{\varepsilon_k}\in \mathcal{C}^{(k)}$. As shown in Ref.~\cite{PhysRevA.95.012329}, the unitary
\begin{equation}
    J_k(j')=\sum_{\substack{j=0,\\j\neq j'}}^{d-1}\ket{j}\bra{j}+e^{2\pi i/d^k}\ket{j'}\bra{j'}
\end{equation}
satisfies $J_k(j')\in\mathcal {C}_{\mathrm{diag}}^{((d-1)(k-1)+(d-1))}=\mathcal {C}_{\mathrm{diag}}^{((d-1)k)}\subset \mathcal {C}^{((d-1)k)}$, where $\mathcal {C}_{\mathrm{diag}}^{((d-1)k)}$ is the set of diagonal gates in the $(d-1)k$ level of the Clifford hierarchy. Taking $d=2$ yields $J_k(1)=U_{\varepsilon_k}\in \mathcal{C}^{(k)}$, where $\varepsilon_k=\frac{\pi}{2^{k-1}}$.

\section{Proof of Theorem~\ref{Thm:Fluctuations}}\label{Proof:Fluctuations}
\begin{theorem}
It holds that
\begin{equation}
	\av{U\sim \mathcal{E}}O_M(U)\geq 1-\delta-\abs{\av{U\sim \mathcal{E}}\OTOC(U)}.
\end{equation}
Moreover, if 
$\av{U\sim \mathcal{E}}\OTOC(U)\rightarrow 0$, then 
\begin{equation}
	\av{U\sim \mathcal{E}}O_M(U)\geq 1-\delta.
\end{equation}
\end{theorem}
\begin{proof}
\begin{equation}
\begin{split}
	\delta
	&=\sqrt{\av{U\sim \mathcal{E}}\abs{\OTOC(U)-\av{V\sim \mathcal{E}}\OTOC(V)}^2}\\
	&\geq\av{U\sim \mathcal{E}}\abs{\OTOC(U)-\av{V\sim \mathcal{E}}\OTOC(V)}\\
	&\geq\av{U\sim \mathcal{E}}\left[\abs{\OTOC(U)}-\abs{\av{V\sim \mathcal{E}}\OTOC(V)}\right]\\
	&\geq1-\av{U\sim \mathcal{E}}O_M(U)-\abs{\av{V\sim \mathcal{E}}\OTOC(V)}.
\end{split}
\end{equation}
In the last inequality, we use the bound: ${O_M(U)=\max_{P_{\vec{a}},P_{\vec{b}}\in \mathcal{P}_2^{\otimes n}}[1-\abs{\OTOC(U)}]\geq 1-\abs{\OTOC(U)}}$. This produces
\begin{equation}
\begin{split}
	\av{U\sim \mathcal{E}}O_M(U)
	&\geq1-\delta-\abs{\av{V\sim \mathcal{E}}\OTOC(V)}.
\end{split}
\end{equation}

\end{proof}

\section{Proof of Theorem~\ref{Theorem:HP}}\label{Sec:ProofThmHP}
Assume that 

\begin{equation}\label{Eq:OTOCAssumption}
	\av{P_A, P_D}\abs{\OTOC(U_{\mathrm{bh}})}=\av{P_A, P_D}\OTOC(U_{\mathrm{bh}})+\eta.
\end{equation}
The OTOC magic is bounded by the decoding fidelity:
\begin{equation}\label{Eq:MagicBoundProof}
\begin{split}
O_M(U_{\mathrm{bh}})
	&=\max_{P_{\vec{a}},P_{\vec{b}}\in \mathcal{P}_2^{\otimes n}}\left[1-\abs{\OTOC(U_{\mathrm{bh}})}\right]\\
	&\geq \av{P_A,P_D}\left[1-\abs{\OTOC(U_{\mathrm{bh}})}\right]\\
	&= 1- \av{P_A,P_D} \OTOC(U_{\mathrm{bh}})-\eta\\
	&= 1- \frac{1}{d^2_A F(U_{\mathrm{bh}})}-\eta.
\end{split}
\end{equation}
In the second line, $\OTOC(U_{\mathrm{bh}})=\langle U^\dagger P_D U P_A U^\dagger P_D U P_A \rangle $. Rewriting the inequality in terms of the decoding fidelity,
\begin{equation}
	F(U_{\mathrm{bh}})\leq \frac{1}{d_A^2(1-O_M(U_{\mathrm{bh}})-\eta)}.
\end{equation}

\section{Proof of Theorem~\ref{Theorem:HPPauliGrowth}} \label{Sec:ProofTheoremHPPAuliGrowth}
In the large $n$ limit, the average OTOC is related to the average Pauli weight via~\cite{Bu2022}
\begin{equation}
	\av{A}\av{P_A \neq I_A}\OTOC(U_{\mathrm{bh}})=1-\frac{4}{3n}W(U^\dagger_{\mathrm{bh}} P_DU_{\mathrm{bh}}).
\end{equation}
Solving for the average over the Pauli weight, 
\begin{equation}
	W(U_{\mathrm{bh}}^\dagger P_DU_{\mathrm{bh}})=\frac{3n}{4}\left[1-\av{A}\av{P_A \neq I_A}\OTOC(U_{\mathrm{bh}})\right].
\end{equation}
The Pauli growth can be bounded via the average OTOC:
\begin{equation}
\begin{split}
	G(U_{\mathrm{bh}})
	&= \max_{\substack{O:  \norm{O}_2=1, W(O)=1,\\\tr{O}=0}}\left[W(U_{\mathrm{bh}}^\dagger O U_{\mathrm{bh}})-1\right]\\
	&\geq  \av{P_D} W(U_{\mathrm{bh}}^\dagger P_D U_{\mathrm{bh}})-1\\
	&= 
	\frac{3n}{4}\left[1-\av{A}\av{P_A \neq I_A}\av{P_D}\OTOC(U_{\mathrm{bh}})\right]-1.
\end{split}
\end{equation}
Rewriting this in terms of the average OTOC,
\begin{equation}\label{Eq:PauliWeightIneq}
\begin{split}
\av{A}\av{P_A\neq I_A}\av{P_D}\OTOC(U_{\mathrm{bh}})\geq 1-\frac{4}{3n}\left(G(U_{\mathrm{bh}})+1\right).
\end{split}
\end{equation}
We next rewrite $\av{P_A \neq I_A}\OTOC(U_{\mathrm{bh}})$ in terms of $\av{P_A}\OTOC(U_{\mathrm{bh}})$:
\begin{equation}
\begin{split}
\av{P_A \neq I_A}\OTOC(U_{\mathrm{bh}})
	&=\frac{1}{d_A^2-1}\sum_{P_A\neq I_A}\OTOC(U_{\mathrm{bh}})\\
	&=\frac{1}{d_A^2-1}\left[\sum_{P_A\neq I_A}\OTOC(U_{\mathrm{bh}})+\sum_{P_A= I_A}(\OTOC(U_{\mathrm{bh}})-\OTOC(U_{\mathrm{bh}}))\right]\\
	&=\frac{1}{d_A^2-1}\left[\sum_{P_A }\OTOC(U_{\mathrm{bh}})-1\right]\\
	&=\frac{1}{d_A^2-1}\left[\frac{d_A^2}{d_A^2}\sum_{P_A }\OTOC(U_{\mathrm{bh}})-1\right]\\
	&=\frac{1}{d_A^2-1}\left[d_A^2\av{P_A }\OTOC(U_{\mathrm{bh}})-1\right].
\end{split}
\end{equation}
Plugging into Ineq.~\eqref{Eq:PauliWeightIneq},
\begin{equation}
\begin{split}
\frac{1}{d_A^2-1}\left[d_A^2\av{A}\av{P_A,P_D }\OTOC(U_{\mathrm{bh}})-1\right]\geq 1-\frac{4}{3n}\left(G(U_{\mathrm{bh}})+1\right).
\end{split}
\end{equation}
Solving for $d_A^2\av{A}\av{P_A,P_D }\OTOC(U_{\mathrm{bh}})$, 
\begin{equation}
\begin{split}
d_A^2\av{A}\av{P_A,P_D }\OTOC(U_{\mathrm{bh}})\geq (d_A^2-1)\left[1-\frac{4}{3n}\left(G(U_{\mathrm{bh}})+1\right)\right]+1.
\end{split}
\end{equation}
Using the relation $\av{P_A, P_D}\OTOC(U_{\mathrm{bh}})=\frac{1}{d_A^2 F(U_{bh})}$,
\begin{equation}
\begin{split}
\av{A}\frac{1}{F(U_{\mathrm{bh}})}\geq (d_A^2-1)\left[1-\frac{4}{3n}\left(G(U_{\mathrm{bh}})+1\right)\right]+1.
\end{split}
\end{equation}

\end{appendix}

\end{document}